\documentclass[11pt, oneside]{article}   	% use "amsart" instead of "article" for AMSLaTeX format
\usepackage{geometry}                		% See geometry.pdf to learn the layout options. There are lots.
\usepackage{graphicx}				% Use pdf, png, jpg, or epsÂ§ with pdflatex; use eps in DVI mode
\usepackage{amssymb}
\usepackage{amsfonts,latexsym,amsthm,amssymb,amsmath,amscd,euscript}
\usepackage{enumitem}
\usepackage{bbm}
\usepackage{framed}
\usepackage{fullpage}
\usepackage{mathrsfs}
\usepackage{hyperref}
\hypersetup{
   colorlinks=true,
   linkcolor=blue,
   filecolor=blue,
   citecolor = black,      
   urlcolor=cyan,
   pdfpagemode=UseOutlines
 }
\usepackage{accents}
\usepackage{setspace}
\hypersetup{colorlinks=true,citecolor=blue,urlcolor=black,linkbordercolor={1 0 0}}
\usepackage{tikz-cd}

\makeatletter
\newtheorem*{rep@theorem}{\rep@title}
\newcommand{\newreptheorem}[2]{%
\newenvironment{rep#1}[1]{%
 \def\rep@title{#2 \ref{##1}}%
 \begin{rep@theorem}}%
 {\end{rep@theorem}}}
\makeatother

\theoremstyle{plain}
\newtheorem{theorem}{Theorem}
\newreptheorem{theorem}{Theorem}
\newtheorem{lemma}[theorem]{Lemma}
\newtheorem{corollary}[theorem]{Corollary}

\newtheorem{proposition}[theorem]{Proposition}

\newtheorem{claim}[theorem]{Claim}
\theoremstyle{definition}
\newtheorem{definition}{Definition}
\newtheorem{defn}[definition]{Definition}

\newtheorem{remark}[definition]{Remark}
\newtheorem{question}[definition]{Question}
\newtheorem{problem}[definition]{Problem}

\numberwithin{theorem}{section} % important bit
\numberwithin{definition}{section}

\newcommand{\nc}{\newcommand}
\nc{\DMO}{\DeclareMathOperator}
\newcount\Comments  % 0 suppresses notes to selves in text
\nc{\todo}[1]{\ifnum\Comments=1 {\color{red}  [TODO: #1]}\fi}

\nc{\BR}{\mathbb R}
\nc{\BA}{\mathbb{A}}
\nc{\BP}{\mathbb{P}}
\nc{\BC}{\mathbb C}
\nc{\MO}{\mathcal O}
\nc{\MU}{\mathcal{U}}
\nc{\ME}{\mathcal{E}}
\nc{\MN}{\mathcal{N}}
\nc{\MK}{\mathcal{K}}
\nc{\MS}{\mathcal{S}}
\nc{\MT}{\mathcal{T}}
\nc{\BF}{\mathbb F}
\nc{\BQ}{\mathbb Q}
\nc{\MX}{\mathcal{X}}
\nc{\MA}{\mathcal{A}}
\nc{\MD}{\mathcal{D}}
\nc{\MB}{\mathcal{B}}
\nc{\MZ}{\mathcal{Z}}
\nc{\MY}{\mathcal{Y}}
\nc{\BZ}{\mathbb Z}
\nc{\BN}{\mathbb N}
\nc{\ep}{\epsilon}
\nc{\BH}{\mathbb H}
\nc{\BG}{\mathbb{G}}
\nc{\D}{\Delta}
\nc{\MF}{\mathcal{F}}
\nc{\One}{\mathbbm{1}}

\nc{\SP}{\mathsf P}
\nc{\SQ}{\mathsf Q}

\nc{\DO}{\accentset{\circ}{\D}}
\nc{\mf}{\mathfrak}
\nc{\mfp}{\mathfrak{p}}
\nc{\mfq}{\mf{q}}
\nc{\Sp}{\mbox{Spec}}
\nc{\Spm}{\mbox{Specm}}
\nc{\hookuparrow}{\mathrel{\rotatebox[origin=c]{90}{$\hookrightarrow$}}}
\nc{\hookdownarrow}{\mathrel{\rotatebox[origin=c]{-90}{$\hookrightarrow$}}}
\nc{\hra}{\hookrightarrow}
\nc{\tra}{\twoheadrightarrow}
\nc{\sgn}{{\rm sgn}}
\nc{\aut}{{\rm Aut}}
\nc{\Hom}{{\rm Hom}}
\nc{\img}{{\rm Im}}
\DMO{\id}{Id}
\DMO{\supp}{supp}
\DMO{\KL}{KL}
\DMO{\BSS}{BSS}
\DMO{\BES}{BES}
\DMO{\BGS}{BGS}
\DMO{\poly}{poly}
\nc{\indep}{\perp}

\nc{\p}{\mathbb{P}}

\nc{\E}{\mathbb{E}}
\nc{\ra}{\rightarrow}

\nc{\BPPdt}{\textsc{BPP}\textsuperscript{dt}}
\nc{\BPPcc}{\textsc{BPP}\textsuperscript{cc}}
\nc{\BQPdt}{\textsc{BQP}\textsuperscript{dt}}
\nc{\BQPcc}{\textsc{BQP}\textsuperscript{cc}}
\nc{\Pdt}{\textsc{P}\textsuperscript{dt}}
\nc{\Pcc}{\textsc{P}\textsuperscript{cc}}
\nc{\AND}{\textsc{AND}}
\nc{\OR}{\textsc{OR}}
\nc{\CP}{\textsc{CP}}
\nc{\QCP}{\textsc{QCP}}
\nc{\BCP}{\textsc{BCP}}
\nc{\Ind}{\textsc{Ind}}
\nc{\PromiseBQP}{\textsc{PromiseBQP}}
\nc{\rag}{\rangle}
\nc{\lag}{\langle}
\nc{\Forrelation}{\textsc{Forrelation}}
\nc{\Qsim}{\textsc{Qsim}}
\nc{\IC}{\textsc{IC}}
\DMO{\ext}{ext}
\DMO{\itrn}{int}
\nc{\ICint}{\IC^{\itrn}}
\nc{\ICext}{\IC^{\ext}}
\nc{\Iint}{I^{\itrn}}
\nc{\Iext}{I^{\ext}}
\DMO{\pubnoss}{pub}
\nc{\pub}{^{\pubnoss}}
\DMO{\prinoss}{pri}
\nc{\pri}{^{\prinoss}}

\DMO{\mix}{mix}
\DMO{\midd}{mid}
\nc{\mumix}{\mu_{\mix}}
\nc{\mumid}{\mu^{\midd}}
\DMO{\disj}{Disj}
\DMO{\YY}{Y}
\DMO{\NN}{N}
\DMO{\Mix}{Mix}
\nc{\disjy}{\disj^{\YY}}
\nc{\disjn}{\disj^{\NN}}
\nc{\dypv}{{D^{\YY}_{\mathrm{PV}}}}
\nc{\dnpv}{{D^{\NN}_{\mathrm{PV}}}}
\nc{\dmpv}{{D^{\Mix}_{\mathrm{PV}}}}
\nc{\Dmpv}{{\mathcal{D}^{\Mix}_{\mathrm{PV}}}}
\nc{\Dmpvp}{{\mathcal{D}^{\Mix+}_{\mathrm{PV}}}}

\DMO{\Qcoh}{Qcoh}
\DMO{\coker}{coker}
\DMO{\polylog}{polylog}
\DMO{\Gr}{Gr}

\DMO{\env}{env}
\nc{\A}{\texttt{A}}
\nc{\B}{\texttt{B}}
\nc{\pubc}{\texttt{Pub}}
\nc{\Crg}{\texttt{cr}}
\nc{\Skg}{\texttt{sk}}
\nc{\crg}{^{\Crg}}
\nc{\skg}{^{\Skg}}
\nc{\RA}{R_\A}
\nc{\RB}{R_\B}
\nc{\QA}{Q_\A}
\nc{\QB}{Q_\B}
\nc{\KA}{K_{\A}}
\nc{\KB}{K_\B}
\nc{\Rpub}{R_\pubc}
\DMO{\CC}{CC}
\DMO{\PP}{P}
\DMO{\NC}{NC}
\DMO{\DD}{D}
\nc{\pubstraight}{\text{pub}}
\DMO{\RR}{R}
\nc{\RRpub}{\RR^{\pubstraight}}
\DMO{\QQ}{Q}
\nc{\dd}{-}
\nc{\amtz}{\texttt{am}}
\nc{\FC}{\mathscr{C}}
\nc{\FCamsk}{\FC^{\amtz\texttt{-}\Skg}}
\nc{\FCamcr}{\FC^{\amtz\texttt{-}\Crg}}
\nc{\tilFCamcr}{\tilde{\FC}^{\amtz\texttt{-}\Crg}}
\nc{\FI}{\mathscr{I}}
\nc{\FF}{\mathscr{F}}
\DMO{\ttd}{d}
\nc{\MTd}{\MT^{\ttd}}
\nc{\ir}{I}
\nc{\ip}{i}
\nc{\jr}{J}
\nc{\jp}{j}
\nc{\kp}{k}
\nc{\pir}{\Sigma}
\nc{\pip}{\sigma}
\nc{\mr}{\Pi}
\DMO{\MAX}{m}
\nc{\rhom}{\rho_{\MAX}}

\title{Round Complexity of Common Randomness Generation: The Amortized Setting}
\author{Noah Golowich\thanks{Massachusetts Institute of Technology, EECS, {\tt nzg@mit.edu}. Currently supported by an MIT Akamai Fellowship and a Fannie \& John Hertz Foundation Fellowship. This work was performed while the author was a student at Harvard University.} \and 
	Madhu Sudan\thanks{Harvard John A. Paulson School of Engineering and Applied Sciences, Harvard University, 33 Oxford Street, Cambridge, MA 02138, USA. {\tt madhu@cs.harvard.edu}. Work supported in part by a Simons Investigator Award and NSF Award CCF 1715187.} 
}
\date{September 1, 2019}							% Activate to display a given date or no date

\begin{document}
\maketitle

\begin{abstract}
	In this work we study the effect of rounds of interaction on the common randomness generation (CRG) problem. In the CRG problem, two parties, Alice and Bob, receive samples $X_i$ and $Y_i$, respectively, where $(X_i, Y_i)$ are drawn jointly from a {\it source distribution} $\mu$. The two parties wish to agree on a common random key consisting of many bits of randomness, by exchanging messages that depend on each party's respective input and the previous messages. In this work we study the {\em amortized} version of the problem,\todo{may want to change the next part of this sentence since we don't study the KBIB} i.e., the number of bits of communication needed per random bit output by Alice and Bob, in the limit as the number of bits generated tends to infinity. The amortized version of the CRG problem has been extensively studied in the information theory literature, though very little was known about the effect of interaction on this problem. Recently Bafna et al. (SODA 2019) considered the {\em non-amortized} version of the problem (so here the goal of the interaction is to generate a fixed number of random bits): they gave a family of sources
	$\mu_{r,n}$ parameterized by $r,n \in \BN$, such that with $r+2$ rounds of communication one can generate $n$ bits of common randomness with this source with $O(r\log n)$ communication, whereas with roughly $r/2$ rounds the communication complexity is $\Omega(n/\poly\log n)$. Note in particular that their source is designed with the target number of bits in mind and hence the result does not apply to the amortized setting.
	
	In this work we strengthen the work of Bafna et al. in two ways: First we show that the results extend to the classical amortized setting. We also reduce the gap between the round complexity in the upper and lower bounds to an additive constant.
	Specifically we show that for every pair $r,n \in \BN$ the (amortized) communication complexity to generate $\Omega(n)$ bits of common randomness from the source $\mu_{r,n}$ using $r+2$ rounds of communication is $O(r \log n)$ whereas the amortized communication required to generate the same amount of randomness from $r$ rounds is $\Omega(\sqrt{n})$.
	% Our sources are the same ones used by Bafna et al.
    Our techniques exploit known connections between information complexity and CRG, and the main novelty is our ability to analyze the information complexity of protocols getting inputs from the source $\mu_{r,n}$.
	
\end{abstract}
	\iffalse \todo{Need to correct the following sentences.}
	Specifically we show that for every parameter $r \in \BN$ and gap $g \in \BR^+$, there exists $\epsilon > 0$ and a source $\mu'_{r,g}$ such that the amortized communication complexity of CRG from $\mu'_{r,g}$ is at least $g\cdot \epsilon$ bits per output bit when restricted to $r$ round protocols, whereas there exists an $(r+2)$-round CRG protocol with communication complexity at most $\epsilon$ per bit of randomness generated. Our sources are the same ones used by Bafna et al. Our techniques exploit the known connections between information complexity and CRG, and the main novelty is our ability to analyze the information complexity of protocols getting inputs from our source $\mu'$. \fi

\newpage \pagenumbering{arabic}
\section{Introduction}
In this paper we study the problem of {\it common randomness generation} (CRG) and the companion problem of {\it secret key generation} (SKG). In each of these problems, there are two parties Alice and Bob, who are given several samples of correlated randomness: Alice is given random variables $X_1, X_2, \ldots$, and Bob is given random variables $Y_1, Y_2, \ldots$, where the pairs $(X_i, Y_i)$ are distributed i.i.d.~according to some distribution $\mu$. In the CRG problem (Figure \ref{fig:crg}), the goal of Alice and Bob is to agree, with high probability, on some shared key $K$ of high entropy by communicating as little as possible. In the SKG problem, they have the additional secrecy requirement that an eavesdropper Eve who observes their transcript of communication cannot determine much information on $K$.

The problems of CRG and SKG were introduced independently by Maurer \cite{maurer_perfect_1991,maurer_conditionally-perfect_1992,maurer1993secret} and by Ahlswede and Csisz{\'a}r \cite{ahlswede1993common,ahlswede1998common}. An important motivation for their work was from cryptography, where the posession of a shared secret key allows parties to securely transmit information using a private-key cryptosystem. Rather than generating private keys based on computational hardness assumptions, as in \cite{diffie_new_1976,rivest_method_1978}, these works suggested the study of secret key generation from an information-theoretic viewpoint, under information-theoretic assumptions such as access to a correlated source. Subsequently  techniques similar to those developed in \cite{maurer1993secret,ahlswede1993common}, such as privacy amplification, have been used in work on quantum key agreement \cite{bennett_experimental_1992,hughes_quantum_1995}. 
Shared common randomness, and the generation thereof, has also found additional applications in identification capacity \cite{ahlswede_identification_1989-1,ahlswede_identification_1989}, communication complexity \cite{canonne2017communication,ghazi_communication_2015,ghazi_power_2017,bavarian2014role}, locality-sensitive hashing, \cite{ghazi2018resource} and coding theory \cite{blackwell_capacities_1960,csiszar_capacity_1991}.

\begin{figure}
	\centering
	\includegraphics[scale=0.3]{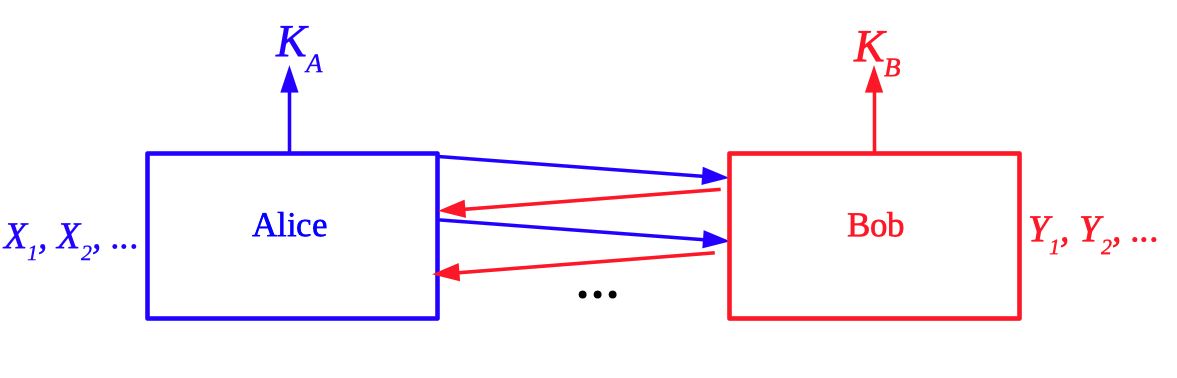}
	\caption{Common randomness generation.}
	\label{fig:crg}
\end{figure}

The initial introduction of CRG and SKG by Maurer, Ahlswede, and Csisz\'{a}r was in the {\it amortized setting}, which has since been studied in many works (such as \cite{csiszar2000common,csiszar2004secrecy,zhao2011efficiency,tyagi2013common,liu2015secret,liurate,liu2017secret,ye_information_2005,gohari_information-theoretic_2010,gohari_information-theoretic_2010-1}). In this setting, given a source of correlation $\mu$, the goal is to characterize the ``achievability region'', i.e.,  those pairs $(C,L)$ of non-negative real numbers, such that if Alice and Bob receive $N$ i.i.d.~copies of the inputs $(X,Y) \sim \mu$, by communicating roughly $C\cdot N$ bits, they can generate nearly $L \cdot N$ bits of common randomness (or secret key) with probability approaching 1 as $N \ra \infty$; a formal definition is presented in Definitions \ref{def:a_crg} and \ref{def:a_skg}.

In the theoretical computer science community the {\it non-amortized setting} of CRG has also been extensively studied. In this setting, Alice and Bob still receive some number $N$ of i.i.d.~samples from the source $\mu$, but the communication and key length do not have to grow linearly with the number of samples, and the probability of agreeing on a key need not approach 1. This problem was first studied in its zero-communication variant, where it is also known as {\it non-interactive correlation distillation}, and in the setting where Alice and Bob wish only to agree on a single bit, by Gacs and K\"{o}rner \cite{gacs1973common} and Witsenhausen \cite{witsenhausen1975sequences}, as well as later works \cite{mossel2005coin,mossel2006non,yang_impossibility_2007}. Bogdanov and Mossel \cite{bogdanov2011extracting} and Chan et al.~\cite{chan2014extracting} study the version where Alice and Bob wish to agree on many bits, again in the zero-communication setting. Finally, several more recent works \cite{canonne2017communication,guruswami2016tight,ghazi2018resource} have studied the non-amortized version of CRG where communication is allowed. These latter works generally study relatively simple sources, such as the bivariate Gaussian source (BGS) and the binary symmetric source (BSS). \footnote{For a parameter $p \in [0,1]$, the {\it binary symmetric source} $\BSS_p$ is the one where $X,Y$ are bits that are each uniformly distributed in $\{0,1\}$ and such that $\p_{\BSS_p}[X\neq Y] = p$. For $\rho \in [-1,1]$, the {\it bivariate Gaussian source} $\BGS_\rho$ is the one where $X,Y \in \BR$ are individually distributed as standard Gaussians and that $\E_{\BGS_\rho}[XY] = \rho$.} % Even for such simple sources, the achievable rates are quite difficult to determine.

\subsection{Overview of main results: does interaction help?}
\label{sec:int_help}
Despite the large amount of work on CRG and SKG in the last several decades, until recently, very little was known about the role of interaction in these problems. While initial work in the area \cite{ahlswede1993common,ahlswede1998common} studied only 1-round and 2-round protocols, recent works \cite{liu2015secret,liurate,liu2017secret} have generalized those initial results to multi-round protocols; however, until our work, it was not known in the amortized setting if increasing the number of rounds of some $r$-round protocol can actually allow the parties to communicate less (and generate random keys of the same length).%  while communicating less. %    In this paper our focus is on the relationship between the number of rounds and the total amount of communication needed in protocols for CRG and SKG.

This question of whether Alice and Bob can reduce the communication cost of their protocol at the expense of increasing the number of rounds is central to our work. Curiously, for the amortized setting, the answer to this question is negative in several cases: for instance, when $(X,Y)$ is distributed according to the binary symmetric source (BSS) or the bivariate Gaussian source (BGS), Liu et al.~\cite{liu2017secret} and Tyagi \cite{tyagi2013common} showed that increasing the number of rounds does not help to reduce communication cost. In terms of separation results, Tyagi \cite{tyagi2013common} presented a source on a ternary alphabet for which a 1-round protocol has smaller communication cost than any 2-round protocol by a constant factor, and this is the only known round-based separation in the amortized setting. 

Orlitsky \cite{orlitsky_worst-case_1990,orlitsky_worst-case_1991} studied a slightly different version of CRG in which the key $K$ is required to be equal to Alice's input $X$; thus the problem becomes that of Bob learning Alice's input. Orlitsky showed (in the non-amortized case) that 2-round protocols can require exponentially less communication than 1-round protocols. However, for any $r > 2$, he showed that $r$-round protocols can save on communication cost over 2-round protocols by at most a factor of 4. This version of the problem was also studied in the amortized case by Ma and Ishwar \cite{ma_distributed_2008}, who showed that interaction does not help at all; in fact, the 1-round protocol that achieves minimum communication cost is simply given by Slepian-Wolf coding \cite{cover2012elements}.

The most relevant work is that of Bafna et al.~\cite{bafna_communication-rounds_2018}, who showed the following in the non-amortized setting (see \cite[Theorems 1.1 \& 1.2]{bafna_communication-rounds_2018}): for any fixed $r$, there are sufficiently large $n$ such that for some source $\mu = \mu_{r,n}$, we have:
\begin{itemize}
	\item Alice and Bob can generate secret keys of length $n$ with $r+2$ rounds of communication and $O(\log n)$ communication cost;
	\item When restricted to $\lfloor (r+1)/2 \rfloor$ rounds, any protocol which generates common random keys of length $n$ must have communication cost $n/\log^{\omega(1)} n$.
\end{itemize}
Notice that the above result is not tight in the dependence on the number of rounds; a tight result (up to polylogarithmic factors) would state that any $(r+1)$-round protocol must have communication cost $n/\log^{\omega(1)}n$. More significantly, the result does not establish any separation in the amortized setting, which is the main target of this paper.

\paragraph{Our results} 
Our main result is an extension of the above results of Bafna et al.~\cite{bafna_communication-rounds_2018} to the amortized setting. Along the way we also get a nearly tight dependence on the number of rounds (losing a  quadratic factor in communication cost and a single additional round of communication). In particular, we show:
\begin{itemize}
	\item For the source $\mu = \mu_{r,n}$ mentioned above, any protocol with at most $r$ rounds and which generates common random keys of length $n$ must have communication cost at least $\sqrt n / \log^{\omega(1)}n$. (See Theorem~\ref{thm:nam_formal} for a formal statement.)
	\item Moreover, an identical rounds-communication tradeoff holds for the amortized case. (See Theorem~\ref{thm:am_formal}.)
\end{itemize}

We emphasize that the second result above gives the first rounds-communication tradeoff for the amortized case (apart from the constant-factor separation between 1-round and 2-round protocols given by Tyagi \cite{tyagi2013common}). 
\iffalse{In particular our result implies that, given an parameter $r \in \BN$ and $g \in \BR^+$, by using an appropriate choice of $n$, we can get a positive $\epsilon$ such that $r+2$ round protocols achieve the point $(\epsilon \cdot L,L)$ on ``achievability region'' for $L = n$, whereas no $r$-round protocol achieves the parameters $(g\cdot \epsilon\cdot L', L')$ for any positive $L'$.}\fi 

\paragraph{Technical Challenge} At a very high level the source in \cite{bafna_communication-rounds_2018} is built around the concept of ``pointer-chasing problems'' that are well-known to lead to separations in round-complexity~\cite{NW,duris_lower_1984,papadimitriou_communication_1982}. The main contribution in their work is to show how the hardness of pointer chasing (or a variation they consider) translates to the hardness of generating common randomness in their source. 

Getting an amortized lower bound turns out to be significantly more challenging. For one thing we can no longer build a source that is crafted around a targeted length of the common random string. Indeed this ability allows Bafna et al.~\cite{bafna_communication-rounds_2018} to focus on the case where the two players get a single copy of the randomness $(X,Y) \sim \mu$, and the core of their negative result is showing that $r/2$ rounds of communication are insufficient to generate any non-trivial randomness from this single copy (``non-trivial'' meaning more than the number of bits communicated). 
In the amortized case such results are not possible: if there is a protocol with small communication and many rounds getting some amount of randomness, then we can simulate the protocol with large communication in two rounds, and then (here using the ability to amortize) we can scale back the communuication and generate proportionately less, but non-trivial amounts of randomness. Thus no matter how small the amortized communication budget is, it is always possible to get some non-trivial amounts of randomness.
\iffalse{In our case, due to the amortized nature of the question and the ability to compress any communication to two rounds, even two rounds give the ability to generate a non-trivial amount of common randomness \todo{(albeit with large communication?)}. }\fi 
So our lower bounds really need to address a ``direct product'' version of the pointer chasing question.

Indeed, the idea of our proof is to ``reduce to the non-amortized case'' by using similar types of techniques that have been applied to show direct sum and direct product results for the communication complexity of functions  \cite{chakrabarti2001informational,jain_direct_2003,harsha_communication_2007,barak2013compress,jain2012direct,braverman2011information,braverman2013direct}. However, the task of CRG is ``more flexible'' than that of computing a function as there is no prescribed output for given inputs, so implementing this reduction is nontrivial. Roughly, our results have to analyze notions such as the internal and external information complexity of {\em all} bounded round protocols (and show that these are close) whereas most of the previous use in communication complexity lower bounds only needed to work with protocols that computed a specific function. We go into further details on this in Section~\ref{sec:proof_overview} after we get more specific about the sources we consider and the kind of results we seek.

\paragraph{Organization of this paper} In Section \ref{sec:background} we formally introduce the problems of CRG and SKG (in both the non-amortized and amortized settings) and state our main results. Section \ref{sec:nam} presents the proof of our main results in the non-amortized setting, and Sections \ref{sec:am} and \ref{sec:converse_proof} present the proof of our main results in the amortized setting. Section \ref{sec:it_lemmas} collects several basic information theoretical lemmas used throughout the paper.

\section{Background and Overview of Main Results}
\label{sec:background}
\subsection{Notation}
We first describe some of the basic notational conventions we use throughout the paper. We use capital script font, such as $\MS, \MX, \MY$, to denote sets, and capital letters, such as $X,Y,Z$, to denote random variables. % We will occasionally have sets that are random variables, and in this case, we will use capital (non-script) leters.
We typically use the letters $\mu, \nu, D$ to denote distributions.  $\MS_n$ denotes the set of all permutations on $[n]$.

\paragraph{Basic probability} If $\ME \subset \MX$ is some event, then we will write $\One[X \in \ME]$ to denote the random variable that is 1 if $X \in \ME$, and 0 otherwise. We will slightly abuse notation, e.g., if $(X,Y) \sim \nu$ then $\One[X=Y]$ is 1 when $X=Y$ and 0 otherwise. If $f : \MX \ra \BR$, then $\E_{\mu}[f(X)]$ denotes the expectation of $f(X)$ when $X$ is distributed according to $\mu$. For $\ME \subset \MX$, $\p_\mu[\ME] := \E_{\mu}[\One[X \in \ME]]$ is the probability that $X \in \ME$ when $X \sim \mu$. We will omit the subscript $\mu$ if the distribution is obvious. This notation extends naturally to conditional expectations.

\paragraph{Total variation distance \& KL divergence} For random variables $X, Y$ distributed according to $\mu, \nu$, respectively, on a finite set $\MX$, $\Delta(\mu, \nu) := \frac 12 \sum_{x \in \MX} | \p_\mu[X = x] - \p_{\nu}[Y = x] |$ denotes the {\it total variational distance} between $X$ and $Y$. For distributions $\mu$ and $\nu$ supported on a set $\MX$, the {\it KL divergence} between $\mu, \nu$, denoted $\KL(\mu || \nu)$, is given by, for $X \sim \mu, Y \sim \nu$, $\KL(\mu || \nu) := \sum_{x \in \MX} \p[X=x] \cdot \log \left(\frac{\p[X=x]}{\p[Y=x]} \right)$. We will often abuse notation when denoting KL divergences or total variation distances: for $X \sim \mu, Y \sim \nu$ supported on a set $\MX$, we will write $\Delta(X, Y) = \Delta(\mu, \nu)$ and $\KL(X || Y) = \KL(\mu || \nu)$.

\paragraph{Information theory}  If $X \sim \mu$, then the {\it entropy} of $X$ is given by $H_\mu(X)=H(X) = \\ \E_{x \sim \mu} [\log(1/\p_\mu[X=x])]$. Now suppose $(X,Y)$ are random variables with $X \in \MX, Y \in \MY$ jointly distributed according to some distribution $\nu$. Letting $X_y$ denote the random variable distributed as $X$, conditioned on $Y=y$, then $H(X | Y=y) := H(X_y)$. Then the {\it conditional entropy} $H_\mu(X|Y) = H(X|Y)$ is given by $H(X|Y) =: \E_{y \sim \nu}[H(X | Y=y)]$. The {\it mutual information} is given by $I_\mu(X;Y) = I(X;Y) := H(X) - H(X|Y)$; it is well-known that $I(X;Y) = H(Y) - H(Y|X)$. If $(X,Y,Z)$ are jointly distributed according to some distribution, then the {\it conditional mutual information} $I(X;Y | Z)$ is given by $I(X; Y|Z) := H(X|Z) - H(X | Y,Z)$.

\paragraph{Multiple random variables}  For random variables $(X,Y) \sim \mu$ distributed jointly, we will often use $XY \in \MX \times \MY$ to denote the pair. The {\it marginals} $X \sim \mu_X, Y \sim \mu_Y$ are the distributions on $\MX$ and $\MY$, respectively, given by $\p_{X \sim \mu_X}[X = x] := \p_{XY \sim \mu}[X=x]$, and similarly for $\mu_Y$. Then $X \otimes Y \in \MX \times \MY$ denotes the random variable distributed according to the product of the marginals $\mu_X \otimes \mu_Y$. % It is well known that for $(X,Y) \sim \mu$, we have $I(X;Y) = \KL(\mu || \mu_X \otimes \mu_Y)$.
  For a sequence of random variables $X_1, X_2, \ldots, X_i, \ldots$, for any $j \geq 1$, we let $X^j$ denote the tuple $(X_1, \ldots, X_j)$, and for $1 \leq j \leq j'$, let $X_j^{j'}$ denote the tuple $(X_j, X_{j+1}, \ldots, X_{j'})$. Two common usages of this notation are as follows: (1) for $N \in \BN$, and a distribution $Z \sim \mu$, the random variable distributed according to $N$ i.i.d.~copies of $\mu$ is denoted as $Z^N = (Z_1, \ldots, Z_N) \sim \mu^{\otimes N}$; (2) if $\Pi_1, \ldots, \Pi_t$ denote the first $t$ messages in a communication protocol (see Section \ref{sec:protocols}), then $\Pi^t = (\Pi_1, \Pi_2, \ldots, \Pi_t)$. % Second,

      \subsection{Communication protocols}
      \label{sec:protocols}
  We follow the standard setup of interactive communication protocols \cite{yao_complexity_1979}, and mostly follow the notational conventions of \cite{barak2013compress,braverman2011information}. There are finite sets $\MX, \MY$, and parties Alice and Bob, who receive inputs $X \in \MX$, $Y \in \MY$, respectively. Depending on the setting, Alice and Bob may additionally have access to private coins $\RA, \RB$, respectively, and public coins $\Rpub$. Formally, $\RA, \RB, \Rpub$ may be interpreted as infinite strings of independently and uniformly distributed random bits. % Alice can see $\RA, \Rpub$ (if they are available), while Bob can see $\RB, \Rpub$ (if they are available).
  % A protocol in which Alice and Bob have access to the public coins $\Rpub$ is known as a {\it public-coin} protocol, and a protocol in which Alice and Bob have access to the private coins $\RA, \RB$, respectively, but not to $\Rpub$, is known as a {\it private-coin} protocol. Finally, a protocol in which Alice and Bob do not have access to any of $\RA, \RB, \Rpub$ is known as a {\it deterministic protocol}.

  An {\it interactive $r$-round protocol} $\Pi$ consists of a sequence of $r$ {\it messages}, $\Pi_1, \ldots, \Pi_r \in \{0,1\}^*$ that Alice and Bob alternatively send to each other, with Alice sending the first message $\Pi_1$. The messages $\Pi_1, \ldots, \Pi_r$ are also referred to as the {\it rounds} of the protocol, and each message is a deterministic function of the previous messages, one party's input, and any randomness (public and/or private) available to that party. %; in other words, each message is a randomized function of the previous messages and one party's input.
  For $1 \leq t \leq r$ with $t$ odd, we will write Alice's message $\Pi_t$ as $\Pi_t = \Pi_t(X, \RA, \Rpub, \Pi^{t-1})$ if the protocol can use public and private coins (with obvious modifications if public and/or private coins are not available), and for $t$ even, Bob's message $\Pi_t$ as $\Pi_t = \Pi_t(Y, \RB, \Rpub, \Pi^{t-1})$.\footnote{It is required that for each $t$ and each instiantion of $\Pi^{t-1}$, the set of possible values of $\Pi_t$ (over all possible instantiaions of $X,Y, \RA, \RB, \Rpub$) must be prefix-free. This technical detail, which is introduced so that each party knows when to ``start speaking'' when the other finishes, will not be important for us.} % (Recall that $\Pi^{t-1} = (\Pi_1, \ldots, \Pi_{t-1})$.)
  The {\it communication cost} of $\Pi$, denoted by $\CC(\Pi)$, is the maximum of $\sum_{t=1}^r | \Pi_t|$, taken over all inputs $X \in \MX, Y \in \MY$, and all settings of the random coins $\RA, \RB, \Rpub$ (if applicable). The tuple consisting of all the messages, i.e., $\Pi^r = (\Pi_1, \ldots, \Pi_r)$, is referred to as the {\it transcript} of the protocol $\Pi$.

  \subsection{Rate regions for amortized CRG \& SKG}
Recall that in amortized CRG, Alice and Bob receive some large number $N$ of copies $(X,Y)$ from the source, are allowed to communicate some number of bits that grows linearly with $N$, and must agree upon a key whose entropy grows linearly with $N$ with probability tending to 1 as $N \ra \infty$. The word ``amortized'' refers to the fact that the communication and key entropy both grow linearly with $N$. The parties may use private but not public coins (as with access to public randomness, there would be no need to generate a shared random string). Definition \ref{def:a_crg} below follows the exposition of Liu et al.~\cite{liu2017secret}.
   \begin{defn}[Amortized common randomness generation (CRG)]
  \label{def:a_crg}
  A tuple $(C,L)$ is {\it $r$-achievable for CRG} for a source distribution $(X,Y) \sim \nu$ if for every $N \in \BN$, there is some $\ep_N$ with $\ep_N \ra 0$ as $N \ra \infty$, a key set $\MK_N$, and a private-coin protocol % \footnote{That is, $\Pi$ can use private {\it but not public} coins.}
  $\Pi = \Pi(N)$ that takes as input $(X^N, Y^N) \sim \nu^{\otimes N}$, such that if $\Pi(N)_t \in \{0,1\}^*$ denotes the message sent in the $t$-th round of $\Pi(N)$, $1 \leq t \leq r$, and $\KA = \KA(N),\KB = \KB(N) \in \MK_N$ denote the output keys of Alice and Bob for the protocol $\Pi(N)$, then:
    \begin{enumerate}
  \item $\lim \sup_{N \ra \infty} \frac{1}{N} \cdot \CC(\Pi(N)) \leq C$.
  \item $\lim \inf_{N \ra \infty} \frac{1}{N} \log |\MK_N|\geq L$.
  \item Letting $K_N$ be the random variable that is uniformly distributed on $\MK_N$, then $$\Delta((\KA(N)\KB(N)), (K_NK_N)) \leq \ep_N.$$ % , then $\lim_{N \ra \infty} \delta_N = 0$.
    In particular, there exists a coupling of $\KA(N)\KB(N)$ with $K_NK_N$ such that $\p[\KA(N) = \KB(N) = K_N] \geq 1-\ep_N \ra 1$ as $N \ra \infty$. (To be clear, $K_NK_N$ denotes the tuple $(K_N,K_N)$ which is distributed uniformly on the set $\{ (k,k) : k \in \MK_N \}$.)
  \end{enumerate}
  We denote the subset of pairs $(C,L) \subset \BR_{\geq 0}^2$ that are $r$-achievable from the source $(X,Y) \sim\nu$ by $\MT_r(X,Y)$; this set $\MT_r(X,Y)$ is known as the {\it achievable rate region for $r$-round CRG} (or simply {\it rate region}, with $r$ and the task of CRG implicit) for the source $\mu$. 
\end{defn}
To interpret Definition \ref{def:a_crg}, notice that $C$ denotes the communication of the protocols $\Pi = \Pi(N)$, whereas $L$ (approximately) gives the entropy of the key produced.

Corresponding to Definition \ref{def:a_crg} for CRG we have the following Definition \ref{def:a_skg} for SKG in the amortized setting:
\begin{defn}[Amortized SKG]
  \label{def:a_skg}
  A tuple $(C,L)$ is {\it $r$-achievable for SKG} for a distribution $\nu$ if there is some choice of a sequence $\ep_N \ra 0$ such that the following holds: for each $N \in \BN$ there is some choice of private coin protocol\footnote{As for CRG, the protocol $\Pi$ cannot use public coins.} $\Pi = \Pi(N)$ such that, first, items 1 and 2 of Definition \ref{def:a_crg} are satisfied for these $\ep_N, \Pi(N), N$, and, second,
  \begin{equation}
    \label{eq:delta_am_skg}
\Delta(\KA(N)\KB(N) \Pi(N)^r, K_N K_N \otimes \Pi(N)^r) \leq \ep_N.
\end{equation}
As in Definition \ref{def:a_crg}, $K_N$ denotes the random variable that is uniform on on $\MK_N$; notice that (\ref{eq:delta_am_skg}) above implies item 3 of Definition \ref{def:a_crg}.

We denote the set of pairs $(C,L)$ that are $r$-achievable for SKG from $\nu$ by $\MS_r(X,Y)$.
\end{defn}
It is clear from the definition that $r$-achievability for SKG is a stronger requirement than $r$-achievability for CRG; that is, for every source $(X,Y) \sim \nu$, we have $\MS_r(X,Y) \subset \MT_r(X,Y)$. It is also well-known \cite{liu2017secret,han_information-spectrum_2003} that both $\MT_r(X,Y)$ and $\MS_r(X,Y)$ are closed subsets of $\BR^2$.

  \subsection{Non-Amortized Setting}
  The non-amortized setting is similar to the amortized setting, in that Alice and Bob receive arbitrarily many i.i.d.~samples of $(X,Y) \sim \mu$, except the entropy of their key and their communication no longer grow linearly with the number of samples. Rather, the keys lie in some fixed set $\MK$, and the goal is to use as little communication (and rounds) as possible to generate a single key uniformly distributed in $\MK$. Moreover, whereas the agreement probability $1 - \ep_N$ in the amortized case was assumed to approach 1 asymptotically, in the non-amortized case, it is often of interest to study settings in which the parties may disagree with some probability that is bounded away from 0. In fact, this probability of disagreement may be arbitrarily close to 1. The non-amortized setting has recently received much attention among the theoretical computer science community \cite{bogdanov2011extracting,canonne2017communication,guruswami2016tight,ghazi2018resource,bafna_communication-rounds_2018}, where it is also known as the {\it agreement distillation problem}.

In the below definition we assume that $(X,Y) \sim \nu$ and $\nu$ is supported on a set $\MX \times \MY$.
  \begin{defn}[Non-amortized common randomness generation]
    \label{def:na_crg}
    For $r, C \in \BN$, and $L, \ep \in \BR_{\geq 0}$, we say that {\it the tuple $(C, L, \ep)$ is $r$-achievable from the source $\nu$ (for CRG)} if there is some $N \in \BN$ and an $r$-round protocol $\Pi$ with private randomness that takes as input $(X^N,Y^N) \sim \nu^{\otimes N}$, such that at the end of $\Pi$, Alice and Bob output keys $\KA, \KB \in \MK$ given by deterministic functions $\KA = \KA(X^N, \RA, \Pi^r)$, $\KB = \KB(Y^N, \RB, \Pi^r)$, such that:
    \begin{enumerate}
    \item $\CC(\Pi) \leq C$.
      \item $|\MK| \geq 2^L$.
    % \item $\min\{ H_\infty(\KA), H_\infty(\KB) \} \geq L$.
    \item There is a random variable $K$ uniformly distributed on $\MK$ such that $\p_{\nu}[K = \KA = \KB] \geq 1-\ep$.
    \end{enumerate}
  \end{defn}
  As in the amortized case, for tuples $(C, L, \ep)$, observe that $C$ denotes communication and $L$ denotes entropy.
  
  Definition \ref{def:na_crg} differs slightly from the definition of achievable rates for non-amortized CRG in \cite{bogdanov2011extracting,canonne2017communication,guruswami2016tight,ghazi2018resource,bafna_communication-rounds_2018}, which do not limit the size of the key space $\MK$, but rather require a lower bound on the min-entropy of each of $\KA, \KB$. We present this latter definition in Appendix \ref{apx:alt_defns} (Definition \ref{def:na_crg_alt}) and show that it is essentially equivalent to Definition \ref{def:na_crg}.

As in the amortized setting, in the non-amortized setting secret key generation is the same as common randomness generation except the key is additionally required to be ``almost independent'' from the transcript of the protocol:
  \begin{defn}[Non-amortized secret key generation]
    For $r, C \in \BN$ and $L \in \BR_{\geq 0}$, $\ep, \delta \in [0,1)$, we say that {\it the tuple $(C, L, \ep, \delta)$ is $r$-achievable from the source $\nu$ (for SKG)} if the tuple $(C, L, \ep)$ is $r$-achievable for CRG from the source $\nu$, and if there exists a protocol $\Pi = (\Pi^1, \ldots, \Pi^r)$ achieving the tuple such that
    \begin{equation}
      \label{eq:inf_delta_skg}
I(\Pi^r; \KA \KB) \leq \delta.
    \end{equation}
  \end{defn}
  Notice that condition (\ref{eq:inf_delta_skg}) is quite strong: it implies, for instance, that $\Delta(\Pi^r \KA \KB, \Pi^r \otimes \KA \KB) \leq \sqrt{\delta/2}$, by Pinsker's inequality.

  \subsection{Main Results: Analogue of Pointer-Chasing Separations for CRG \& SKG}
  
%   \todo{Madhu: General comment - eveywhere we should cite the arxiv version of BGGS19 ~\cite{bafna_communication-rounds_2018} and use theorem-etc numbering from there. I included a line in the citation to say that an extended abstract appeared in SODA. }
  In this section we present our main results. We first state formally the main result of \cite{bafna_communication-rounds_2018} discussed in Section \ref{sec:int_help}, which establishes an exponential separation in communication cost between $\lfloor (r+1)/2 \rfloor$-round protocols and $(r+2)$-round protocols in the non-amortized setting:
  % Question \ref{ques:main} was considered in \cite{bafna2018communication} for the non-amortized setting, where the following partial answer was given, establishing a separation in communication cost between $(r+2)$-round and $\lfloor (r+1)/2\rfloor$-round protocols, for any $r \in \BN$:
  
  \begin{theorem}[Thms.~1.1 \& 1.2 of \cite{bafna_communication-rounds_2018}]
    \label{thm:nam_formal_bggs}
    For each $r \in \BN, \ep \in [0,1)$, there exists $\eta > 0$, $\beta < \infty$, $n_0 \in \BN$ such that for any $n \geq n_0$ and any $\ell \in \BN$, there is a source $\mu_{r,n,\ell}$ such that, in the non-amortized setting:
    \begin{enumerate}
    \item[(1)] The tuple $((r+2) \lceil \log n \rceil, \ell, 0,0)$ is $(r+2)$-achievable for SKG from $\mu_{r,n,\ell}$ (and thus $((r+2) \lceil \log n \rceil, \ell,0)$ is $(r+2)$-achievable for CRG).
    \item[(2)] For any $L \in \BN$ and $C \leq \min \{ \eta L - \beta, n/\log^\beta n \}$, the tuple $(C, L, \ep)$ is not $\lfloor (r+1)/2\rfloor$-achievable for CRG (and thus the tuple $(C, L, \ep, \delta)$ is not $\lfloor (r+1)/2\rfloor$-achievable for all $\delta \geq 0$).
    \end{enumerate}
  \end{theorem}
  
  % It was shown there that for each $r \in \BN$, $\ep \in [0,1)$ there are sufficiently large $n$ such that for any $\ell$, there is a source, denoted $\mu_{r,n,\ell}$, such that the tuple $(O(\log n), \ell, 0)$ is $r$-achievable for CRG from $\mu_{r,n,\ell}$, but that for any $L \in \BN$, the tuple $(O \left(\min\{ L, n / \poly \log n \}\right), L, \ep)$ is not $(r/2 - 2)$-achievable for CRG. (Moreover, similar results hold for SKG.)
 The interpretation of the parameters $n,\ell$ in Theorem \ref{thm:nam_formal_bggs} is described in detail in Definition \ref{def:PCS} of the source $\mu_{r,n,\ell}$. We remark that the proof of item (1) of the theorem is immediate once this definition is made, and so the main content of Theorem \ref{thm:nam_formal_bggs} is in the second item (i.e., the lower bound).

  To aid understanding of Theorem \ref{thm:nam_formal_bggs}, fix any $r \in \BN, \ep \in [0,1)$, and consider parameters $\ell = n \ra \infty$; the length of Alice's and Bob's inputs under $\mu_{r,n,n}$ are $O(n^2)$. The theorem gives that with only $O(\log n)$ communication, $n$ bits of entropy can be generated in $r+2$ rounds, but if we have only roughly half as many rounds (i.e., $\lfloor (r+1)/2\rfloor$ rounds) then generating $n$ bits of entropy takes at least $n/\poly \log n$ communication, which is exponentially larger than $\log n$. It follows that for some $r'$ with $\lfloor (r+1)/2 \rfloor \leq r' < r+2$, the ratio in communication cost between the best $r'$-round protocol and the best $(r'+1)$-round protocol is at least $n^{1/(1 + \lceil (r+1)/2 \rceil)}/\log^{\omega(1)}n$. Our first main result improves this ratio to  $n^{1/4}/\log^{\omega(1)}n$ and moreover shows that such an $r'$ lies in $\{r,r+1\}$:

    \begin{theorem}[Tighter round dependence than Theorem \ref{thm:nam_formal_bggs}; non-amortized setting]
      \label{thm:nam_formal}
          For each $r \in \BN, \ep \in [0,1)$, there exists $\eta > 0$, $\beta < \infty$, $n_0 \in \BN$ such that for any $n \geq n_0$ and any $\ell \in \BN$, the source $\mu_{r,n,\ell}$ of Theorem \ref{thm:nam_formal_bggs} satisfies: % such that, in the non-amortized setting:
    \begin{enumerate}
          \item[(1)] The tuple $((r+2)\lceil \log n \rceil, \ell, 0,0)$ is $(r+2)$-achievable for SKG from $\mu_{r,n,\ell}$ (and thus $((r+2) \lceil \log n \rceil, \ell,0)$ is $r$-achievable for CRG).
    \item[(2)] For any $L \in \BN$, $C \leq  \min\{ \eta L - \beta, \sqrt n / \log^\beta n \}$, the tuple $(C, L, \ep)$ is not $r$-achievable for CRG from $\mu_{r,n,\ell}$ (and thus for any $\delta \geq 0$, the tuple $(C, L, \ep, \delta)$ is not $r$-achievable for SKG).
    % \item The tuple $(O(\log n), \ell, 0, 0)$ is $r$-achievable for SKG from $\mu_{r,n,\ell}$, but for any $\ell' \in \BN$,  $C \leq O(\left( \min\{ \ell', \sqrt n / \poly \log n \} \right))$, and $\delta \in [0,1]$, the tuple $(C, \ell', \ep, \delta)$ is not $(r-2)$-achievable for SKG from $\mu_{r,n,\ell}$.
    \end{enumerate}
  \end{theorem}
  
  Our second main result provides an exact analogue of Theorems \ref{thm:nam_formal_bggs} and \ref{thm:nam_formal} for the amortized setting:
  \begin{theorem}[Amortized setting]
    % \begin{reptheorem}{thm:am_informal}[restated]
    \label{thm:am_formal}
      For each $r \in \BN, \gamma \in (0,1)$, there is a constant $c_0 > 0$ such that for $n \geq c_0$, the source $\mu_{r,n,\ell}$ of Theorem \ref{thm:nam_formal_bggs} satisfies: %, such that: % are sufficiently large $n$ such that for any $\ell$, there is a source $\mu_{r,n,\ell}$ such that, in the amortized setting:
    \begin{enumerate}
    \item[(1)] The tuple $((r+2) \lceil \log n \rceil, \ell)$ is $(r+2)$-achievable for SKG (and thus CRG) from $\mu_{r,n,\ell}$.
    \item[(2)] Set $\ell = n$. For any $C,L \in \BR$ with $C \leq n/\log^{c_0} n$ and $L >\gamma \ell = \gamma n$, the tuple $(C,L)$ is not $\lfloor (r+1)/2 \rfloor$-achievable for CRG (and thus for SKG) from $\mu_{r,n,n}$.
    \item[(3)] Again set $\ell = n$. For any $C, L \in \BR$ with $C \leq \sqrt{n} / \log^{c_0} n$ and $L > \gamma n$, the tuple $(C,L)$ is not $r$-achievable for CRG (and thus for SKG) from $\mu_{r,n,n}$.
    \end{enumerate}
    % \end{reptheorem}
  \end{theorem}
  Notice that parts (2) and (3) of Theorem \ref{thm:am_formal} only provide a lower bound on the communication rate $C$ for protocols when the entropy rate $L$ is at least a constant factor times $n$. The problem of determining such a result for $L$ that grow sublinearly with $n$, or even those $L$ that do not grow at all (such as $L = o_n(1)$) remains open. Such a problem boils down to showing a rounds-communication tradeoff for the $r$-round {\it common random bits per interaction bit} (CBIB) of the source $\mu_{r,n,\ell}$, or equivalently, for the $r$-round {\it strong data processing constant} (SDPC) \cite{liu2017secret}; see Problem \ref{prob:cbib}. As we discuss in Section \ref{sec:mimk}, this problem seems to be quite difficult as a proof of it would immediately imply Theorem \ref{thm:am_formal}.

  \subsection{Discussion and overview of proof of Theorems \ref{thm:nam_formal} \& \ref{thm:am_formal}}
  \label{sec:proof_overview}
  The source $\mu_{r,n,\ell}$ referred to in Theorems \ref{thm:nam_formal_bggs}, \ref{thm:nam_formal} and \ref{thm:am_formal} is a variant of the well-known {\it pointer chasing distribution} from communication complexity \cite{NW,duris_lower_1984,papadimitriou_communication_1982}. This distribution was introduced to show a similar type of rounds/communication tradeoff as in the above theorems, except for the task of {\it computing functions} rather than generating a shared string. % (Recall the definition of communication complexity of functions in Section \ref{sec:protocols}.)

  Alice's and Bob's inputs from $\mu_{r,n,\ell}$ are given as follows: for an integer $n$ and odd $r$, Alice receives permutations indexed by odd integers $\pir_1, \pir_3, \ldots, \pir_r : [n] \ra [n]$, and Bob receives permutations indexed by even integers $\pir_2, \pir_4, \ldots, \pir_{r-1} : [n] \ra [n]$, as well as an integer $\ir_0 \in [n]$. Let $\jr_0 = \pir_r ( \pir_{r-1} ( \cdots \pir_1 (\ir_0))) \in [n]$. Alice and Bob also receive strings $A_1, \ldots, A_n \in \{0,1\}^\ell$ and $B_1, \ldots, B_n \in \{0,1\}^\ell$, respectively, which are distributed uniformly at random conditioned on $A_{\jr_0} = B_{\jr_0}$. If Alice and Bob have $r+2$ rounds, then the following protocol generates secret keys distributed uniformly on $\{0,1\}^\ell$: Alice sends Bob $\ir_0$, who responds with $\pir_1(\ir_0)$, Alice responds with $\pir_2(\pir_1(\ir_0))$, and so on, until both parties possess $\jr_0$, at which point they can output $A_{\jr_0} = B_{\jr_0}$.

  To prove that Alice and Bob cannot generate shared common random strings with high entropy and communication $n/\log^{\omega(1)}n$ (item (2) of Theorem \ref{thm:nam_formal_bggs}), the following approach was used: Bafna et al.~\cite{bafna_communication-rounds_2018} first reduced the problem to showing that Alice and Bob cannot succeed with high probability on a distributional version of the following communication problem: Alice receives permutations $\pir_1, \pir_3, \ldots, \pir_r : [n] \ra [n]$, an Bob receives permutations $\pir_2, \pir_4, \ldots, \pir_{r-1} : [n] \ra [n]$ and indices $\ir_0, \jr_0 \in [n]$. Their task is to determine if $\pir_r(\pir_{r-1} \cdots \pir_1(\ir_0)) = \jr_0$. This problem, called {\it pointer verification}, has a protocol with $(r+5)/2$ rounds and communication $O(\log n)$, given by Alice and Bob chasing the pointers forwards and backwards simultaneously. Bafna et al.~\cite{bafna_communication-rounds_2018} showed however that there is no protocol with $(r+3)/2$ rounds and communication $n/\log^{\omega(1)}n$, and this led to item (2) of Theorem \ref{thm:nam_formal_bggs}. We are able to prove Theorem \ref{thm:nam_formal} by employing a reduction from the CRG/SKG problem to the pointer verification problem with indices in $[n^2]$ (as opposed to in $[n]$) and with $2r$ permutations (as opposed to $r$ permutations). %, we are able to prove Theorem \ref{thm:nam_formal}.

  The proof of Theorem \ref{thm:am_formal} (in particular, of the lower bounds (2) and (3) in the theorem, as (1) is immediate) is somewhat more involved. The overall goal is to reduce to the non-amortized case (Theorems \ref{thm:nam_formal_bggs} and \ref{thm:nam_formal}), and to do this, three main ingredients are needed. The first ingredient is a characterization of the achievable rate region $\MT_r(X,Y)$ for CRG in terms of the {\it internal information cost} and {\it external information cost} \cite{barak2013compress} of private-coin communication protocols, which has been referred to many times in the literature (e.g., \cite{sudan_communication_2019,ghazi2018resource}). This characterization shows that that if for $L \leq \ell$, the pair $(C,L)$ is $r$-achievable for CRG from $\mu_{r,n,\ell}$ (i.e., belongs to $\MT_r(X,Y)$), then there is an $r$-round private-coin protocol $\Pi$ with inputs $(X,Y) \sim \mu_{r,n,\ell}$ with internal information cost at most $C$ and external information cost at least $L$ (see Corollary \ref{cor:det_inf}). % \todo{confirm}.

  The second ingredient of the proof is a result of Jain et al.~\cite{jain2012direct} (which is implicit in the earlier work of Braverman and Rao \cite{braverman2011information}) stating that for any $r$-round protocol $\hat \Pi$ with internal information cost $I$, there exists an $r$-round protocol $\hat \Pi'$ that simulates $\hat \Pi$ up to some accuracy loss $\ep$ and has communication cost at most $\frac{I + O(r)}{\ep}$ (see Theorem \ref{thm:jain_compress}). Applying this to $\hat \Pi = \Pi$, one might hope to show that $\Pi$ leads to a protocol with communication $O(C)$ and key length $\Omega(L)$ for non-amortized CRG. However, the error $\ep$ introduced in the information-to-communication compression result of Jain et al.~\cite{jain2012direct} makes this conclusion nontrivial, which necessitates the third ingredient: a delicate argument that makes use of the specific structure of $\mu_{r,n,\ell}$ is needed to complete the reduction (see Lemmas \ref{lem:get_useful_info} and \ref{lem:test_i}).

  \section{Proof of Theorem \ref{thm:nam_formal}; non-amortized setting}
  \label{sec:nam}
  In this section we prove Theorem \ref{thm:nam_formal}. To prove this theorem we need to introduce the pointer chasing source of \cite{bafna_communication-rounds_2018}, and also recall the notion of ``indistinguishability'' of two distributions to low-round low-communication protocols. We then state our main technical theorem (Theorem~\ref{thm:indist})  about the indistinguishability of the pointer chasing source from an ``independent source'' (where Alice and Bob get inputs that are independent of each other). Section~\ref{ssec:indist} is devoted to the proof of Theorem~\ref{thm:indist}.

  We begin by formally defining the {\it pointer-chasing source} $\mu_{r,n,\ell}$ that the theorem uses to achieve the rounds-communication tradeoff.
  \begin{defn}[The Pointer Chasing Source $\mu_{r,n,\ell}$, \cite{bafna_communication-rounds_2018}, Definition 2.1]
\label{def:PCS}
	For positive integers $r$, $n$ and $\ell$, the support of $\mu = \mu_{r,n,\ell}$ is 
	$(\MS_n^{\lceil r/2\rceil} \times \{0,1\}^{n\ell}) \times ([n]\times \MS_n^{\lfloor r/2 \rfloor}\times \{0,1\}^{n\ell})$. Denoting $X = (\pir_1,\pir_3,\ldots,\pir_{2\lceil r/2 \rceil-1},A_1,\ldots,A_n)$ and $Y = (\ir,\pir_2,\pir_4,\ldots,\pir_{2\lfloor r/2 \rfloor},B_1,\ldots,B_n)$, a sample $(X,Y) \sim \mu$ is drawn as follows:
	\begin{itemize}
		\item $\ir \in [n]$ and $\pir_1,\ldots,\pir_r \in \MS_n$ are sampled uniformly and independently. 
		\item Let $\jr = \pir_r(\pir_{r-1}(\cdots \pir_1(\ir)\cdots)) \in [n]$. 
		\item $A_\jr = B_\jr \in \{0,1\}^\ell$ is sampled uniformly and independently of $\ir$ and $\pir$'s. 
	    \item For every $\kp \ne \jr$, $A_\kp\in \{0,1\}^\ell$ and $B_\kp\in \{0,1\}^\ell$ are sampled uniformly and independently.
        \end{itemize}
      \end{defn}
      We use the following notational convention for samples $(X,Y) \sim \mu_{r,n,\ell}$. We write $\ir_0 := \ir$, and for $1 \leq t \leq r$, $\ir_t := \pir_t(\ir_{t-1})$. Similarly, we write $\jr_0 := \jr$, and for $1 \leq t \leq r$, $\jr_{t-1} = \pir_t^{-1}(\jr_t)$. Over the distribution $\mu_{r,n,\ell}$, we thus have $\ir_t = \jr_{r-t}$ for $0 \leq t \leq r$ with probability 1.

      We establish the following basic property of the pointer chasing source $\mu_{r,n,\ell}$ for future reference:
        \begin{lemma}
    \label{lem:mu_inf}
    When $(X,Y) \sim \mu_{r,n,\ell}$, $I(X; Y) = \ell$.
  \end{lemma}
  \begin{proof}
    Notice that $H(X) = r \log(n!) + n\ell$ since $\pir_1, \pir_3, \ldots, \pir_{2\lceil r/2 \rceil - 1}$ are uniformly random in $\MS_n$ and $A_1, \ldots, A_n$ are uniformly random in $\{0,1\}^\ell$. Moreover, % $H(X | Y) = r \log(n!) + (n-1)\ell$, since for any  $y$, the distribution of $X | Y = y$ is that where $\pir_1, \pir_3, \ldots, \pir_{2 \lceil r/2 \rceil - 1}$ are uniformly random in $\MS_n$, $A_{i_r}$ is fixed to $B_{i_r}$, and the remaining $n-1$ $A_i$s are uniformly random in $\{0,1\}^\ell$. Formally, we have:
    \begin{eqnarray}
      H(X | Y) &=& H(\pir_1, \pir_3, \ldots, \pir_{2\lceil r/2 \rceil - 1}, A_1, \ldots, A_n | Y) \nonumber\\
               &=& H(\pir_1, \ldots, \pir_{2\lceil r/2 \rceil - 1} | Y) + H(A_1, \ldots, A_n | Y, \pir_1, \ldots, \pir_{2\lceil r/2 \rceil - 1})\nonumber\\
      &=& r \log(n!) + (n-1)\ell\nonumber.
    \end{eqnarray}
  \end{proof}

    It is immediate from the definition of $\mu_{r,n,\ell}$ that part (1) (i.e., the upper bound) of Theorem \ref{thm:nam_formal} holds: in particular, the parties ``chase the pointers'', i.e., alternatively send $\ir_t$, $0 \leq t \leq r$, and finally output $A_{\ir_r} = B_{\ir_r}$ as their keys.
  \if 0
  \begin{lemma}[Upper bound for Theorems \ref{thm:nam_formal_bggs} \& \ref{thm:nam_formal}]
    \label{lem:pc_ub}
    For every $r,n,\ell$, the tuple $((r+2) \lceil \log n \rceil, \ell, 0)$ is $(r+2)$-achievable for SKG (and thus $((r+2) \lceil \log n \rceil, \ell)$ is $(r+2)$-achievable for CRG) from $\mu_{r,n,\ell}$.
  \end{lemma}
  \begin{proof}
    Consider the protocol in which Alice sends Bob an arbitrary bit in the first round, and in round $t+1$, for $1 \leq t \leq r$, the next party to speak sends over $\ir_t = \pir_t(\ir_{t-1}) \in [n]$, which takes ⌈log n⌉ bits. Then Alice outputs $A_{\ir_r}$ as her key and Bob outputs $B_{\ir_r}$ (which is equal to $A_{\ir_r}$ with probability 1 over $\mu_{r,n,\ell}$) as his key. By construction of $\mu_{r,n,\ell}$, $A_{\ir_r} = B_{\ir_r}$ is independent of $\ir_1, \ldots, \ir_r$, which is the transcript of the protocol.
  \end{proof}
  \fi
  The main content of Theorems \ref{thm:nam_formal_bggs} and \ref{thm:nam_formal} is then in part (2) (i.e., the lower bound) of each; its proof, for both theorems, proceeds via arguments about indistinguishbility of inputs to protocols, which we will now define. 
    For $r, C \in \BR_+$, we say that a communication protocol $\Pi$ is an {\it $(r,C)$ protocol} if $\Pi$ has at most $\lfloor r \rfloor$ rounds and communication cost at most $\lfloor C \rfloor$.
  \begin{defn}[Indistinguishability, \cite{bafna_communication-rounds_2018}, Definition 3.1]
    Let $0 \leq \ep \leq 1$. Two distributions $\mu_1, \mu_2$ on pairs $(X,Y)$ are {\it $\ep$-distinguishable} to a protocol $\Pi$ if the distribution of the transcript $\Pi^r$ when $(X,Y) \sim \mu_1$ has total variation distance at most $\ep$ from the distribution of $\Pi^r$ when $(X,Y) \sim \mu_2$.

    Two distributions $\mu_1, \mu_2$ are {\it $(\ep, C, r)$-indistinguishable} if they are $\ep$-indistinguishable to every $(r,C)$ protocol. The distributions $\mu_1, \mu_2$ are {\it $(\ep, C, r)$-distinguishable} if they are not $(\ep, C, r)$-indistinguishable. If $\Pi$ is a protocol such that the total variation distance of the transcript between inputs $(X,Y) \sim \mu_1$ and inputs $(X,Y) \sim \mu_2$ is at least $\ep$, then we say that $\Pi$ distinguishes between $\mu_1$ and $\mu_2$ {\it with advantage $\ep$}.
  \end{defn}

  Proposition \ref{prop:reduce_to_ind} reduces the problem of showing that certain tuples $(C, L)$ are not achievable for CRG from $\mu_{r,n,\ell}$ to that of showing indistinguishability of $\mu_{r,n,\ell}$ from the product of its marginals $(\mu_{r,n,\ell})_X \otimes (\mu_{r,n,\ell})_Y$.
  \begin{proposition}[\cite{bafna_communication-rounds_2018}, Propositions 3.3 \& 3.4]
    \label{prop:reduce_to_ind}
    There are positive constants $\eta, \xi$ such that the following holds. Suppose $\rho, C, L \in \BN$ and $0 < \gamma < 1$. Suppose that $C < \eta L - 3/2 \cdot \log 1/\gamma - \xi$ and that the tuple $(C, L, 1-\gamma)$ is $\rho$-achievable for CRG from the source $\mu_{r,n,\ell}$. Then there is some $N \in \BN$ such that $\mu_{r,n,N\ell}$ and $(\mu_{r,n,N\ell})_X \otimes (\mu_{r,n,N\ell})_Y$ are $(\gamma / 10, C + \xi \log 1/\gamma, \rho + 1)$-distinguishable.
  \end{proposition}

Our main theorem for this section is the following indistiguishability result for $\mu = \mu_{r,n,\ell}$ versus $\mu_X \times \mu_Y$. In contrast to the analogous result in \cite[Lemma 4.5]{bafna_communication-rounds_2018}, our result shows indistinguishability for protocols with $r+1$ rounds albeit with a smaller communication budget.

    \begin{theorem}
    \label{thm:indist}
    For every $\ep > 0$ and $r \in \BN$ there exists $\beta, n_0$ such that for every $n \geq n_0$ and $\ell$, the distributions $\mu = \mu_{r,n,\ell}$ and $\mu_X \otimes \mu_Y$ are $(\ep, r+1, \sqrt{n}/\log^{\beta}n)$-indistinguishable.
  \end{theorem}
%   We present the proof of Theorem \ref{thm:nam_formal} below; the proof of Theorem \ref{thm:nam_formal_bggs} is omitted since it is very similar, and can be found in \cite{bafna2018communication}.

  Using Proposition \ref{prop:reduce_to_ind}, the proof of Theorem \ref{thm:nam_formal} follows from Theorem \ref{thm:indist}.
  \begin{proof}[Proof of Theorem \ref{thm:nam_formal}]
%     Recall that item (1) of Theorem \ref{thm:nam_formal} was shown in Lemma \ref{lem:pc_ub}, so we only have to prove item (2).
  We need to show item (2).   
    Fix $\ep > 0$ and $r \in \BN$. Let $\xi, \eta$ be the constants from Proposition \ref{prop:reduce_to_ind}. Also let $\beta_0$ be the constant $\beta$ from Theorem \ref{thm:indist} with $(1-\ep)/20$ as the variational distance parameter. Also let $\beta$ be a constant such that $\beta > \max \{ \beta_0, 3/2 \cdot \log 1/(1-\ep) + \xi \}$ and $\sqrt n/\log^\beta n + \xi \log 1/(1-\ep) \leq \sqrt n / \log^{\beta_0}n$, which is possible for sufficiently large $n$. Suppose for purpose of contradiction that for some $L > 0$, the tuple $(\min \{ \eta L - \beta, \sqrt n / \log^\beta n \}, L, \ep)$ were $r$-achievable for CRG from $\mu_{r,n,\ell}$. Since $\beta > 3/2 \log 1/(1-\ep) + \xi$, it follows from Proposition \ref{prop:reduce_to_ind} that for some $N \in \BN$, $\mu_{r,n,\ell N}$ and $(\mu_{r,n,\ell N})_X \otimes (\mu_{r,n,\ell N})_Y$ are $((1-\ep)/10,\sqrt n / \log^{\beta_0} n, r+1 )$-distinguishable. But  this contradictions Theorem \ref{thm:indist}, which states that $(\mu_{r,n,\ell N})_X \otimes (\mu_{r,n,\ell N})_Y$ are $((1-\ep)/20,\sqrt n / \log^{\beta_0} n, r+1 )$-indistinguishable.
  \end{proof}

  To complete the proof of Theorem \ref{thm:nam_formal} it therefore suffices to prove Theorem \ref{thm:indist}. We do so in the following subsection.

  \subsection{Disjointness and Proof of Theorem \ref{thm:indist}}\label{ssec:indist}

% \todo{Madhu: Changed the citation below to point to Lemma 4.5 of BGGS. We should follow similar rules (specific theorem numbers) elsewhere in this paper. Also changed lower bound to $n/poly log n$ from $\sqrt{n}/poly log n$.}
  Next we work towards the proof of Theorem \ref{thm:indist}; the proof parallels that of a corresponding result of Bafna et al., which shows that the distributions $\mu = \mu_{r,n,\ell}$ and $\mu_X \otimes \mu_Y$ are $(\ep, \lfloor (r+3)/2 \rfloor,n/\log^\beta n)$-indistinguishable (see \cite[Lemma 4.5]{bafna_communication-rounds_2018}).\todo{Noah: Changed $(r-2)/2$ to $(r+3)/2$ -- I should check once more...} A central ingredient in the proof of \cite{bafna_communication-rounds_2018} is a ``pointer verification problem'' (see Definition~\ref{defn:pv} below) and an indistinguishability result they show for this problem (see Theorem~\ref{thm:pv-main}). We use the same notion and indistiguishability result, with the main difference being that we are able to reduce a ``$2r$''-round pointer verification problem to our problem whereas the proof in \cite{bafna_communication-rounds_2018} could only reduce an $r$round pointer verification problem to the same. This factor of $2$ leads to the gain in this section.

  The proof proceeds by eliminating each of two possible strategies Alice and Bob can use to distinguish $\mu_{r,n,\ell}$ and $(\mu_{r,n,\ell})_X \otimes (\mu_{r,n,\ell})_Y$: first, they can try to follow the chain of pointers, compute $\ir_r$, and check if $A_{\ir_r} = B_{\ir_r}$ (which is true with probability 1 under $\mu_{r,n,\ell}$ but only with probability $1/2^\ell$ under $(\mu_{r,n,\ell})_X \otimes (\mu_{r,n,\ell})_Y$). Computing $\ir_r$, however, with fewer than $r+2$ rounds requires communication $\Omega(n)$ by standard results for the pointer chasing problem \cite{NW}. Alternatively, Alice and Bob can ignore the chain of pointers and try to determine if there is {\it any} $i$ such that $A_i = B_i$ (under the product distribution the probability that such an $i$ exists is at most $n/2^\ell \ll 1$). As observed in \cite{bafna_communication-rounds_2018}, determining the existence of such an $i$ is no easier than solving the set disjointness problem \cite{Razborov}, which requires communcation $\Omega(n)$. However, combining the pointer chasing and set disjointness lower bounds takes some care, and ultimately leads to the fact that we are only able to lower-bound the communication cost of $r$-round (as opposed to $(r+1)$-round) protocols, and get a bound of $\tilde \Omega(\sqrt{n})$ (as opposed to $\tilde \Omega(n)$). We begin by recalling the $\Omega(n)$ lower bound on the distributional communication complexity of disjointness with respect to a particular distribution: % Both theorems will use the following result that the distributional communication complexity of disjointness (over a particular distribution) is $\Omega(n)$.
  \begin{theorem}[\cite{Razborov}]
    \label{thm:disj}
  For every $\ep > 0$ there exists $\delta > 0$ such that for all $n$ the following holds. Let $\disjy = \disjy_n$ (respectively, $\disjn = \disjn_n$) denote the uniform distribution on pairs $(U, V)$ with $U, V \subseteq [n]$ and $|U| = |V| = n/4$ such that $|U \cap V| = 1$ (respectively, $|U \cap V| = 0$). Then if Alice gets $U$ and Bob gets $V$ as inputs, $\disjy$ and $\disjn$ are $(\ep, \delta n, \delta n)$-indistinguishable to Alice and Bob.
  \end{theorem}

  We will use the following corollary of Theorem \ref{thm:disj}:
  % For Theorem \ref{thm:indist} we will need the following corollary.
  \begin{corollary}
  \label{cor:sqrtn_disj}
  For every $\ep > 0$ there exists $\delta > 0$ such that for all $n$ the following holds. Let $\disjy_{n,\sqrt n}$ (respectively, $\disjn_{n,\sqrt n}$) denote the uniform distribution on pairs $(U, V)$ with $U, V \subseteq [n]$ and $|U| = |V| = n/4$ such that $|U \cap V| =\lfloor \sqrt n\rfloor$ (respectively, $|U \cap V| = 0$). Then if Alice gets $U$ and Bob gets $V$ as inputs, $\disjy_{n,\sqrt n}$ and $\disjn_{n,\sqrt n}$ are $(\ep, \delta \sqrt n, \delta \sqrt n)$-indistinguishable to Alice and Bob.
\end{corollary}
\begin{proof}
  A protocol $\Pi$ that distinguishes $\disjy_{n^2,n}$ and $\disjn_{n^2, n}$ with communication $C$ may be converted into a protocol $\Pi'$ with communication $C$ that distinguishes $\disjy_{n}$ and $\disjn_n$ with advantage $\ep$. In particular, the protocol $\Pi'$ proceeds as follows: given inputs $(U,V)$, $|U| = n/4, |V| = n/4$, Alice and Bob construct an instance $(U',V')$, that is distributed according to $\disjy_{n^2,n}$ if $(U,V) \sim \disjy_n$ and that is distributed according to $\disjn_{n^2,n}$ if $(U,V) \sim \disjn_n$. In particular, Alice and Bob first construct sets $(\tilde U, \tilde V)$ as follows: for each $u \in U \subset [n]$, Alice places the elements $(u-1)n + j$, for $1 \leq j \leq n\}$ in $\tilde U$, and Bob constructs $\tilde V$ in an analogous fashion. Then, using public randomness, they randomly permute the elements of $\tilde U, \tilde V$ (according to the same permutation) to obtain sets $U', V'$. It is clear that $|\tilde U| = | \tilde V| = |U'| = |V'| = n \cdot |U| = n^2/4$. Moreover, if $|U \cap V| = 0$, then $|\tilde U \cap \tilde V| = |U' \cap V'| = 0$, and if $|U \cap V| = 1$, then $|U' \cap V'| = n = \sqrt{n^2}$.

  By Theorem \ref{thm:disj}, for any $\ep > 0$, there is $\delta > 0$ such that the protocol $\Pi'$ must have communication at least $\delta n$. Thus the protocol $\Pi$ must have communication at least $\delta n = \sqrt{\delta^2 n^2}$.

  It follows in a similar manner as the above argument that any protocol $\Pi$ distinguishing $\disjy_{n',n}$ and $\disjn_{n',n}$ with $n^2 \leq n' < (n+1)^2$ with communication $C$ may be converted into a protocol $\Pi'$ with communication $C$ that distinguishes $\disjy_n$ and $\disjn_n$ with advantage $\ep$. This completes the proof of the corollary even for non-perfect squares $n$.
\end{proof}

Next we state the second main ingredient in the proof of Theorem \ref{thm:indist}, which is a hardness result for the {\it pointer verification} problem introduced in \cite[Definition 4.1]{bafna_communication-rounds_2018}. The inputs to pointer verification are similar to those of the standard pointer chasing problem, except that Alice and Bob receive as inputs a final pointer $\jr_0$ in addition to the initial pointer $\ir_0$, and the goal is to determine if $\pir_r \circ \cdots \circ \pir_1(\ir_0) = \jr_0$: % We define a distinguishability version of this problem below: % This problem is certainly easier than the standard pointer chasing problem since Alice and Bob can % precise result on hardness of pointer chasing that is needed in Theorems \ref{thm:indist} and \ref{thm:indist_bggs}. We first introduce some notation, following \cite{bafna2018communication}.
\begin{defn}[\cite{bafna_communication-rounds_2018}, Definition 4.1]\label{defn:pv}
  Let  $r, n \in \BN$ with $r$ odd. Then the distributions $\dypv = \dypv(r,n)$ and $\dnpv = \dnpv(r,n)$ are both supported on $((\MS_n^{\lceil r/2}) \times ([n]^2 \times S_{n}^{\lfloor r/2 \rfloor})$, and are defined as follows:
  \begin{itemize}
  \item $\dnpv$ is the uniform distribution on $((\MS_n^{\lceil r/2}) \times ([n]^2 \times S_{n}^{\lfloor r/2 \rfloor})$.
  \item $(X,Y) \sim \dypv$, with $X = (\pir_1, \pir_3, \ldots, \pir_r), Y = (\ir_0, \jr_0, \pir_2, \pir_4, \ldots, \pir_{r-1})$ is sampled by letting $\pir_1, \pir_2, \ldots, \pir_r$ be independent and uniform over $\MS_n$, letting $\ir_0 \in [n]$ be uniform and independent of the $\pir_t$, and setting $\jr_0 = \pir_r \circ \cdots \circ \pir_1(\ir_0)$.
  \end{itemize}
\end{defn}

Notice that with $(r+5)/2$ rounds of communication, by communicating at most $1 + (r+1) \lceil \log n \rceil$ bits, Alice and Bob can distinguish between $\dypv(r,n)$ and $\dnpv(r,n)$ with advantage $1-1/n$. In particular, Alice sends Bob an arbitrary bit in the first round, Bob sends $\ir_0, \jr_0$ in the second round, Alice responds with $\ir_1 = \pir_1(\ir_0)$ and $\jr_1 = \pir_r^{-1}(\jr_0)$, Bob responds with $\ir_2$ and $\jr_2$, and so on. After $(r+3)/2$ rounds either Alice or Bob will know both $\ir_{(r-1)/2}$ and $\jr_{(r-1)/2}$, and this person sends $\One[\pir_{(r+1)/2}(\ir_{(r-1)/2}) = \jr_{(r-1)/2}]$ (which is 1 with probability 1 under $\dypv$ and only with probability $1/n$ under $\dnpv$) as the final bit. % Therefore, PV with $r$ permutations is easier than the standard pointer chasing problem, which requires $r$ rounds for a protocol with communication cost $O(\log n)$.

Theorem \ref{thm:pv-main} states that if Alice and Bob are only allowed $1$ fewer round, then they must communicate exponentially more bits to distinguish $\dypv$ and $\dnpv$:
\begin{theorem}[\cite{bafna_communication-rounds_2018}, Theorem 4.2]
	\label{thm:pv-main}
	For every $\epsilon > 0$ and odd $r$ there exists $\beta, n_0$ such for every $n \geq n_0$, $\dypv(r,n)$ and $\dnpv(r,n)$ are $(\epsilon,(r+3)/2,n/\log^\beta n)$-indistinguishable.
  \end{theorem}

  Using Theorem \ref{thm:pv-main} and Corollary \ref{cor:sqrtn_disj}, we now prove Theorem \ref{thm:indist}.
  \begin{proof}[Proof of Theorem \ref{thm:indist}]
    We introduce a new distribution, which we denote by $\hat \mu$ (or $\hat \mu_{r,n,\ell}$ when we want to emphasize dependence on $r,n,\ell$); $\hat \mu$ is a distribution supported on $(\MS_n^{\lceil r/2 \rceil} \times (\{0,1\}^\ell)^n \times (\MS_n^{\lfloor r/2 \rfloor} \times [n] \times (\{0,1\}^\ell)^n)$. We denote a sample from $\hat \mu$ by $(X, Y)$, with $$X = (\pir_1, \pir_3, \ldots, \pir_{2\lceil r/2 \rceil - 1}, A_1, \ldots, A_n),\ \  Y = (i, \pir_2, \pir_4, \ldots, \pir_{2\lfloor r/2 \rfloor}, B_1, \ldots, B_n),$$ which is distributed as follows:
\begin{itemize}
\item $\ir_0 \in [n]$ and $\pir_1, \ldots, \pir_r \in \MS_n$ are sampled uniformly and independently. Let $\ir_r = \pir_r \circ \cdots \circ \pir_1(\ir_0)$. 
\item Let $P \subset [n]$ be a uniformly random subset of size $\lfloor\sqrt n\rfloor$, conditioned on the event that it contains $\ir_r$.
\item For every $j \in P$, $A_j = B_j \in \{0,1\}^L$ is sampled uniformly and independently of $i$, $\pir$'s, and $P$.
\item For every $j \not \in P$, $A_j, B_j \in \{0,1\}^L$ are sampled uniformly and independently (and independently of all $\pir$'s, $j$, and $P$).
\end{itemize}
\begin{claim}
  \label{clm:muhatmu}
  For every $\ep > 0$, there exists $\delta > 0$ such that the distributions $\mu_{r,n,\ell}$ and $\hat \mu_{r,n,\ell}$ are $(\ep, \delta \sqrt n, \delta \sqrt n)$-indistinguishable.
\end{claim}
\begin{proof}[Proof of Claim \ref{clm:muhatmu}]
  We show that any protocol $\Pi$ with $\CC(\Pi) \leq C$ distinguishing $\mu = \mu_{r,n,\ell}$ and $\hat \mu = \hat \mu_{r,n,\ell}$ with advantage $\ep$ can be converted into a protocol $\Pi'$ with $\CC(\Pi') \leq C$ and which distinguishes $\disjy_{n,\sqrt n}$ and $\disjn_{n,\sqrt n}$ (as in Corollary \ref{cor:sqrtn_disj}) with advantage $\ep$.

  The protocol $\Pi'$ proceeds as follows: suppose Alice and Bob are given sets $U,V$, respectively, with $U,V \subseteq [n]$. Let $m = (\lfloor \sqrt n \rfloor + 1)^2$. Using public randomness, Alice and Bob sample a random injective function $\tau : [n] \ra [m]$, and set $U' = \{ \tau(u) : u \in U\} , V'  = \{ \tau(v) : v \in V \}$. Let $\jr_0 \in [m]$ denote the sole index not in the image of $\tau$. 
  Using public randomness, Alice and Bob sample $r$ permutations $\pir_1, \ldots, \pir_r\in S_{m}$ uniformly and independently, and let $\ir_0 = (\pir_r \circ \cdots \circ \pi_1)^{-1}(\jr_0)$. They also sample $3m$ strings $A_1, \ldots, A_{m}, B_1, \ldots, B_{m}, C_1, \ldots, C_{m} \in \{0,1\}^\ell$ uniformly and independently using public randomness. Then for $1 \le u \leq {m}$, Alice sets:
  $$
   A_u' := \begin{cases}
    A_u &: u \not \in U' \\
    C_u &: u \in U',
  \end{cases}
  $$
  and Bob sets:
    $$
  B_u' := \begin{cases}
    B_u &: u \not \in U' \\
    C_u &: u \in U'.
  \end{cases}
  $$
  It is now clear that the tuple
  \begin{equation}
    \label{eq:disj_hatmu}
((\pir_1, \pir_3, \ldots, \pir_r, A_1', A_2', \ldots, A_m'), (\ir_0, \pir_2, \pir_4, \ldots, \pir_{r-1}, B_1', B_2', \ldots, B_m'))
\end{equation}
is distributed according to $\mu_{r,m,\ell}$ if $(U,V) \sim \disjn_{n,\sqrt n}$ and is distributed according to $\hat \mu_{r,m,\ell}$ if $(U,V) \sim \disjy_{n,\sqrt n}$. Now Alice and Bob run the protocol $\Pi$ with their inputs as in (\ref{eq:disj_hatmu}).

By Corollary \ref{cor:sqrtn_disj}, for each $\ep > 0$, there exists $\delta > 0$ such that $\mu_{r,m,\ell}$ and $\hat \mu_{r,m,\ell}$ are $(\ep, \delta \sqrt n, \delta \sqrt n)$-indistinguishable. Using that $\sqrt{m} - \sqrt n = O(1)$, the lemma statement follows.
\end{proof}

Next, notice that the two distributions $(\mu_{r,n,\ell})_X \otimes (\mu_{r,n,\ell})_Y$ and $(\hat \mu_{r,n,\ell})_X \otimes (\hat \mu_{r,n,\ell})_Y$ are identical. Thus by Claim \ref{clm:muhatmu} and the triangle inequality for total variation distance, Theorem \ref{thm:indist} will follow from the following claim:
\begin{claim}
  \label{clm:hatmu_indist}
     For every $\ep > 0$ and $r \in \BN$ there exists $\beta, n_0$ such that for every $n \geq n_0$ and $\ell$, the distributions $\hat \mu = \hat\mu_{r,n,\ell}$ and $\hat \mu_X \otimes \hat \mu_Y$ are $(2\ep, r+1, \sqrt{n}/\log^{\beta}n)$-indistinguishable.
\end{claim}

We next introduce a distribution $\mumid = \mumid_{r,n,\ell}$, which is the same as $\hat \mu_{r,n,\ell}$, except the distribution of the uniformly random subset $P \subset [n]$ with $|P| = \lfloor \sqrt n \rfloor$ is not conditioned on the event that it contains $\ir_r$ (i.e.~it is drawn uniformly at random from the set of all $\sqrt n$-element sets, independent of $\ir_0, \pir_1, \ldots, \pir_r$). Thus, with probability at least $1 - 1/\sqrt n$, $\ir_r \not \in P$ under $\mumid$. Now Claim \ref{clm:hatmu_indist} follows directly from the triangle inequality and Claims \ref{clm:mumid_mu} and \ref{clm:mumid_prod} below.% Finally, we will let $\mu_{X}$ and  $\mu_{Y}$ denote the marginals of Alice's and Bob's inputs uner $\mu_{r,n,L,\sqrt n}$, so that $\mu_{X} \otimes \mu_{Y}$ is the product of marginals. Note that these are the same marginals as for the distribution $\mu_{r,n,L}$.
\begin{claim}
  \label{clm:mumid_mu}
For every $\ep > 0$ and $r \in \BN$ there exists $\beta, n_0 \in \BR_+$ such that for all integers $n \geq n_0$ and $\ell$, the distributions $\hat \mu_{r,n,\ell}$ and $\mumid_{r,n,\ell}$ are $(\ep, r+1, \sqrt n /\log^\beta n)$-indistinguishable.
\end{claim}
\begin{claim}
  \label{clm:mumid_prod}
  For every $\ep > 0$, there exists $\delta > 0$ such that $\mumid_{r,n,\ell}$ and $(\hat \mu_{r,n,\ell})_X \otimes (\hat \mu_{r,n,\ell})_Y$ are $(\ep, \delta \sqrt n, \delta \sqrt n)$-indistinguishable for all $n \in \BN$.
\end{claim}
Now we prove each of Claims \ref{clm:mumid_mu} and \ref{clm:mumid_prod} in turn.
\begin{proof}[Proof of Claim \ref{clm:mumid_mu}]
  We first prove the statement of the claim for the case that $n$ is a perfect square.
    Fix $r,n,\ell$, and suppose that $\Pi$ is a $\rho$-round protocol ($\rho \in \BN$) with communication at most $C$ that distinguishes between $\hat \mu_{r,n^2,\ell}$ from $\mumid_{r,n^2,\ell}$ with advantage $\ep$. (Notice that we are replacing $n$ with $n^2$ in the notation.)

    We now construct a protocol $\Pi'$ with the same number of rounds and communication as $\Pi$ and which distinguishes between $\dypv(2r-1, n)$ and $\dnpv(2r-1,n)$ with advantage at least $\ep$. Suppose Alice and Bob are given inputs $X = (\pir_1, \pir_3, \ldots, \pir_{2r-1})$ and $Y = (\ir_0, \jr_0, \pir_2, \pir_4,  \ldots, \pir_{2r-2}),$ respectively, which are distributed according to $\dypv(2r-1,n)$ or $\dnpv(2r-1,n)$. Next, for $1 \leq t \leq r-1$, let $\pir_t' = \pir_t$, and for $r+2 \leq t \leq 2r$, let $\pir_t' = \pir_{t-1}$. Finally let $\pir_r', \pir_{r+1}' \in \MS_n$ be uniformly random conditioned on $\pir_{r+1}' \circ \pir_r' = \pir_r$. Notice that each $\pir_t'$, $1 \leq t \leq 2r$ may be computed by either Alice or Bob.
    % $\tilde \mu_{r,n,\ell}$ from $(\tilde \mu_{r,n,\ell})_X \otimes (\tilde \mu_{r,n,\ell})_Y$, as follows. Given inputs $X = (\pir_1, \pir_3, \ldots, \pir_r, \pir_{r+1}, \pir_{r+3}, \ldots, \pir_{2r})$ and $Y = (\ir_0, \jr_0, \pir_2, \pir_4, \ldots, \pir_{r-1}, \pir_{r+2}, \pir_{r+4}, \ldots, \pir_{2r-1})$ to Alice and Bob respectively, the protocol proceeds as follows.
    Next, interpret $[n^2] \simeq [n] \times [n]$, so that any pair $\sigma, \tau \in \MS_n$ of permutations on $[n]$ determines a permutation on $[n^2]$, which we denote by $\sigma || \tau$, so that $(\sigma || \tau) ((i,j)) = (\sigma(i), \tau(j))$. (Note that the vast majority of permutations on $[n^2]$ cannot be obtained in this manner, however.) The protocol $\Pi'$ proceeds as follows:
  \begin{enumerate}
  \item Alice and Bob use their common randomness to generate uniformly random permutations $\tau_0, \tau_1, \ldots, \tau_r \in S_{n^2}$ and uniformly random strings $A_1, \ldots, A_{n^2 - n}, B_1, \ldots, B_{n^2 - n}, C_1, \ldots, C_n\in \{0,1\}^\ell$. 
  \item Bob computes $\hat \ir_0 := \tau_1((\ir_0, \jr_0)) \in [n] \times [n] \simeq [n^2]$.
  \item For $t = 1, 3, \ldots, 2\lfloor (r+1)/2 \rfloor$, Alice computes $\hat \pir_t := \tau_t \circ (\pir'_t || (\pir'_{2r + 1-t})^{-1} ) \circ \tau_{t-1}^{-1} \in S_{n^2}$.
  \item For $t = 2, 4, \ldots, 2\lfloor r/2 \rfloor$, Bob computes $\hat \pir_t := \tau_t \circ (\pir'_t || (\pir'_{2r + 1-t})^{-1} ) \circ \tau_{t-1}^{-1} \in S_{n^2}$.
  \item For $1 \leq i \leq n$, Alice and Bob set $\hat A_{\tau_r((i, i))} = \hat B_{\tau_r((i,i))} = C_i$.
  \item For the $n^2 - n$ pairs $(i,j) \in [n] \times [n]$ with $i \neq j$, Alice sets $\hat A_{(i,j)}$ to be equal to one of the $A_k$, $1 \leq k \leq n^2 - n$ so that each $A_k$ is used once. Bob does the same with $\hat B_{(i,j)}$ with respect to the $B_k$.
  \item Alice and Bob now run the protocol $\Pi$ on the inputs $\hat X := (\hat \pir_1, \hat \pir_3, \ldots, \hat \pir_{2\lfloor (r+1)/2 \rfloor}, \hat A_1, \ldots, \hat A_{n^2})$ and $\hat Y := (\hat \ir_0, \hat \pir_2, \hat \pir_4, \ldots, \hat \pir_{2\lfloor r/2 \rfloor}, \hat B_1, \ldots, \hat B_{n^2})$. 
  \end{enumerate}
  Certainly the communication cost and number of rounds of $\Pi'$ are both the same as the communication cost and number of rounds, respectively, of $\Pi$.

  We will show that (1) if $(X,Y) \sim \dypv(2r-1, n)$, then $(\hat X, \hat Y) \sim \hat \mu_{r,n^2,\ell}$, and (2) if $(X,Y) \sim \dnpv(2r-1,n)$, then $(\hat X, \hat Y) \sim \mumid_{r,n^2,\ell}$.

  We first prove (1). Suppose $(X,Y) \sim \dypv(2r-1,n)$. That is, $X,Y$ are uniformly random conditioned on $\pir_{2r-1} \circ \cdots \circ \pir_1(\ir_0) = \jr_0$; therefore,
  $$
\pir_1', \pir_2', \ldots, \pir_{2r}', \ir_0, \jr_0
$$
are uniformly random conditioned on $\pir_{2r}' \circ \cdots \circ \pir_1'(\ir_0) = \jr_0$. For $0 \leq t \leq 2r$, set $\ir_t' = \pir'_t \circ \pir'_{t-1} \circ \cdots \circ \pir'_1(\ir_0)$ (so that, in particular, $\ir_0' = \ir_0$).  Then the distribution of $\pir_1', \ldots, \pir_{2r}', \ir_0, \jr_0$ may be expressed equivalently as follows: $X,Y$ are chosen as follows: $\pir'_1, \ldots, \pir'_{2r}$ are first drawn uniformly and independently form $S_{n}$, an index $\ir'_r \in [n]$ is chosen uniformly in $[n]$ independent of $\pir'_1, \ldots, \pir'_{2r}$, and then we set $\jr_0 = \pir'_{2r} \circ \cdots \circ \pir'_{r+1}(\ir'_r)$ and $\ir_0 = (\pir'_1)^{-1} \circ \cdots \circ (\pir'_r)^{-1}(\ir'_r)$. % and then $(\ir_0, \jr_0)$ is chosen uniformly at random from the set of $n$ pairs $(i,j) \in [n] \times [n]$ such that $\pir_{2r} \circ \cdots \circ \pir_1(i) = j$.

  Notice that the set $P:= \{ \tau_r ((i,i)) : i \in [n] \}$ is a uniformly random set of size $n$ in $[n^2] \simeq [n] \times [n]$. Next, note that if $\pi$ is any distribution on $S_{n^2}$ and $\tau$ is distributed uniformly on $S_{n^2}$, then $\pi \circ \tau$ is distributed uniformly on $S_{n^2}$.  It follows from this fact $\hat \pir_1, \ldots, \hat \pir_r$ are distributed uniformly and independently in $S_{n^2}$, all independent of the set $P = \{ \tau_r ((i,i)) : i \in [n]\}$. Next, we have that
  \begin{eqnarray}
    \hat \pir_r \circ \cdots \circ \hat \pir_1(\hat \ir_0) &=& \tau_r \circ (\pir'_r || (\pir'_{r+1})^{-1}) \circ \tau_{r-1}^{-1} \circ \tau_{r-1} \circ \cdots \circ \tau_1^{-1} \circ \tau_1 \circ (\pir'_1 || (\pir'_{2r})^{-1}) \circ \tau_0^{-1} \circ \tau_0((\ir_0, \jr_0))\nonumber\\
                                                       &=& \tau_r \circ (\pir'_r || (\pir'_{r+1})^{-1}) \circ \cdots \circ ((\pir'_1) || (\pir_{2r}')^{-1}) ((\ir_0, \jr_0))\nonumber\\
                                                       &=& \tau_r((\pir'_r \circ \cdots \circ \pir'_1(\ir_0), (\pir'_{r+1})^{-1} \circ \cdots \circ (\pir'_{2r})^{-1}(\jr_0)))\nonumber\\
    &=& \tau_r((\ir'_r, \ir'_r))\nonumber,
  \end{eqnarray}
  where we have used the fact that $(X,Y) \sim  \dypv(2r-1,n)$ in the last line. Recall from the discussion above that $\ir'_r$ is independent of $\pir'_1, \ldots, \pir'_{2r}, \tau_0, \ldots, \tau_r$, and therefore $\tau_r((\ir'_r, \ir'_r))$ is a uniformly random element of the set $P = \{ \tau_r((i,i)) : i \in [n]\}$, independent of $\hat \pir_1, \ldots, \hat \pir_r, P$. Therefore, $\hat \ir_0$ is a uniformly random element of $[2n]$, independent of $P, \hat \pir_1, \ldots, \hat \pir_r$, conditioned on the event $\hat \pir_r \circ \cdots \circ \hat \pir_1(\hat \ir_0) \in P$. This establishes that $( \hat X, \hat Y) \sim \hat \mu_{r,n^2,\ell}$, finishing the proof of point (1).

  We next prove (2); suppose that $(X,Y) \sim \dnpv(2r-1,n)$. Then all of the random variables $\pir'_1, \ldots, \pir'_{2r} \in S_{n^2}$, and $\ir_0, \jr_0 \in [n]$ are uniform and independent on their respective domains. Moreover, the set $P := \{ \tau_r((i,i)) : i \in [n]\}$ is a uniformly random set of size $n$ in $[n^2] \simeq [n] \times [n]$. Thus $\hat \pir_1, \ldots, \hat \pir_r \in S_{n^2}$ are uniform and independent in $S_{n^2}$, independent of $P$, and $\hat \ir_0 \in [n^2]$ is uniform, independent of $P, \hat \pir_1, \ldots, \hat \pir_r$. This establishes that in this case $(\hat X, \hat Y) \sim \mumid_{r,n^2,\ell}$.

  Thus the distribution of the transcript of $\Pi'$ (excluding the additional public randomness used by $\Pi'$ in the simulation above) when run on $\dypv$ (respectively, $\dnpv$) is the same as the distribution of the transcript of $\Pi$ when run on $\hat \mu_{r,n^2, \ell}$ (respectively, $\mumid_{r,n^2,\ell}$). It then follows from Theorem \ref{thm:pv-main} and the fact that $((2r-1) + 3)/2 = r+1$ that for every $\ep > 0$, there exists $\beta, n_0 \in \BR_+$ such that for all $\ell \in \BN$ and perfect squares $n \geq n_0$, the distributions $\hat \mu_{r,n,\ell}$ and $\mumid_{r,n,\ell}$ are $(\ep, r+1, \sqrt n / \log^\beta n)$-indistinguishable.

  The case that $n$ is not a perfect square follows immediately: in particular, given a sample $(X,Y)$ from either $\hat \mu_{r,n,\ell}$ or $\mumid_{r,n,\ell}$, let $m$ denote the smallest perfect square greater than $n$. Notice that by viewing $[n]$ as a subset of $[m]$ and using public randomness Alice and Bob can create a sample $(X',Y')$ that is sampled from $\hat \mu_{r,m,\ell}$ if $(X,Y) \sim \hat \mu_{r,n,\ell}$ and that is sampled from $\mumid_{r,m,\ell}$ if $(X,Y) \sim \mumid_{r,n,\ell}$ with no communication.
  % Since $\Pi$ distinguishes between the distributions $\mu_{r,n^2, \ell, n}$ and $\mumid$, with advantage $2\ep$, it follows that $\Pi'$ distinguishes between $(\tilde \mu_{r,n,\ell})_X \otimes (\tilde \mu_{r,n,\ell})_Y$ with advantage $2\ep$.
\end{proof}
Next, Claim \ref{clm:mumid_prod} follows as a simple corollary of Corollary \ref{cor:sqrtn_disj}.
\begin{proof}[Proof of Claim \ref{clm:mumid_prod}]
  The proof is similar to that of Claim \ref{clm:muhatmu}. We reduce the task of distinguishing $\mumid_{r,n,\ell}$ and $(\hat \mu_{r,n,\ell})_X \otimes (\hat \mu_{r,n,\ell})_Y$ to the task of distinguishing $\disjy_{n,\sqrt n}$ and $\disjn_{n,\sqrt n}$ (See Corollary \ref{cor:sqrtn_disj}).

  In particular, suppose Alice and Bob are given $U, V \subseteq [n]$. Alice and Bob share common random uniform strings $Z_1, \ldots, Z_n \in \{0,1\}^\ell$. Given $U \subset [n]$, Alice sets $A_u = Z_u$ for $u \in U$ and samples $A_u \in \{0,1\}^\ell$ uniformly and independently for all $u \in [n] \backslash U$. Similarly, for $V \subset [n]$, Bob sets $B_v = Z_v$ for $v \in V$, and samples $B_v \in \{0,1\}^\ell$ uniformly and independently for all $v \in [n] \backslash V$. Alice also samples $\pir_1, \pir_3, \ldots, \pir_r \in \MS_n$ uniformly and independently and Bob samples $\pir_2, \pir_4, \ldots, \pir_{r-1} \in \MS_n$, $\ir, \jr \in [n]$ uniformly and independently. Letting $X = (\pir_1, \pir_3, \ldots, \pir_r, A_1, \ldots, A_n)$ and $Y = (\ir, \jr, \pir_2, \pir_4, \ldots, \pir_{r-1}, B_1, \ldots, B_n)$, it is easy to see that $(X,Y) \sim \mumid_{r,n,\ell}$ if $(U,V) \sim \disjy_{n,\sqrt n}$ and that $(X,Y) \sim (\hat \mu_{r,n,\ell})_X \otimes (\hat \mu_{r,n,\ell})_Y$ if $(U,V) \sim \disjn_{n,\sqrt n}$. It follows from Corollary \ref{cor:sqrtn_disj} that for any $\ep > 0$ there exists $\delta > 0$ such that $\mumid_{r,n,\ell}$ and $(\hat \mu_{r,n,\ell})_X \otimes (\hat \mu_{r,n,\ell})_Y$ are $(\ep, \delta \sqrt n, \delta \sqrt n)$-indistinguishble.
\end{proof}
We have now verified Claims \ref{clm:mumid_mu}, \ref{clm:mumid_prod}, which establishes Claim \ref{clm:hatmu_indist}, which completes the proof of Theorem \ref{thm:indist}, and thus of Theorem \ref{thm:nam_formal}.
\end{proof}

% The proof of Theorem \ref{thm:indist_bggs} is similar to that of Theorem \ref{thm:indist}. The two main ingredeints are (1) the standard $\Omega(n)$ lower bound for disjointness, Theorem \ref{thm:indist}, and (2) Theorem \ref{thm:pv-main} on the hardness of pointer verification. We omit the details, which can be found in \cite{bafna2018communication}.

\todo{Madhu needs to read the rest}

\section{Proof of Theorem \ref{thm:am_formal}; amortized setting}
\label{sec:am}
In this section we work towards the proof of Theorem \ref{thm:am_formal}; recall that part (1) is immediate, so the main work is in proving parts (2) and (3). As discussed in Section \ref{sec:proof_overview}, there are 3 main steps in the proof, which proceeds by initially assuming that the tuple $(C,L)$ is $r$-achievable for appropriate values of $C,L$ and eventually deriving a contradiction. The first step is to establish a {\it single-letter characterization}\footnote{The term ``single-letter characterization'' is used relatively loosely in the literature. Following \cite{csiszar_information_1981}, for any $k \in \BN$ and a closed subset $\MS \subset \BR^k$, we call a characterization of $\MS$ a {\it single-letter characterization} if it implies, for any $\eta > 0$, the existence of an algorithm that decides whether a point $x \in \BR^k$ is of Euclidean distance at most $\eta$ to $\MS$. Moreover, this algorithm must run in time at most $T_\MS(\eta)$, for some function $T_\MS : \BR_+ \ra \BN$. This is related, for instance, to ideas on the computability of subsets of $\BR^k$ considered in \cite{braverman_complexity_2005}. % For instance, for the characterization given in Theorem \ref{thm:ahlswede_1998}, the set $\MS$ is given by $\MT_1(X,Y)$, and the algorithm iterates through all possible conditional distributions of $U | X$ where $U$ is supported on some set $\MU$ of size $|\MU| \leq |\MX|$, with a sufficiently small granularity (depending on $\eta$). Correctness follows by the continuity of the Shannon entropy.
} of the achievable rate region $\MT_r(X,Y)$ for amortized CRG, which we explain in Section \ref{sec:single_letter}. This single-letter characterization will show that if the tuple $(C,L)$ is $r$-achievable for CRG from any source $\mu$, then there is an $r$-round protocol with internal information cost at most $C$ and external information cost at least $L$. In Section \ref{sec:compress_ic_cc}, we show how to convert this protocol into a nearly equivalent protocol whose {\it communication cost} is at most $C$ (recall that in general, $\CC(\Pi) \geq \ICint_\mu(\Pi)$, so upper bounding communication cost is more difficult). Finally, in Section \ref{sec:am_formal_proof} we show how to use the fact that the external information cost is at least $L$ to obtain a protocol that can distinguish between the pointer-chasing distribution $\mu_{r,n,\ell}$ and the product of the marginals $(\mu_{r,n,\ell})_X \otimes (\mu_{r,n,\ell})_Y$. At this point we will obtain a contradiction for appropriate values of $C,L$ by Theorems \ref{thm:indist} and \ref{thm:indist_bggs}, which were the key ingredients in the proof for the corresponding lower bounds in the non-amortized setting (i.e., item (2) of Theorems \ref{thm:nam_formal_bggs} and \ref{thm:nam_formal}).

\subsection{Single-letter characterization of $\MT_r(X,Y)$}
\label{sec:single_letter}
It follows immediately from Definitions \ref{def:a_crg} and \ref{def:a_skg} that the $r$-round rate region for amortized CRG and SKG is completely characterized by, for each communication rate $C$, the maximum real number $L$, known as the {\it capacity}, such that $(C,L)$ is $r$-achievable for CRG or SKG:
  \begin{defn}[CR \& SK capacity]
    Suppose a source $(X,Y) \sim \mu$ is fixed. Then for $r \in \BN, C \in \BR_+$, define the {\it CR capacity with communication $C$} to be
    $$
    \FCamcr_r(C) := \sup_{(C,L) \in \MT_r(X,Y)} L,
    $$
    and the {\it SK capacity with communication $C$} to be
    $$
\FCamsk_r(C) := \sup_{(C,L) \in \MS_r(X,Y)} L.
$$
% (Recall the definitions of $\MT_r(X,Y)$ and $\MS_r(X,Y)$ in Definitions \ref{def:a_crg} and \ref{def:a_skg}.) When we want to emphasize dependence of $\FCamcr_r(\cdot), \FCamsk_r(\cdot)$ on $\mu$, we write $\FCamcr_r(C | \mu)$ and $\FCamsk_r(C | \mu)$, respectively.
  \end{defn}
%   In their seminal work on CRG in the amortized setting, Ahlwede and Csisz\'{a}r \cite{ahlswede1998common} computed the following {\it single-letter characterization} of $\FCamcr_1(C)$, or equivalently, of $\MT_1(X,Y)$:

  The single-letter characterization of $\MT_r(X,Y)$ relies on the concepts of {\it internal information cost} and {\it external information cost} of a protocol $\Pi$ \cite{barak2013compress,braverman2011information,braverman2013direct,braverman_information_2013,braverman_interactive_2012}. %, where they are referred to as the {\it external information cost} and {\it internal information cost}, respectively, of the protocol induced by $\Pi_1 = U$. More generally,
  The external information cost of a (multiple-round) protocol $\Pi$ describes how much information $\Pi$ reveals about the inputs $X,Y$ to an external observer who only sees the transcript of the protocol, while the internal information cost describes how much information Alice and Bob reveal to {\it each other} about their own inputs:
    \begin{defn}[External and internal information costs]
      Given any communication protocol $\Pi$ with a maximum of $r$ rounds, public randomness $\Rpub$, and a distribution $(X,Y) \sim \mu$ of inputs, the {\it external information cost} $\ICext_\mu(\Pi)$ is given by:
      $$
\ICext_\mu(\Pi) := I(\Pi^r, \Rpub ; X, Y).
$$
If $\Pi$ does not use public randomness, then $\ICext_\mu(\Pi) := I(\Pi^r; X,Y)$.

The {\it internal information cost} $\ICint_\mu(\Pi)$ is given by
$$
\ICint_\mu(\Pi) := I(\Pi^r, \Rpub; X | Y) + I(\Pi^r, \Rpub; Y | X).
$$
If $\Pi$ does not use public randomness, then $\ICint_\mu(\Pi):= I(\Pi^r; X|Y) + I(\Pi^r; Y|X)$.
\end{defn}
\begin{remark} It is well-known that for any distribution $\mu$, and any protocol $\Pi$, $\ICint_\mu(\Pi) \leq \ICext_\mu(\Pi) \leq \CC(\Pi)$.
  \end{remark}

    An original motivation behind the introduction of internal and external information costs was to understand the possibility of proving {\it direct sum} results for communication complexity \cite{chakrabarti2001informational,jain_direct_2003,harsha_communication_2007,barak2013compress}. % Such a direct sum result would state that the communication complexity of computing $N$ independent copies of a function (i.e., with $N$ independent pairs of inputs $(X_i, Y_i)$) is roughly $N$ times the communication complexity of computing a single copy of the function. This problem was initially considered in \cite{karchmer_super-logarithmic_1995}, where a direct sum result for deterministic communication complexity was conjectured for a certain relation, and it was shown that a proof of this conjecture would imply $\PP \not\subseteq \NC^1$. A (weak) direct sum result was shown for the deterministic communication complexity of computing functions \cite{feder_amortized_1995}, where it was proven that if the deterministic communication complexity of computing $f$ is $C$, then the deterministic communication complexity of computing $n$ copies of $f$ is $\Omega(\sqrt{C} n)$. For the case of randomized and distributional communication complexity, it is known that no tight direct sum theorem (i.e., one that states that the complexity of computing $n$ copies of any function $f$, each correctly with probability $2/3$, is $\Omega(C n)$) holds \cite{ganor2014exponential,ganor_exponential_2016,rao_simplified_2018}, but the possibility of a weak direct sum result still remains open \cite{braverman_candidate_2018}.
    In light of the connection with direct sum results, the fact that internal and external information costs appear in characterizations for amortized CRG and SKG is not surprising. In particular, the amortized CRG and SKG problems can be viewed as the task of solving $N$ independent instances of CRG or SKG from a source $\mu$, with an additional requirement that each of Alice's $N$ output strings must agree with each of Bob's $N$ output strings {\it simultaneously} with high probability. 
    % computing common randomness or a secret key with entropy proportional to $N$ from each of $N$ independent copies of a distribution $(X,Y) \sim \mu$.
    % In fact, the proof of our Theorem \ref{thm:am_formal} uses Theorem \ref{thm}, which has been used to prove direct sum results for bounded-round protocols \cite{braverman2011information}.
    % as one of the main motivations behind those works is to understand the communication complexity of completing $N$ independent copies of some given task, for large $N$, and in the amortized CRG problem such a task would roughly be generating common randomness from a single copy $(X,Y) \sim \mu$.

    An additional ingredient in the single-letter chararcterization of $\MT_r(X,Y)$ is the {\it minimum $r$-round interaction for maximum key rate} (i.e., the {\it $r$-round MIMK}). Ahlswede and Csisz\'{a}r showed in their seminal work \cite{ahlswede1993common} that the maximum key rate $L$ that Alice and Bob can generate from a source $(X,Y) \sim \mu$, without restricting communication, is $I_\mu(X;Y)$. In other words, we have: $\sup_{C \geq 0} \FCamsk_r(C) = I(X;Y)$. The $r$-round MIMK describes the minimum amount of communication needed to obtain this key rate of $I(X;Y)$:
    \begin{defn}
      \label{def:mimk}
      If $(X,Y) \sim \mu$ is a source and $r \geq 1$, Then the $r$-round MIMK is defined as
      $$
      \FI_r(X;Y) = \inf_{C \geq 0 : \FCamsk_r(C) = I(X;Y)} \{ C \}.
      $$
    \end{defn}
    Tyagi \cite{tyagi2013common} proved the following single-letter characterization of  the $r$-round MIMK $\FI_r(X;Y)$:
    % Recall the definition of the minimum $r$-round interaction for achieving the maximum key rate ($r$-round MIMK; Definition \ref{def:mimk}). The following theorem provides a single letter characterization of the $r$-round MIMK in terms of $\MTd_r(X,Y)$.
\begin{theorem}[\cite{tyagi2013common}, Theorem 4]
  \label{thm:tyagi_sl}
  For a source $(X,Y) \sim \mu$, the $r$-round MIMK is the infimum of all $C \geq 0$ such that there exists an $r$-round private-coin protocol $\Pi$ such that $\ICint_\mu(\Pi) \leq C$ and $\ICext_\mu(\Pi) \geq C + I(X;Y)$. 
%  Suppose we are given a source $(X,Y) \sim \mu$. Then for $r \in \BN$, $C \in \BR_+$, the {\it minimum interaction for maximum key rate} is
%   \begin{equation}
%     \label{eq:fi_ext_int}
% \FI_r(X;Y) = \inf \left\{L - I(X;Y) \ : \ (L - I(X;Y),L) \in \MTd_r(X;Y)  \right\}.
%   \end{equation}
\end{theorem}
%     Theorem \ref{thm:tyagi_sl} is proved by relating $\FI_r(X;Y)$ to a generalization of Wyner's common information \cite{Wyner_CommonInfo}. % The above quantity $\FI_r(X;Y)$ is called the minimum interaction for maximum key rate since it is the minimum communication rate for which an $r$-round amortized protocol can achieve a SK of rate $I(X;Y)$, which is the maximum possible achievable rate:

%     Notice that for all $r$, $\FI_r(X;Y) \leq H(X|Y)$, since the 1-round protocol $\Pi$ in which Alice sends her input $X = \Pi_1$ satisfies $\ICint(\Pi) = I(X;X|Y) = H(X|Y)$ and $\ICext(\Pi) = I(X;XY) = H(X)$, and thus $\ICext(\Pi) - \ICint(\Pi) = I(X;Y)$. It follows similarly that for $r \geq 2$, $\FI_r(X;Y) \leq \min \{ H(X|Y), H(Y|X) \}$.
    Using Theorem \ref{thm:tyagi_sl}, we finally can state the single-letter characterization of achievable rates for $r$-round CRG. It is stated most precisely in \cite{sudan_communication_2019}, but similar results are shown in \cite{liu2017secret,ghazi2018resource,liurate,ye_information_2005,gohari_information-theoretic_2010,gohari_information-theoretic_2010-1}.
    \begin{theorem}[\cite{sudan_communication_2019}, Theorem III.2]
      \label{thm:achievability_main}For $C \geq 0$, define
      \begin{equation}
        \label{eq:tilde_fc}
        \tilFCamcr_r(C) := 
         \begin{cases}
    \sup_{\Pi = (\Pi_1, \ldots, \Pi_r) : \ICint_\mu(\Pi) \leq C} \{\ICext_\mu(\Pi)\} \quad &: \quad C \leq \FI_r(X;Y) \\ 
%     \sup_{\Pi = (\Pi_1, \ldots, \Pi_r) :\ICint_\mu(\Pi)\leq C} \ICext_\mu(\Pi) \quad : \quad C \leq \FI_r(X;Y)\\
    I(X;Y) + C \quad &: \quad C > \FI_r(X;Y),
  \end{cases}
\end{equation}
  where the supremum is ocver all $r$-round private-coin protocols $\Pi$ with $\ICint_\mu(\Pi) \leq C$.

  Then for a source $(X,Y) \sim \mu$, the $r$-round CR capacity is given by
        \begin{equation}
          \label{eq:ach_main}
\FCamcr_r(C) = \tilFCamcr_r(C).
  \end{equation}
% \item[(2)] The region $\MTd_r(X,Y) \subset \MT_r(X,Y)$ is exactly the set of tuples $(C,L)$ that are achievable by {\it deterministic} protocols $\Pi$.
%   \item[(3)] $\FCamsk_r(C) = \FCamcr_r(C) - C$. 
\end{theorem}
\begin{remark}
  We briefly explain how Theorem \ref{thm:achievability_main} does in fact provide a single-letter characterization for $\MT_r(X,Y) = \{ (C,L) : C \geq 0, L \leq \FCamcr_r(C) \}$. It follows from the support lemma \cite[Lemma 15.4]{csiszar_information_1981} that the protocols $\Pi = (\Pi_1, \ldots, \Pi_r)$ in the definition of $\tilFCamcr_r(C)$ can be restricted to the class of protocols where $\Pi_t$, $1 \leq t \leq r$, falls in a finite set of size $\MU_t$ at most $|\MX| |\MY| \prod_{t'=1}^{t-1} |\MU_{t'}| + 1$. Then by iterating through all possible distributions of $\Pi_t | \Pi^{t-1}X$, for $t$ odd, and $\Pi_t | \Pi^{t-1} Y$, for $t$ even, at a sufficiently small granularity, we can approximate $\tilFCamcr_r(C)$ to any given precision. % By Theorem \ref{thm:tyagi_sl}, similar considerations apply regarding the computation of $\FI_r(X;Y)$ (which is expressed in (\ref{eq:fi_ext_int}) entirely in terms of $\MTd_r(X,Y)$).
\end{remark}
For the purpose of proving Theorem \ref{thm:am_formal}, we will only need the inequality $\FCamcr_r(C) \leq \tilFCamcr_r(C)$ (which is often called the {\it converse direction} of the equality in Theorem \ref{thm:achievability_main}). As a full proof of Theorem \ref{thm:achievability_main} (and in particular, of this inequality) does not appear to have been collected in the literature, we provide one in Section \ref{sec:converse_proof}. The following is an immediate consequence of this inequality:
\begin{corollary}
  \label{cor:det_inf}
For each tuple $(C, L)\in \MT_r(X,Y)$ with $L < I(X;Y)$, there is some protocol $\Pi = (\Pi_1, \ldots, \Pi_r)$ such that $\ICint_\mu(\Pi) \leq C$ and $\ICext_\mu(\Pi) \geq L$. 
\end{corollary}
\begin{proof}
  First suppose that $C \leq \FI_r(X;Y)$. Then the  existence of the $r$-round protocol $\Pi$ follows from (\ref{eq:ach_main}).
  
  Next suppose $C > \FI_r(X;Y)$. Notice that $(\FI_r(X;Y), I(X;Y)) \in \MT_r(X,Y)$, since $\FCamcr_r(\FI_r(X;Y)) = I(X;Y) + \FI_r(X;Y)$. Therefore, the case $C \leq \FI_r(X;Y)$ gives that there is an $r$-round protocol $\Pi$ such that $\ICint_\mu(\Pi) \leq \FI_r(X;Y) < C$ and $\ICext_\mu(\Pi) \geq I(X;Y) > L$, as desired.
  % First we suppose $L = \FCamcr_r(C)$. The function $C \mapsto \FCamcr_r(C)$ is non-decreasing by definition, so there cannot be any $C' \leq C$ with $(C', C' + I(X;Y)) \in \MTd_r(X,Y)$ (as otherwise we would have $\FCamcr_r(C') \geq I(X;Y) > \FCamcr_r(C)$). In particular, $C \leq \FI_r(X;Y)$, and thus $(C, \FCamcr_r(C)) \in \MTd_r(X,Y)$.
\end{proof}

\subsection{Using the compression of internal information to communication}
\label{sec:compress_ic_cc}
A crucial technical ingredient in doing so is the use of an ``compression of internal information cost to communication'' result for bounded round protocols, saying that for any protocol with a fixed number $r$ of rounds and internal information cost $I$, there is another protocol with the same number $r$ of rounds and communication cost not much larger than $I$. As we discussed in Section \ref{sec:background}, these types of theorems were originally proved in order to establish direct sum and direct product results for communication complexity. Our use of these compression results may be interpreted as a roughly analogous approach for the setting of amortized CRG and SKG, which can be thought of as the ``direct sum version of non-amortized CRG and SKG''.

    \begin{theorem}[Lemma 3.4, \cite{jain2012direct}]
    \label{thm:jain_compress}
    Suppose that $(X,Y) \sim \nu$ are inputs to an $r$-round communication protocol $\Pi$ with public randomness $\Rpub$ (and which may use private coins as well. Then for every $\ep > 0$, there is a public coin protocol $L$ with $r$ rounds and communication at most $\frac{\ICint_\mu(\Pi) + 5r}{\ep} + O(r \log(1/\ep))$ such that at the end of the protocol each party possesses a random variable $(\hat \Pi_1, \ldots, \hat \Pi_r)$ representing a transcript for $\Pi$, which satisfies
    $$
\Delta((\Rpub, X, Y, \Pi_1, \ldots, \Pi_r), (\Rpub, X, Y, \hat \Pi_1, \ldots, \hat \Pi_r)) \leq 6 \ep r.
    $$
  \end{theorem}

  Our first lemma, Lemma \ref{lem:get_useful_info}, uses Theorem \ref{thm:jain_compress} to show that for any protocol $\Pi$ which satisfies $\ICext_\mu(\Pi) \gg \ICint_\mu(\Pi)$ then there exists another protocol $\Pi$ with communication cost not much greater than $\ICint_\mu(\Pi)$ and which satisfies some additional properties:
  \begin{lemma}
  \label{lem:get_useful_info}
  Fix any $r,n,\ell \in \BN$, and let $\mu = \mu_{r,n,\ell}$. Suppose $\rho \in \BN$ and $C, L \in \BR_+$.
  Suppose $\Pi$ is a $\rho$-round protocol with $\ICext_\mu(\Pi) = L$ and $\ICint_\mu(\Pi) = C$ and public randomness $\Rpub$ (and which may use private randomness as well).
  % Suppose that there is an $r'$-round deterministic \todo{randomized?} protocol with inputes $(X^N, Y^N) \sim \mu^{\otimes N}$ that outputs keys $\KA, K_B \in \{0,1\}^{N\alpha L}$, respectively, such that for some $R, H, \ep > 0$, it $r$-achieves the tuple $(R, H)$ with error $\theta \leq 1/(\alpha L)$.
  Then for every $\ep > 0$ there is some $\rho$-round protocol $\Pi'$ with inputs $(X, Y) \sim \mu$, public randomness $\Rpub$, with communication at most $\frac{C + 5\rho}{\ep} + O(\rho \log 1/\ep)$ and which outputs keys $ \KA',  \KB'$, such that
  \begin{enumerate}
  \item $\p_{\mu}[ \KA' =  \KB'] = 1$. %  \geq 1 - \ep(1 + 6r)$.
  \item When inputs $(X,Y)$ are drawn from $\mu$, $I( \KA'; B_{\ir_r}) = I( \KA' ; A_{\ir_r}) \geq L - (C +1 + 2 \log n + 36\ep \rho \ell)$. 
  \item When inputs $(X,Y)$ are drawn from $\mu_X \otimes \mu_Y$,
    \begin{equation}
      \label{eq:ka_bs}
      I_{\mu_X \otimes \mu_Y}( \KA', \Rpub, (\Pi')^{\rho}; B_1, \ldots, B_n) \leq \frac{C + 5\rho}{\ep} + O(\rho \log 1/\ep)
    \end{equation}
    and
    \begin{equation}
      \label{eq:kb_as}
      I_{\mu_X \otimes \mu_Y}( \KB', \Rpub, (\Pi')^{\rho}; A_1, \ldots, A_n) \leq \frac{C + 5\rho}{\ep} + O(\rho \log 1/\ep).
      \end{equation}
  \end{enumerate}
\end{lemma}
\begin{proof}
% Recall the definition $R' = 2(R + \ep c + 1)$.
Let $\Pi'$ be the protocol given by Theorem \ref{thm:jain_compress} for the protocol $\Pi$ and the given $\ep$. Then the communication of $\Pi'$ is at most $\frac{C + 5\rho}{\ep} + O(\rho \log 1/\ep))$. At the  end of $\Pi'$, Alice and Bob each possess a random variable $(\hat \Pi_1, \ldots, \hat \Pi_{\rho})$, such that, when $(X,Y) \sim \mu$,
\begin{equation}
  \label{eq:tvd_simulation}
\Delta((\Rpub, X, Y, \hat \Pi_1, \ldots, \hat \Pi_{\rho}), (\Rpub, X, Y, \Pi_1, \ldots, \Pi_{\rho})) \leq 6 \ep \rho.
\end{equation}
(Notice that $\hat \Pi^\rho = (\hat \Pi_1, \ldots, \hat \Pi_\rho)$ is different from the transcript $(\Pi')^\rho = (\Pi'_1, \ldots, \Pi'_\rho)$ of $\Pi'$.) 
Now set $\KA' = \KB' = (\hat \Pi_1, \ldots, \hat \Pi_{\rho})$, which immediately establishes item (1) of the lemma. % \todo{seems a bit sketchy that agree w/ prob 1 -- verify!}
% By item (3) of Lemma \ref{lem:cr_sim} and Theorem \ref{thm:jain_compress}, and the fact that $\p_{\mu}[\KA = K_B] \geq 1-\ep$, it follows that $\p_{\mu}[\KA'' = K_B''] \geq 1-\ep - 6\ep r$, which establishes item (1) of the lemma.

  To establish point (2), we will first argue that it holds for $\Pi$; in particular we show that when $(X,Y) \sim \mu$,
  \begin{equation}
    \label{eq:piprime_inf}
  H(B_{\ir_r} | \Pi^{\rho}) \leq \ell + C - L + 2\log n.
  \end{equation}
(Since $H(B_{i_r}) = \ell$ it will follow from (\ref{eq:piprime_inf}) that $I_\mu( \Pi^{\rho}; B_{\ir_r}) \geq L - C - 2\log n$, though we will not use this directly.)
  % To see this, let $\tilde \Pi$ be the protocol where Alice outputs $\KA$ at the conclusion of $\Pi$, and since the protocol $\tilde \Pi'$ given by Lemma \ref{lem:cr_sim} is simply the protocol $\Pi$ where Alice outputs $\KA'$ at the end, item (2) of Lemma \ref{lem:cr_sim} gives that
  To see this, first notice that\footnote{We remark that the equality of $I(X;Y | \Pi^{\rho})$ to (\ref{eq:int_m_ext}) also played a crucial role in \cite{liu2017secret} which derived a characterization of the achievable rate region in terms of the convex envelope of a functional on source distributions.}
  \begin{eqnarray}
    I(X; Y | \Pi^\rho) &=& I(Y ;X,  \Pi^{\rho}) - I(\Pi^\rho; Y) \nonumber\\
                                &=& I(X;Y) + I(\Pi^\rho; Y | X) + I(\Pi^\rho; X | Y) - I(\Pi^\rho; X, Y)\nonumber\\
    \label{eq:int_m_ext}
                             &=& I(X; Y) + \ICint_\mu(\Pi) - \ICext_\mu(\Pi)\\
    &\leq & \ell + C - L\nonumber.
  \end{eqnarray}
  
  Recalling the notation $\ir_r = \pir_r \circ \cdots \circ \pir_1(\ir_0)$, we observe by Lemma \ref{lem:cond_ineq_1} and the data processing inequality that
  \begin{eqnarray}
    I(X; Y | \Pi^\rho) &\geq & I(X; Y | \Pi^\rho, \ir_r) - \log n\nonumber\\
                             & \geq & I(A_{\ir_r}; B_{\ir_r} | \Pi^\rho, \ir_r) - \log n\nonumber\\
                             & \geq & I(A_{\ir_r} ; B_{\ir_r} | \Pi^\rho) - 2\log n\nonumber\\
    & =& H(A_{\ir_r} | \Pi^\rho) - 2 \log n = H(B_{\ir_r} | \Pi^\rho) - 2\log n\nonumber,
  \end{eqnarray}
  since $H(A_{\ir_r} | B_{\ir_r}, \Pi^\rho) = H(A_{\ir_r} | B_{\ir_r}) = 0$ as $A_{\ir_r} = B_{\ir_r}$ for all inputs in the support of $\mu$. It then follows that $H(B_{\ir_r} | \Pi^\rho, \Rpub) \leq \ell + C- L + 2\log n$, establishing (\ref{eq:piprime_inf}). % Since $H(B_{i_r}) = L$, (\ref{eq:piprime_inf}) follows.

  Next, (\ref{eq:tvd_simulation}) and the data processing inequality give us that $\Delta((\Rpub, B_{\ir_r}, \Pi^\rho), (\Rpub, B_{\ir_r}, \hat \Pi^\rho)) \leq 6 \ep \rho$. Corollary \ref{cor:rev_cond_pinsk} and (\ref{eq:piprime_inf}) then give that
  $$
H(B_{\ir_r} | \hat \Pi^\rho, \Rpub) \leq H(B_{\ir_r} | \hat \Pi^\rho) \leq \ell + C- L + 2 \log n + 36 \ep \rho \ell + 1.
$$
Since $\KA' = \hat \Pi^\rho$, we get that
$$
I(B_{\ir_r} ; \KA') \geq L - (C + 1 + 2 \log n + 36\ep \rho \ell),
$$
which establishes point (2).

Finally, to establish point (3), first notice that some inputs $(X,Y) \sim \mu_X \otimes \mu_Y$ may not be in the support of $\mu$. We may extend the protocol $\Pi'$ to be defined for all pairs of inputs $(X,Y) \in \MX \times \MY$, by choosing an arbitrary behavior (e.g., terminating immediately) whenever there is a partial transcript $(\Pi')^{t-1}$ for which the distribution of the next message $\Pi'_t$ has not been defined. 

Recall that $(\Pi'_1, \ldots, \Pi'_\rho)$ denotes the transcript of communication of $\Pi'$ and $\Rpub$ is the public randomness of $\Pi'$, so that when $(X,Y) \sim \mu_X \otimes \mu_Y$,
$$
I_{\mu_X \otimes \mu_Y}((\Pi')^\rho, X, \Rpub; Y) = I_{\mu_X \otimes \mu_Y}((\Pi')^\rho ; Y | X, \Rpub) \leq H_{\mu_X \otimes \mu_Y}((\Pi')^\rho) \leq \frac{C + 5\rho}{\ep} + O(\rho \log 1/\ep).
$$
Recalling that $\KA' = \hat \Pi^\rho$, by construction of $\Pi'$ (and $\hat \Pi$) from Theorem \ref{thm:jain_compress}, it follows that
$$
(\KA', \Rpub, (\Pi')^\rho) \dd (X, (\Pi')^\rho, \Rpub) \dd Y
$$
is a Markov chain. % \todo{I think $\KA'$ can be taken to be a deterministic function of $X, (\Pi')^\rho, \Rpub$ (though we don't need this).} 
It then follows from the data processing inequality that 
% Since $\hat \KA'' = (\Pi_1'', \ldots, \Pi_{r'}'', \KA'')$ is a deterministic function of $\hat \Pi_1'', \ldots, \hat \Pi_{r'}'', X, \rho''$ \todo{assume all priv rand is pub: verify!}, it follows by the data processing inequality that
$$
I_{\mu_X \otimes \mu_Y}(\KA', \Rpub, (\Pi')^\rho ; B_1, \ldots, B_n) \leq I_{\mu_X \otimes \mu_Y}(\hat \KA', \Rpub, (\Pi')^\rho; Y) \leq \frac{C + 5\rho}{\ep} + O(\rho \log 1/\ep),
$$
which gives (\ref{eq:ka_bs}); (\ref{eq:kb_as}) follows in a similar manner.
\end{proof}

Roughly speaking, the next lemma, Lemma \ref{lem:test_i}, shows how the protocol $\Pi'$ constructed in Lemma \ref{lem:get_useful_info} can use the properties (2) and (3) of Lemma \ref{lem:get_useful_info} to distinguish between the distributions $\mu$ ($\nu_1$ in the below statement) and $\mu_X \otimes \mu_Y$ ($\nu_2$ in the below statement). This, in combination with the result from Theorem \ref{thm:indist} stating that $\mu$ and $\mu_X \otimes \mu_Y$ are indistinguishable to protocols with little communication, will ultimately complete the proof of Theorem \ref{thm:am_formal}.
\begin{lemma}
  \label{lem:test_i}
  Suppose $\nu_1, \nu_2$ are distributions over tuples of random variables $(Z_1, \ldots, Z_n, I, K,\tilde K)$, where $Z_1, \ldots, Z_n \in \{0,1\}^\ell$, $I \in [n]$, and $K \in \MK$, where $\MK$ is a finite set. Suppose that the marginal distribution of $Z_1, \ldots, Z_n, I$ over each of $\nu_1, \nu_2$ is uniform over $\{0,1\}^{n\ell} \times [n]$. Finally suppose that  $0 < \xi < 1$ and $C$ satisfy $\log n \leq C \leq \frac{(1-\xi)^3 \ell}{1620}$ as well as:
  \begin{enumerate}
  \item $I_{\nu_1}(K; Z_1, \ldots, Z_n) \leq C$.
  \item $I_{\nu_2}(K; Z_I) \geq \ell(1-\xi)$.
  \item $\p_{\nu_2}[K = \tilde K] = 1$, and $\p_{\nu_1}[K = \tilde K] \geq 1 - (1-\xi)^2/36$.
  \end{enumerate}
  Then there is some function $f : \MK \times \{0,1\}^{n\ell} \ra \{0,1\}$ such that
  $$
\left| \E_{\nu_1}[f(\tilde K, Z_1, \ldots, Z_n)] - \E_{\nu_2}[f(\tilde K, Z_1, \ldots, Z_n)] \right| \geq p/2,
$$
where $p = (1-\xi)^2/18$.
\end{lemma}
We first establish some basic lemmas before proving Lemma \ref{lem:test_i}.
\begin{lemma}
  \label{lem:entropy_set}
  Suppose $W \in \{0,1\}^\ell$ is a random variable, and $H(W) = c$. For any $\delta \in (0,1]$ there is some set $\MS \subset \{0,1\}^\ell$ such that $|\MS| \leq 2^{c/\delta}$ and $\p[W \not \in \MS] \leq \delta$.
\end{lemma}
\begin{proof}
  Set
  $$
\MS = \{ w \in \{0,1\}^\ell : \p[W = w] \geq 2^{-c/\delta}\}.
$$
We know that $c = H(W) = \E_{w \sim W}[\log(1/\p[W=w])]$, so the probability that $\p[W=w] < 2^{-c/\delta}$, i.e.~that $\log(1/\p[W=w]) > c/\delta$, over $w \sim W$ is at most $\delta$. Thus $\p[W \not \in \MS] \leq \delta$. Clearly, by the definition of $\MS$, we have that $|\MS| \leq 2^{c/\delta}$.
\end{proof}

\begin{lemma}
  \label{lem:zi_lb}
  Suppose that random variables $I, Z_1, \ldots, Z_n$ are distributed jointly so that the marginal of $Z_1, \ldots, Z_n \in \{0,1\}^\ell$ is uniform on $\{0,1\}^{n\ell}$. Then $H(Z_{I}) \geq \ell - \log n$.
\end{lemma}
\begin{proof}
  % For each $i \in [n]$, since $H(\hat I) \leq \log n$, we have that $H(Z_i | \hat I) \geq L - \log n$. Thus, if $\p[\hat I = i] = p$, then $H(Z_i | \hat I = i) \geq L - \frac 1p \cdot \log n$.
  Notice that
  \begin{eqnarray}
    H(Z_I, Z_{I+1}, \ldots, Z_{I+n-1}) & \geq & H(Z_I, \ldots, Z_{I+n-1} | I)\nonumber\\
                                       &=& \E_{i \sim I}\left[ H(Z_i, \ldots, Z_{i+n-1} | I=i)\right]\nonumber\\
                                       &=& \E_{i \sim I}\left[ H(Z_1, \ldots, Z_n | I=i) \right]\nonumber\\
                                       &=& H(Z_1, \ldots, Z_n | I)\nonumber\\
                                       & \geq & \ell n - \log n,
  \end{eqnarray}
  where addition of subscripts is taken modulo $n$. 
  Since $(Z_{I+1}, \ldots, Z_{I+n-1}) \in \{0,1\}^{\ell n-\ell}$, we get that
  $$
H(Z_I) \geq H(Z_I | Z_{I+1}, \ldots, Z_{I+n-1}) \geq H(Z_I, \ldots, Z_{I+n-1}) - (\ell n-\ell) \geq \ell-\log n,
$$
as desired.
\end{proof}

\begin{lemma}
  \label{lem:hient_smallset}
  Suppose that $W \in \{0,1\}^\ell$ is a random variable with $H(W) = h \leq \ell$. Let $\MS \subset \{0,1\}^\ell$ be a subset with size $|\MS| \leq 2^c$, for some $c < \ell$. Then $\p[W \in \MS] \leq \frac{\ell+1-h}{\ell-c}$.
\end{lemma}
\begin{proof}
  Write $p = \p[W \in \MS]$. Let $J = \One[W \in \MS]$. Then $pc + (1-p)\ell \geq pc + (1-p)\log(2^\ell - 2^c) \geq H(W | J) \geq H(W) - 1 = h-1$. Hence $p(c - \ell) \geq h-1-\ell$, so $p \leq \frac{\ell+1-h}{\ell-c}$. 
\end{proof}

Now we prove  Lemma \ref{lem:test_i}.
\begin{proof}[Proof of Lemma \ref{lem:test_i}]
  % For each $k \in \MK$, we will define a function $g_k : \{0,1\}^L \ra \{0,1\}$, as follows.
  We will first define $f$ and determine a lower bound on $\E_{\nu_2}[f(K, Z_1, \ldots, Z_n)]$. By assumption, $H_{\nu_2}(Z_I) = \ell$, so $H_{\nu_2}(Z_I | K) \leq \xi \ell$. For each $k \in \MK$, let $\gamma_k = H(Z_I | K=k)/\ell$, so that $\E_{k \sim K}[\gamma_k] \leq \xi$. Pick some $\eta > 1, \zeta > 1$ to be specified later. By Lemma \ref{lem:entropy_set}, for each $k \in \MK$, there is a set $\MT_k \subset \{0,1\}^\ell$ of size at most $2^{\eta \gamma_k \ell}$ such that $\p_{\nu_2}[Z_I \not \in \MT_k | K=k] \leq 1/\eta$. Next, set $\MS = \{ k \in \MK : \gamma_k \leq \zeta \xi\}$. By Markov's inequality, $\p_{\nu_2}[K \in \MS] \geq 1 - 1/\zeta$. Thus $\p_{\nu_2}[K \in \MS] \cdot \p_{\nu_2}[Z_I \in T_K | K \in \MS] \geq (1-1/\zeta) \cdot (1-1/\eta)$, and for all $k \in \MS$, $|\MT_k| \leq 2^{\eta \zeta\xi \ell} < 2^{\eta \zeta \ell}$.
  % Thus $\p_{\nu_2}[Z_I \not \in T_K] \leq 1/\eta$.

% For $k \in \MK$, we set $g_k(z)$. 
% Next, for each $k \in \MK, i \in [n]$, we have that $H_{\nu_2}(Z_I| I,K)$.
  We now set
  $$
  f(K, Z_1, \ldots, Z_n) = \begin{cases}
    \bigvee_{i \in [n]} \One[Z_i \in \MT_K] \quad : \quad K \in \MS \\
    0 \quad : \quad \mbox{else}.
    \end{cases}
$$
Since $\p[K \in \MS] \cdot \p[Z_I \in \MT_K | K \in \MS] \leq \E\left[\vee_{i \in[n]} \One[Z_i \in \MT_K]\right]$, $$\E_{\nu_2}[f(K, Z_1, \ldots, Z_n)] \geq (1-1/\eta) \cdot (1-1/\zeta).$$

Next we determine an upper bound on $\E_{\nu_1}[f(K, Z_1, \ldots, Z_n)]$. Define a random variable $\hat I = \hat I(Z_1, \ldots, Z_n, K)$, by $\hat I = \min \{ i \ :\ Z_i \in \MT_K\}$, if the set $\{ i \ : \ Z_i \in \MT_K\}$ is nonempty, else $\hat I = 1$. Thus $H(\hat I) \leq \log n$. Consider the random variable $Z_{\hat I} \in \{0,1\}^\ell$. It follows that $f(K, Z_1, \ldots, Z_n) \leq \One[Z_{\hat I} \in \MT_K]$. By Lemma \ref{lem:cond_ineq_1} and the data processing inequality, we have that
$$
I_{\nu_1}(K; Z_{\hat I}) - \log n \leq I_{\nu_1}(K; Z_{\hat I} | \hat I) \leq I_{\nu_1}(K; Z_1, \ldots, Z_n| \hat I) \leq I_{\nu_1}(K; Z_1, \ldots, Z_n) + \log n \leq C + \log n.
$$
Lemma \ref{lem:zi_lb} gives that $H_{\nu_1}(Z_{\hat I}) \geq \ell - \log n$, so $H_{\nu_1}(Z_{\hat I} | K) \geq \ell - C - 3 \log n$. For each $k \in \MK$, let $h_k = H_{\nu_1}(Z_{\hat I} | K=k)$, so that $\E_{\nu_1}[h_K] \geq \ell - C - 3\log n$. By Lemma \ref{lem:hient_smallset}, for each $k \in \MK$ with $\eta \gamma_k < 1$, $\p[Z_{\hat I} \in \MT_K | K=k] \leq \frac{\ell+1-h_k}{\ell(1 - \eta \gamma_k)}$,
% \frac{L+1-h_k}{(1 - \eta \zeta \xi )L - (1 + \eta \zeta)\log n}$,
by our upper bound $|\MT_k| \leq 2^{\eta \gamma_k \ell}$.

Recall that $\E_{\nu_2}[\gamma_k] \leq \xi$. For $i \in \{1,2\}$, let $K_{\nu_i}$ be the marginal distribution of $K$ according to $\nu_i$. We must have that $\Delta(K_{\nu_2}, K_{\nu_1}) < p$, else we could choose $f$ to be a function of only $K$ and would get that $|\E_{\nu_1}[f] - \E_{\nu_2}[f]| \geq p$. Thus $1 - 1/\zeta -p \leq \p_{\nu_1}[K \in \MS] \leq 1$. Next notice that $\E_{\nu_1}[\ell-h_K] \leq C + 3\log n$, and that $\ell-h_K \geq 0$ with probability 1. Therefore, $\E_{\nu_1}[\ell-h_K | K \in \MS] \leq \frac{C + 3\log n}{1 - 1/\zeta - p}$. Since $\gamma_k \leq \zeta \xi$ for all $k \in \MS$, it follows that
\begin{eqnarray}
  \E_{\nu_1}[f(K, Z_1, \ldots, Z_n)] & \leq & \E_{\nu_1}[f(K, Z_1, \ldots, Z_n) | K \in \MS]\nonumber\\
                                     & \leq & \p_{\nu_1}[Z_{\hat I} \in \MT_K | K \in \MS]\nonumber\\
                                     & \leq & \frac{ 1 + \frac{C + 3\log n}{1-1/\zeta - p}}{\ell(1 - \eta \zeta \xi)}\nonumber.
\end{eqnarray}
Thus
$$
\E_{\nu_2}[f(K, Z_1, \ldots, Z_n)] - \E_{\nu_1}[f(K, Z_1, \ldots, Z_n)] \geq (1-1/\zeta) \cdot \left( (1-1/\eta) - \frac{1}{1-1/\zeta} \cdot  \frac{ 1 + \frac{C + 3\log n}{1-1/\zeta - p}}{\ell(1 - \eta \zeta \xi)}\right).
$$
Now, choose $\eta = \zeta = \xi^{-1/3}$, and let $\xi' = 1-\xi$, so that $p \leq \xi'/6 \leq \frac{1 - (1-\xi')^{1/3}}{2} = \frac{1-1/\zeta}{2}$. Using the inequality $ax \leq 1 - (1-x)^a \leq x$ for $0 < a < 1$, $x \in [0,1]$ and $C \geq \log n$ gives
\begin{eqnarray}
  \E_{\nu_2}[f(K, Z_1, \ldots, Z_n)] - \E_{\nu_1}[f(K, Z_1, \ldots, Z_n)] & \geq & \xi'/3 \cdot \left( \xi'/3 - \frac{1}{(1-1/\zeta)(1 - 1/\zeta - p)} \cdot \frac{15C}{\xi' \ell}\right)\nonumber\\
                                & \geq & \xi'/3 \cdot \left( \xi'/3 - \frac{270 C}{(\xi')^2 \ell} \right)\nonumber\\
  & \geq & (\xi')^2 / 18 = p\nonumber,
\end{eqnarray}
where the last inequality follows from $C \leq \frac{(\xi')^3 \ell}{1620}$.

Since $K = \tilde K$ over $\nu_2$ and are only nonequal with probability at most $p/2$ over $\nu_1$, it follows that
$$
  \E_{\nu_2}[f(\tilde K, Z_1, \ldots, Z_n)] - \E_{\nu_1}[f(\tilde K, Z_1, \ldots, Z_n)] \geq p/2,
  $$
  as desired.
\end{proof}

\subsection{Proof of Theorem \ref{thm:am_formal}}
\label{sec:am_formal_proof}
Using Lemmas \ref{lem:get_useful_info}, \ref{lem:test_i}, and Theorem \ref{thm:achievability_main}, we now may prove Theorem \ref{thm:am_formal}:
\begin{proof}[Proof of Theorem \ref{thm:am_formal}]\
  The first part of Theorem \ref{thm:am_formal} follows in the same way as the amortized case: given $N$ i.i.d.~samples of $(X,Y) \sim \mu_{r,n,\ell}$, by following the pointers for each sample, Alice and Bob can use $r+2$ rounds of communication (simultaneously over all samples), communicate a total of $(r+2) \lceil \log n \rceil $ bits, and generate $N$ i.i.d.~strings uniformly distributed on $\{0,1\}^\ell$. Their resulting keys (of length $N\ell$) will agree with probability 1 and be independent of the transcript of communication.
  
  To prove the second part of Theorem \ref{thm:am_formal}, first suppose $r$ is odd. We take $\mu = \mu_{r,n,\ell}$ and set $\ep = \gamma/  (54 (r+1))$.

    We argue by contradiction. Suppose the theorem statement is false: namely, that for some $C \leq n/\log^{c_0} n$ and $L > \gamma \ell$, the tuple $(C,L)$ is $\lfloor (r+1 )/ 2\rfloor$-achievable from $\mu$. We can assume without loss of generality that $L < \ell$. % By increasing $L$ if necessary, we may assume without loss of generality that $L = \FCamcr_r(C)$. % If the theorem statement is false, then by Definition \ref{def:a_crg_alt}, for each $L > \gamma \ell$, we may choose $N$ sufficiently large so that there is some protocol $\Pi$ with $(r-1)/2$ rounds that achieves the rate $(C,L)$ from the source $\mu^{\otimes N}$ with error $\theta \leq 1/L$, where the keys $\KA, \KB \in \{0,1\}^{N L}$.
    By Theorem \ref{thm:achievability_main} (and in particular, Corollary \ref{cor:det_inf}), since $I_\mu(X;Y) = \ell > L$, there is a $\lfloor (r+1)/2 \rfloor$-round protocol $\Pi$ such that $\ICint_\mu(\Pi) \leq C$ and $\ICext_\mu(\Pi) \geq L$.

    By Lemma \ref{lem:get_useful_info}, there is an $\lfloor (r+1)/2 \rfloor$-round public-coin protocol $\Pi'$ with inputs $(X,Y) \sim \mu$ and communication at most $\frac{C + 3 + 5r/2}{\ep} + O(r \log 1/\ep)$ such that at the end of $\Pi'$ with inputs  $(X,Y) \sim \mu$, Alice and Bob output keys $\KA'= \KB'$, respectively, which satisfy $I_{\mu}(\KB'; B_{\ir_r}) \geq L - (C + 1 +2\log n + 18 \ep (r+1) \ell)$. Moreover, when $(X,Y) \sim \mu_X \otimes \mu_Y$,
    \begin{equation}
    \max \{ I_{\mu_X \otimes \mu_Y}(\KA' ; B_1, \ldots, B_n),  I_{\mu_X \otimes \mu_Y}(\KB' ; A_1, \ldots, A_n) \}\leq \frac{C + 3 + 5r/2}{\ep} + O(r \log 1/\ep).\nonumber
    \end{equation}
    Next, let $\Pi''$ be the protocol where the parties run $\Pi'$, and the last party (suppose it is Alice, for concreteness) to speak in $\Pi'$ sends over a random hash $h(\KA')$ of length $O( \log 1/ \gamma )$, so that for any $\KA' \neq \KB'$, $\p_h[h(\KA') = h(\KB')] \leq \gamma^2 / 648$, and the other party, Bob, outputs a final bit equal to $\One[h(\KA') = h(\KB')]$. For sufficiently large $n$, we have that \begin{equation}
            \label{eq:pipp_cc}
            \CC(\Pi'') \leq \frac{C + 3 + 5r/2}{\ep} + O(r \log 1/\ep) + O(\log 1/\gamma) \leq n/\log^{(c_0 - 1)}n.
            \end{equation}
    \begin{claim}
      \label{clm:am_dist}
      $\Pi''$ distinguishes $\mu$ and $\mu_X \otimes \mu_Y$ with advantage at least $\gamma^2 / 324$.
    \end{claim}
    \begin{proof}
      To prove Claim \ref{clm:am_dist}, we consider two cases.

      The first case is that $\p_{\mu_X \otimes \mu_Y} [\KA' \neq \KB'] \geq \gamma^2 / 324$. In this case, the last bit output by Bob will be 0 with probability at least $\gamma^2 / 648$ when $(X,Y) \sim \mu_X \otimes \mu_Y$. Since $\KA' = \KB'$ with probability 1 when $(X,Y) \sim \mu$, it follows that $\Pi''$ distinguishes between the two distributions with advantage at least $\gamma^2/648$ in this case.

      The second case is that $\p_{\mu_X \otimes \mu_Y} [\KA' \neq \KB'] \leq \gamma^2 / 324$. Here we will use Lemma \ref{lem:test_i}. Since $18 \ep (r+1) \leq \gamma/3$, and since for sufficiently large $n$, $C + 1+ 2\log n \leq \gamma n/3 = \gamma \ell /3$, we see that $I_\mu(\KB'; B_{\ir_r}) \geq \gamma \ell - 2\gamma \ell /3 = \gamma \ell /3$. 

  We apply Lemma \ref{lem:test_i}, with $(Z_1, \ldots, Z_n) = (B_1, \ldots, B_n), \ir = \ir_r, K = \KA', \tilde K = \KB', \nu_1 = \mu_X \otimes \mu_Y, \nu_2 = \mu$, $\xi = 1 - \gamma/3$ and $L = n/\log^{(c_0 - 1)}n$. Here we use that $n/\log^{(c_0 - 1)}n \leq \frac{(\gamma/3)^3 n}{1620}$ for sufficiently large $n$  (depending on $\gamma$), as well as $\p_{\mu_X \otimes \mu_Y} [\KA' \neq \KB'] \leq \gamma^2 / 324 = (1-\xi)^2/  36$. Then Lemma \ref{lem:test_i} gives that Bob can output a bit as a deterministic function of $ \KB', B_1, \ldots, B_n$ (all of which Bob holds at the conclusion of $\Pi'$), that distinguishes $\mu$ and $\mu_X \otimes \mu_Y$ with advantage at least $\gamma^2 /324$.
\end{proof}

By Theorem \ref{thm:indist_bggs} below (which is analogous to Theorem \ref{thm:indist}), with $\ep = \gamma^2 /324$, and as long as $c_0$ is large enough so that the right-hand side of (\ref{eq:pipp_cc}) holds for $n \geq c_0$, and such that $c_0 - 1 \geq \beta$ (where $\beta$ is chosen from Theorem \ref{thm:indist_bggs}, given $\ep = \gamma^2/325$), we arrive at a contradiction.
    \begin{theorem}[\cite{bafna_communication-rounds_2018}, Lemma 4.5]
    \label{thm:indist_bggs}
    For every $\ep > 0$ and odd $r$ there exists $\beta, n_0$ such that for every $n \geq n_0$ and $\ell$, the distributions $\mu = \mu_{r,n,\ell}$ and $\mu_X \otimes \mu_Y$ are $(\ep, (r+3)/2, n/\log^{\beta}n)$-indistinguishable.
  \end{theorem}

For even $r$, we use the distribution $\mu = \mu_{r-1,n,\ell}$. Part (1) of the theorem still holds (in fact, we even have $(r+1)$-achievability). For part (2), the argument above applies, except now the lower bound on round complexity is $\lceil ((r-1) + 1)/2 \rceil = \lceil r/2 \rceil  = r/2$.

Finally, to prove part (3) of Theorem \ref{thm:am_formal}, an argument virtually identical to the one for part (2) applies, except that the protocols $\Pi$ and $\Pi'$ have $r$ rounds, the protocol $\Pi''$ has $r+1$ rounds, and the upper bound in (\ref{eq:pipp_cc}) is $\sqrt{n} / \log^{(c_0 - 1)}n$, which needs to be less than $\frac{(\gamma/3)^3 n}{1620} = \frac{(\gamma / 3)^3 \ell}{1620}$ (which it is, for sufficiently large $n$). In the last step fo the proof, we use Theorem \ref{thm:indist} (instead of Theorem \ref{thm:indist_bggs}), which establishes that $\mu$ and $\mu_X \otimes \mu_Y$ are $(\ep, r+1, \sqrt n / \poly \log n)$-indistinguishable for any constant $\ep > 0$.
\end{proof}

\subsection{Separations in MIMK}
\label{sec:mimk}
In this section we use Theorem \ref{thm:am_formal} to derive separations in the MIMK for the pointer chasing source $\mu_{r,n,\ell}$ (recall Definition \ref{def:mimk}). The below Theorem \ref{thm:mimk_sep} generalizes a result of Tyagi \cite{tyagi2013common}, which established a constant-factor separation in the MIMK for 2-round and 1-round protocols for a certain source.
\begin{theorem}
  \label{thm:mimk_sep}
  For each $r \in \BN$, there is a $c_0$ such that for each $n \geq c_0$, the pointer chasing source $\mu_{r,n,n}$ satisfies:
  \begin{enumerate}
  \item $\FI_{r+2}(X; Y) \leq (r+2) \lceil \log n \rceil$.
  \item $\FI_{\lfloor (r+1)/2 \rfloor}(X; Y) > n / \log^{c_0} n $.
  \item $\FI_r(X;Y) >\sqrt n / \log^{c_0} n$.
  \end{enumerate}
\end{theorem}
\begin{proof}
  Let the constant $c_0$ be that given by Theorem \ref{thm:am_formal} for an arbitrary $\gamma$.
  
  The  first item follows from the definition of $\FI_r(X;Y)$ in Definition \ref{def:mimk}, the fact that $I_{\mu_{r,n,n}}(X;Y) = n$ (Lemma \ref{lem:mu_inf}), and the first item of Theorem \ref{thm:am_formal} stating that the tuple $((r+2) \lceil \log n \rceil, n)$ is $(r+2)$-achievable for SKG from the source $\mu_{r,n,n}$.

  To see the second item, suppose that $\FI_{\lfloor (r+1)/2 \rfloor}(X;Y) \leq n/\log^{c_0} n$. Then the tuple $(n/\log^{c_0} n, \ell)$ is $\lfloor(r+1)/2 \rfloor$-achievable from the source $\mu_{r,n,n}$, contradicting the second item of Theorem \ref{thm:am_formal}.

  Similarly, for the third item, if $\FI_r(X;Y) \leq \sqrt n / \log^{c_0} n$, then the  tuple $(\sqrt n / \log^{c_0} n, \ell)$ would be $r$-achievable from the source $\mu_{r,n,n}$, contradicting the third item of Theorem \ref{thm:am_formal}.
\end{proof}

Notice that the MIMK deals with very large rates of communication; in particular, communication at rates larger than the MIMK is no longer interesting, as, for instance, the entropy rate $L$ for SKG is fixed at $I(X;Y)$. One can ask, on the other hand, whether Theorem \ref{thm:am_formal} allows us to determine a separation in some measure that determines the efficiency of CRG and SKG at very {\it small} rates of communication. Formally, we consider the {\it common random bits per $r$-round interaction bit ($r$-round CBIB)} and the {\it secret key bits per $r$-round interaction bit ($r$-round KBIB)}:
% For a given source $(X,Y) \sim \mu$, the $r$-round CBIB and KBIB can be determined from the achievable rate regions $\MT(X,Y)$ and $\MS(X,Y)$, respectively:
\begin{defn}[\cite{liu2017secret}, Corollary 2\footnote{Typically the $r$-round CBIB and KBIB are introduced in a different way by explicitly considering families of protocols, though this will not be important for our purposes.}]
  \label{def:gamma}
  For a source $(X,Y) \sim \mu$ and $r \in \BN$, define:
  $$
\Gamma_r\crg(X,Y) = \sup \left\{ \frac{L}{C} : (C,L) \in \MT_r(X,Y), C > 0 \right\}
$$
and
$$
\Gamma_r\skg(X,Y) = \sup \left\{ \frac{L}{C} : (C,L) \in \MS_r(X,Y), C > 0 \right\}.
$$
\end{defn}
  Notice that $\Gamma_r\crg(X,Y)$ and $\Gamma_r\skg(X,Y)$ can be infinite, if, for instance, there are functions $f_\A : \MX \ra \{0,1\}$ and $f_\B : \MY \ra \{0,1\}$ such that $\p_\mu[f_\A(X) = f_\B(Y)] = 1$ and $L := H(f_\A(X)) = H(f_\B(Y)) > 0$. In such a case, $(0, L) \in \MT_r(X,Y)$. % Alice and Bob can generate infinitely many bits of entropy with perfect agreement and 0 communication by setting their keys to be $(f_\A(X_1), \ldots, f_\A(X_N)) = (f_\B(Y_1), \ldots, f_\B(Y_N))$, for any $N \in \BN$.
  It is easy to see that whenever $\Gamma_r\skg(X,Y)$ or $\Gamma_r\crg(X,Y)$ is finite, we have $\Gamma_r\crg(X,Y) = 1 + \Gamma_r\skg(X,Y)$.
  
  Intuitively, the $r$-round CBIB (KBIB, respectively) can be roughly interpreted as the maximum number of additional bits of common randomness (secret key, respectively) that Alice and Bob can obtain by communicating an additional bit, where the maximum is over ``all protocols and any communication rate''.

We also remark that it follows from Theorem \ref{thm:achievability_main} and Lemma \ref{lem:fc_concave} that $\Gamma_r\crg(X,Y)$ is the derivative of the function $\FCamcr_r(C)$ at $C = 0$.

Next we would like to derive similar separations for the $r$-round interactive CBIB and KBIB to that in Theorem \ref{thm:mimk_sep} for the $r$-round MIMK. Notice that from the first item of Theorem \ref{thm:am_formal} we have immediately that $\Gamma_{r+2}\crg(X,Y) \geq \frac{n}{(r+2) \lceil \log n \rceil}$. We might hope to use the second and third items of Theorem \ref{thm:am_formal} to derive upper bounds on $\Gamma_{\lfloor (r+1)/2 \rfloor}\crg(X,Y)$ and $\Gamma_r\crg(X,Y)$ that grow as $\log^{c_0} n$ and $\sqrt n \log^{c_0} n$, respectively. However, such upper bounds do not immediately follow from Theorem \ref{thm:am_formal} since Theorem \ref{thm:am_formal} requires a lower bound on $L$ in order to show that certain tuples $(C,L)$ are not achievable. In particular, Theorem \ref{thm:am_formal} leaves open the possibility that tuples such as $(\log n, \sqrt n)$, or even $(2^{-n}, 1)$ are $\lfloor (r+1)/2 \rfloor$-achievable for CRG from $\mu_{r,n,n}$. This limitation of Theorem \ref{thm:am_formal} results from the fact that Lemmas \ref{lem:get_useful_info} and \ref{lem:test_i} give vacuous bounds on the disintuishability of $\mu = \mu_{r,n,n}$ and $\mu_X \otimes \mu_Y$ when the tuple $(C,L)$ is such that $L$ is small compared to $n$. We leave the problem of remedying this issue for future work:
\begin{problem}
  \label{prob:cbib}
  For each $r \in \BN$, show (perhaps using Theorem \ref{thm:am_formal}) that there is a $c_0$, such  that for each $n \geq c_0$, the pointer chasing source $(X,Y) \sim \mu_{r,n,n}$ satisfies:
  \begin{enumerate}
  \item $\Gamma_{\lfloor (r+1)/2 \rfloor}\crg(X,Y) \leq \log^{c_0} n$.
  \item $\Gamma_{r}\crg(X,Y) \leq \sqrt n \log^{c_0} n$.
  \end{enumerate}
  It seems that in fact the even stronger result $\Gamma_{r+1}\crg(X,Y) \leq 1 + o_n(1)$ holds.
\end{problem}
Problem \ref{prob:cbib} seems to be quite difficult; a result that $\Gamma_{r'}\crg(X,Y) < f(n)$, for $(X,Y) \sim \mu_{r,n,n}$, some $r' \in \BN$, and some function $f(n)$ would imply, by concavity of the function $C \mapsto \FCamcr_r(C)$ (Lemma \ref{lem:fc_concave}), that for any $C \geq 1$, the tuple $(C, f(n) \cdot C)$ is not $r'$-achievable for CRG from $\mu_{r,n,n}$. For $r' = \lfloor (r+1)/2 \rfloor$ and $f(n) = \poly \log(n)$, this would imply part (2) of Theorem \ref{thm:am_formal}, and for $r' = r$ and $f(n) = \sqrt{n} \poly \log(n)$, this would imply part (3) of Theorem \ref{thm:am_formal}.
% A proof of the last sentence of Problem \ref{prob:cbib} would imply a corresponding separation between the $(r+2)$-round and $(r+1)$-round strong data processing constants for the source $\mu_{r,n,n}$: while trivially we have $s_{r+2}^*(X,Y) \geq 1 - \tilde O(1/n)$, $\Gamma_{r+1}\crg(X,Y) \leq O(1)$ is equivalent to $s_{r+1}^*(X,Y) \leq 1 - c$ for some constant $c$. In view of Theorem \ref{thm:gamma}, an answer to Problem \ref{prob:cbib} would imply similar types of separations for the concave enelopes $\omega_\rho^\lambda$ as well. % as well as similar types of separations for the concave envelopes $\omega_\rho^\lambda$ (Theorem \ref{thm:gamma}).

\section{Proof of the converse direction of Theorem \ref{thm:achievability_main}}
  \label{sec:converse_proof}
  Recall our definition of
  $$
          \tilFCamcr_r(C) := 
         \begin{cases}
    \sup_{\Pi = (\Pi_1, \ldots, \Pi_r) : \ICint_\mu(\Pi) \leq C} \{\ICext_\mu(\Pi)\} \quad &: \quad C \leq \FI_r(X;Y) \\ 
%     \sup_{\Pi = (\Pi_1, \ldots, \Pi_r) :\ICint_\mu(\Pi)\leq C} \ICext_\mu(\Pi) \quad : \quad C \leq \FI_r(X;Y)\\
    I(X;Y) + C \quad &: \quad C > \FI_r(X;Y),
  \end{cases}
    $$
    % and that of Theorem \ref{thm:achievability_main} states that $\FCamcr_r(C) = \tilFCamcr_r(C)$.
    Our goal in this section is to establish the following:
    \begin{theorem}[Converse direction of Theorem \ref{thm:achievability_main}]
      \label{thm:achievability_converse}
      $\FCamcr_r(C) \leq \tilFCamcr_r(C)$.
    \end{theorem}

    Theorem \ref{thm:achievability_converse} essentially states that any (private-coin) protocol $\Pi$ for CRG can be converted into a (private-coin) protocol whose internal and external information costs are related to the communication and common randomness rates of $\Pi$ in a particular way. We prove Theorem \ref{thm:achievability_converse} by first establishing such a statement for deterministic protocols in Lemma \ref{lem:augment} and Lemma \ref{lem:cr_sim} below. We will then use certain properties of $\tilFCamcr_r(C)$ to ``upgrade'' this statement to apply to randomized protocols.
  \begin{lemma}
    \label{lem:augment}
    Suppose $(X,Y) \sim \mu$ for some source $\mu$, and that the tuple $(C,L)$, for $C, L \in \BR_+$ is achievable by an $r$-round {\it deterministic} protocol (in the sense of Definition \ref{def:a_crg}; that is, all properties of Definition \ref{def:a_crg} hold verbatim, except $\Pi$ is not allowed to use private random coins). Then for any $L' < L, C' > C$, there is some $N_0$ such that for all $N \geq N_0$, there is an $r$-round deterministic protocol $\Pi'$ with inputs $(X^N, Y^N) \sim \mu^{\otimes N}$ such that
    \begin{enumerate}
    \item[(1)] $\ICext_{\mu^{\otimes N}}(\Pi') \geq L'N$.
    \item[(2)] $\ICint_{\mu^{\otimes N}}(\Pi') \leq C'N$.
    \end{enumerate}
  \end{lemma}
  \begin{proof}
    Choose $C''$ with $C' > C'' > C$ and $L''$ with $L' < L'' < L$. By Definition \ref{def:a_crg}, there is some $N_0$ so that for each $N \geq N_0$, there is an $r$-round protocol $\Pi$ taking inputs from $\mu^{\otimes N}$ and producing keys $\KA, \KB$ in some set $\MK_N$ with $|\MK_N| \geq L'N$ so that $\CC(\Pi) = \sum_{t=1}^r |\Pi_t| \leq C''N$ and $\Delta(\KA\KB, KK) \leq \ep_N$ for some
    $$
    \ep_N < \min \left\{\frac{C'  - C'' - 2/N}{L'}, \frac{L'' - L' - 1/N}{L''}\right\}.
    $$
    (Here $K \in \MK_N$ denotes the random variable uniformly distributed on $\MK_N$.) By truncating the keys we may assume without loss of generality that $|\MK_N| \leq 2^{\lceil L' N \rceil}$. It follows from $\Delta(\KA\KB, KK) \leq \ep_N$ that $\p[\KA \neq \KB] \leq \ep_N$. Moreover, using Lemma \ref{lem:reverse_pinsker}, we obtain
    \begin{equation}
      \label{eq:ent_lb}
\min \{ H(\KA), H(\KB) \} \geq \log|\MK_N| - (h(\ep_N) + \ep_N \cdot \log |\MK_N|) \geq (1-\ep_N)  L'' N - 1 \geq L' N,
\end{equation}
where we have used $\ep_N \leq \frac{L'' - L' - 1/N}{L''}$.

Now let $\Pi'$ be the following protocol:
\begin{enumerate}
  \item Alice and Bob first simulate $\Pi$, i.e., they exchange the messages $\Pi_1, \ldots, \Pi_r$.
  \item Then the last person to speak in $\Pi$ outputs their key (i.e., if it is Alice, then she outputs $\KA$ and if it is Bob then he outputs $\KB$).
\end{enumerate}
Suppose for simplicity that $r$ is odd, so that Alice is the last person to speak in $\Pi$ (the case $r$ even is nearly identical). 
Then since $\Pi$ is deterministic, $\KA, \Pi^r$ is a deterministic function of $X, Y$, so $H(\KA, \Pi^r | X, Y) = 0$. Noting the transcript of $\Pi'$ is given by $(\Pi_1, \ldots, \Pi_{r-1}, (\Pi_r, \KA))$, it follows that
  $$
\ICext(\Pi') = I(\KA, \Pi^{r}; X, Y) = H(\KA, \Pi^{r}) - H(\KA, \Pi^{r} | X,Y) \geq H(\KA, \Pi^{r}) \geq H(\KA) \geq L'N,
$$
where the last inequality uses (\ref{eq:ent_lb}).

To upper bound $\ICint(\Pi')$, notice that
\begin{eqnarray}
  \ICint_{\mu^{\otimes N}}(\Pi') &=& I(\Pi^{r}, \KA; X^N | Y^N) + I(\Pi^r, \KA ; Y^N | X^N) \nonumber\\
                                 &=& I(\Pi^{r}; X^N | Y^N) + I(\KA; X^N | \Pi^{r}, Y^N) + I(\Pi^r ; Y^N | X^N) \nonumber\\
                                 &=& \ICint_{\mu^{\otimes N}}(\Pi) + I(\KA; X^N | \Pi^{r}, Y^N) \nonumber\\
                &\leq & \CC(\Pi) + H(\KA | \Pi^{r}, Y^N)\nonumber\\
                                 & \leq & C'' N +  \ep_N \lceil L' N \rceil + 1\nonumber\\
                                 & \leq & C' N \nonumber,
\end{eqnarray}
where we have used Fano's inequality, the fact that $\p[\KA \neq \KB] \leq \ep_N$, and that $\KB$ is a deterministic function of $\Pi^R, Y^N$. Moreover, the last inequality uses $\ep_N < \frac{C' - C'' - 2/N}{L'}$.

  \end{proof}

  The next lemma, which states that the internal and external information complexities {\it tensorize} (i.e., they satisfy a direct sum property), was proved in \cite{ghazi2018resource}.
   \begin{lemma}[\cite{ghazi2018resource}, Lemma 14]
    \label{lem:cr_sim}
Suppose that $\Pi$ is an $r$-round private-coin protocol with inputs $(X^N, Y^N) \sim \nu^{\otimes N}$. Then there is an $r$-round private-coin protocol $\Pi'$ with only private randomness, inputs $(X,Y) \sim \nu$, such that:
    \begin{enumerate}
    \item[(1)] $\ICint_{\nu^{\otimes N}}(\Pi) = N \cdot \ICint_{\nu}(\Pi')$. % , $\ICint_B(\Pi) = N \cdot \ICint_B(\Pi')$.
      % \item[(2)] $N \cdot H(K_A' | \rho')\geq H(K_A)$, $N \cdot H(K_B' | \rho') \geq H(K_B)$.
    \item[(2)] $\ICext_{\nu^{\otimes N}}(\Pi) \leq N \cdot \ICext_{\nu}(\Pi')$. 
%     \item[(3)] The distribution of the transcript (including the keys) under $\Pi'$ is the same as that under $\Pi$. In particular, $\p_{\nu^{\otimes n}} [K_A = K_B] = \p_{\nu} [K_A' = K_B']$.
    % \item[(4)] $I(K_A ; X^N | Y^N, \Pi^r) = n \cdot I(K_A' ; X | Y, (\Pi')^r, \rho')$, $I(K_B; Y^N | X^N, \Pi^r) = n \cdot I(K_B' ; Y | X, (\Pi')^r, \rho')$, where $\Pi^r$ denotes the messages of $\Pi$ and $(\Pi')^r$ denotes the messages of $\Pi'$.
    \end{enumerate}
  \end{lemma}

  Lemmas \ref{lem:augment} and \ref{lem:cr_sim} are sufficient to prove the converse direction of a version of Theorem \ref{thm:achievability_main} for deterministic protocols. % In particular, they establish the direction of point (2) of Theorem \ref{thm:achievability_main} stating that any tuple $(C,L) \in \MT_r(X,Y)$ achievable by a deterministic protocol in fact lies in $\MTd_r(X,Y)$.
  To prove the converse direction for randomized protocols (i.e., Theorem \ref{thm:achievability_converse}), we first need to establish some properties of $\tilFCamcr_r(C)$ in Lemma \ref{lem:fc_concave} below.
  \begin{lemma}
  \label{lem:fc_concave}
  For each fixed $r \in \BN$, $\tilFCamcr_r(\cdot)$ is a nondecreasing concave function on $\BR_{\geq 0}$. In particular, it is continuous, and $\frac{d\tilFCamcr_r(C)}{dC} \geq 1$ for all $C \geq 0$.
\end{lemma}
\begin{proof}
  First we suppose $C' < C \leq \FI_r(X;Y)$. 
  That $\tilFCamcr_r(C)$ is non-decreasing for $C$ in this range is immediate from the definition. To show concavity, we use a simple time-sharing argument. In particular, pick any $L < \tilFCamcr_r(C)$ and $L' < \tilFCamcr_r(C')$, and suppose some $r$-round protocol $\Pi = (\Pi_1, \ldots, \Pi_r)$ has $\ICext_\mu(\Pi) \geq L$, $\ICint_\mu(\Pi) \leq C$, and that some $r$-round protocol $\Pi' = (\Pi_1', \ldots, \Pi_r')$ has $\ICext_\mu(\Pi') \geq L'$ and $\ICint_\mu(\Pi') \leq C'$. For any $0 < \delta < 1$, construct a protocol $\Pi''$ in which Alice, using her private randomness, generates a bit $B$ which is 1 with probability $\delta$, and sends it to Bob as part of the first message. If $B = 0$, Alice and Bob run the protocol $\Pi'$, and if $B = 1$, then Alice and Bob run the protocol $\Pi$. Formally, we write:
  $$
  \Pi_i'' := \begin{cases}
    (B, \Pi_i) \quad : \quad i = 1, B = 1\\
    (B,\Pi_i') \quad : \quad i = 1, B = 0\\
    \Pi_i \quad : \quad i > 1, B = 1\\
    \Pi_i' \quad : \quad i > 1, B = 0.
  \end{cases}
  $$
  Then by linearity of expectation,
  $$
I(\Pi_1'' ; X | Y) = I(B ; X | Y) + I(\Pi_1'' ; X | YB) = I(\Pi_1'' ; X | YB) = \delta \cdot I(\Pi_1 ; X | Y) +  (1-\delta) \cdot I(\Pi_1' ; X |Y).
$$
and
$$
I(\Pi_1'' ; XY) = I(B ; XY) + I(\Pi_1'' ; XY | B) = \delta \cdot I(\Pi_1 ; XY) + (1-\delta) \cdot I(\Pi_1' ; XY).
$$
It follows in an even simpler manner that for all $i$,
$$I(\Pi_i'' ; XY | (\Pi'')^{i-1}) = \delta \cdot I(\Pi_i ; XY | \Pi^{i-1}) + (1-\delta) \cdot I(\Pi_i' ; XY | (\Pi')^{i-1}),$$
that for odd $i \in [r]$,
$$I(\Pi_i'' ; X | Y (\Pi'')^{i-1}) = \delta \cdot I(\Pi_i ; X | Y \Pi^{i-1}) + (1-\delta) \cdot I(\Pi_i' ; X | Y (\Pi')^{i-1}),$$
and for even $i \in [r]$,
$$
I(\Pi_i'' ; Y | X (\Pi'')^{i-1}) = \delta \cdot I(\Pi_i ; Y | X \Pi^{i-1}) + (1-\delta) \cdot I(\Pi_i' ; Y | X (\Pi')^{i-1}).
$$
Thus $\ICint_\mu(\Pi'') = \delta \cdot \ICint_\mu(\Pi) + (1-\delta) \cdot \ICint_\mu(\Pi')$ and $\ICext_\mu(\Pi'') = \delta \cdot \ICext_\mu(\Pi) + (1-\delta) \cdot \ICext_\mu(\Pi')$. In particular,
$$
\tilFCamcr_r(\delta  C + (1-\delta)  C') \geq \tilFCamcr_r(\delta  \ICint_\mu(\Pi) + (1-\delta)  \ICint_\mu(\Pi')) \geq \delta  \ICext_\mu(\Pi) + (1-\delta)  \ICext_\mu(\Pi') \geq \delta L + (1-\delta)L',
$$
and taking $L \ra \tilFCamcr_r(C)$ and $L' \ra \tilFCamcr_r(C')$ gives $\tilFCamcr_r(\delta C + (1-\delta)C') \geq \delta \cdot \tilFCamcr_r(C) + (1-\delta) \cdot \tilFCamcr_r(C')$, establishing that $\tilFCamcr_r(\cdot)$ is convex on $[0, \FI_r(X;Y)]$.

To complete the proof of the lemma it suffices to show that:
\begin{enumerate}[label=(\arabic*)]
\item \label{it:closed} $\tilFCamcr_r(\FI_r(X;Y)) = \FI_r(X;Y) + I(X;Y)$, and
\item \label{it:derivative} the left-sided derivative of $\tilFCamcr_r(C)$ at $C = \FI_r(X;Y)$ with respect to $C$ is at least 1.
\end{enumerate}
For statement \ref{it:closed} above, we need the following lemma:
\begin{lemma}
  \label{lem:trd_closed}
  Fix a source $(X,Y) \sim \mu$. Then the set $\tilde \MT$ of all pairs $(C,L)$ with $C, L \geq 0$, such that there is some $r$-round private-coin protocol $\Pi$ such that $\ICint_\mu(\Pi) \leq C$ and $\ICext_\mu(\Pi) \geq L$ is a closed subset of $\BR^2$.
  % \label{lem:trd_closed} Define $\MTd_r(X,Y) \subset \BR^2$ by $\MTd_r(X,Y) := \{ (C,L) : C \geq 0, L \leq \tilFCamcr_r(C)\}$.\footnote{The reason for the similarity of notation between $\MTd_r(X,Y)$ and $\MT_r(X,Y)$ is as follows: recall (Definition \ref{def:a_crg}) that $\MT_r(X,Y)$ is the set of pairs $(C,L)$ which are $r$-achievable by a protocol with {\it private randomness}. It turns out that $\MTd_r(X,Y)$ is the set of pairs $(C,L)$ which are $r$-achievable by a protocol with {\it no randomness}, i.e., a deterministic protocol. 
 %  } Then $\MTd_r(X,Y)$ is closed.
\end{lemma}
\begin{proof}
  By the support lemma \cite[Lemma 15.4]{csiszar_information_1981}, we can restrict our attention to protocols $\Pi = (\Pi_1, \ldots, \Pi_r)$ such that $\Pi_t$, $1 \leq t \leq r$, falls in a finite set of size $\MU_t$ at most $|\MX| |\MY| \prod_{t'=1}^{t-1} |\MU_{t'}| + 1$. For each odd $t$, the space of all possible $\Pi_t$ is the $|\MX| \cdot \prod_{t'=1}^{t-1} |\MU_{t'}|$-fold product of all probability distributions on $\MU_t$ (as $\Pi_t$ specifies a probability distribution on $\MU_t$ for each possible value of $X\Pi^{t-1}$), which is compact, and in fact homeomorphic to a closed ball in some $\BR^K$. We have an analogous statement for even $t$, and therefore the space of all possible $\Pi$ is compact. Since the functions $\Pi \mapsto \ICint_\mu(\Pi)$ and $\Pi \mapsto \ICext_\mu(\Pi)$ are continuous, it follows that the set of all possible $(\ICint_\mu(\Pi), \ICext_\mu(\Pi)) \in \BR_{\geq 0}^2$, over all $r$-round protocols $\Pi$, is compact (and in particular closed). Thus $\tilde \MT$ is closed as well.
%   This follows from the fact that entropy is a continuous functional on (finitely supported) distributions and the support lemma. \todo{say more -- is this even true? Also make sure we state that $\MT_r(X,Y), \MS_r(X,Y)$ are closed.}
\end{proof}
By Theorem \ref{thm:tyagi_sl}, there is a sequence $(C_i, L_i)$ with $\lim_{i \ra \infty} C_i = \FI_r(X;Y)$ and $\lim\inf_{i \ra \infty} \geq \FI_r(X;Y) + I(X;Y)$ such that $\tilFCamcr_r(C_i) \geq L_i$. It follows by Lemma \ref{lem:trd_closed} that % and the definition of $\FI_r(X;Y)$, we have that $(\FI_r(X;Y) + I(X;Y), \FI_r(X:Y)) \in \MTd_r(X, Y)$, so we must have
$\tilFCamcr_r(\FI_r(X;Y)) \geq \FI_r(X;Y) + I(X;Y)$. To see that $\tilFCamcr_r(\FI_r(X;Y)) \leq \FI_r(X;Y) + I(X;Y)$, we note that for any protocol $\Pi$, $\ICext_\mu(\Pi) \leq \ICint_\mu(\Pi) + I(X;Y)$ by the data processing inequality. This establishes \ref{it:closed}.

To show that \ref{it:derivative} holds, consider any $C < \FI_r(X;Y)$, which implies that $\tilFCamcr_r(C) < C + I(X;Y)$. Let $\Pi = (\Pi_1, \ldots, \Pi_r)$  be any protocol with $\ICint_\mu(\Pi) = C$ and $L := \ICext_\mu(\Pi)$ arbitrarily close to $\tilFCamcr_r(C)$. For any $0 \leq \delta \leq 1$, consider the protocol $\Pi'$ in which Alice uses private randomness to generate a bit $B \in \{0,1\}$ that is 1 with probability $\delta$ and otherwise 0 and sends it to Bob. Then, if $B = 1$, Alice sends Bob $X$ and the protocol terminates (for a total of $1 \leq r$ rounds), and if $B = 0$, Alice and Bob simulate $\Pi$. In a similar manner as above, it is easy to see that
\begin{eqnarray}
  \ICint_\mu(\Pi') &=& (1-\delta)C + \delta \cdot H(X|Y)\nonumber\\
  \ICext_\mu(\Pi') &=& (1-\delta)L + \delta \cdot H(X)\nonumber.
\end{eqnarray}
Since $C < \FI_r(X;Y) \leq H(X|Y)$, there is some $\delta \in (0,1]$, which we denote by $\delta'$, such that $(1-\delta)C + \delta \cdot H(X|Y) = \FI_r(X;Y)$. Then $(1-\delta')L + \delta' \cdot H(X) \leq \tilFCamcr_r(\FI_r(X;Y))$. Then the secant line of the graph of $\tilFCamcr_r(\cdot)$ between the points $C$ and $\FI_r(X;Y)$ has slope at least
$$
\frac{(1-\delta')L + \delta' \cdot H(X) - L}{(1-\delta')C + \delta'\cdot H(X|Y) - C} = \frac{H(X) - L}{H(X|Y) - C} > 1,
$$
where the last inequality follows since $I(X;Y) > L-C$ by assumption that $C < \FI_r(X;Y)$. 
\end{proof}

The case $r=1$ of the next lemma, Lemma \ref{lem:simulate_pri_rand} was proven as part of the proof of Theorem 4.1 in \cite{ahlswede1998common}. It is also stated without proof in \cite{sudan_communication_2019}. 
  The proof will use the following elementary fact: the fact that $\Pi = (\Pi_1, \ldots, \Pi_r)$ is an $r$-round private-coin protocol, is equivalent to the fact that the following Markov conditions hold:
  $$
\Pi_t \dd X \Pi^{t-1} \dd Y, \ \ \mbox{$t$ odd}, \quad\quad X \dd Y \Pi^{t-1} \dd Y, \ \ \mbox{$t$ even},
$$
where a source $(X,Y) \sim \mu$ is fixed and the messages $\Pi_1, \ldots, \Pi_r$ are random variables.
% In fact, the above Markov conditions are {\it equivalent} to the fact that the messages $\Pi_1, \ldots, \Pi_t$ form a private-coin protocol.
\begin{lemma}
  \label{lem:simulate_pri_rand}
  Suppose that $\nu$ is a distribution with samples $(X\QA, Y\QB) \sim \nu$, where $\QA,\QB$ are uniform and independent infinite strings of bits that are independent of $(X,Y)$. Denote the marginal distribution of $(X,Y)$ by $\mu$. 
  Suppose that $\Pi$ is an $r$-round private-coin protocol with inputs $(X\QA,Y\QB) \sim \nu$, and write $\Iint = \ICint_\nu(\Pi), \Iext = \ICext_\nu(\Pi)$. Then there is a non-negative real number $\alpha$ and a protocol $\Pi'$ with inputs $(X,Y) \sim \mu$ such that
  % $\Pi$ may be interpreted as a deterministic protocol with inputs $(X\RA, Y\RB)$. Suppose that the internal and external information costs of $\Pi$ with respect to the inputs $(X\RA, Y \RB)$ are given by $\Iint$ and $\Iext$, respectively. Then there is a non-negative real number $\alpha$ such that % there is a protocol $\Pi'$ with only private randomness and a non-negative real number $\alpha$ such that
  \begin{eqnarray}
    \ICext_\mu(\Pi') & = & \Iext - \alpha \\
    \ICint_\mu(\Pi') & = & \Iint - \alpha.
  \end{eqnarray}
\end{lemma}
% \begin{remark}
%   An alternate approach to deal with private randomness in the original protocol $\Pi$ is to use the result by Braverman and Garg \cite{braverman_public_2014} showing that for any $r$-round protocol $\Pi$ using only private coins and with internal information complexity $\Iint$, then there is a protocol $\Pi'$ with the same number of rounds, only public coins, internal information complexity at most $\Iint + r \log (\Iint / r) + O(1)$, and which simulates the protocol $\Pi$ (i.e., upon completion of $\Pi'$, the parties output a transcript for $\Pi$ distributed exactly the same as the real transcript for $\Pi$). We choose the approach given by Lemma \ref{lem:simulated_pri_rand} since it is conceptually simpler and since the protocol $\Pi'$ from \cite{braverman_public_2014} uses public coins, so this approach does not even work, I think...(the problem is that you are saying nothing about the external IC of the public coin protocol $\Pi'$ with respect to the real inputs $(X,Y)$)
% \end{remark}
\begin{proof}
  % We denote the distribution of $(X\RA, Y \RB)$ by $\nu$, so that the hypotheses of the lemma state that $\ICext_\nu(\Pi) = \Iext$ and $\ICint_\nu(\Pi) = \Iint$.
  The protocol $\Pi'$ proceeds as follows: given inputs $(X,Y) \sim \mu$, Alice uses her private randomness to generate a uniform infinite string $\QA$ independent of $X$ and Bob does the same to generate a uniform infinite string $\QB$. Then certainly the resulting pair $(X\QA, Y\QB)$ are distributed according to $\nu$. Then Alice and Bob simply run the protocol $\Pi$. Notice that the joint distribution of $((\Pi')^r, (\QA X, \QB, Y))$ is identical to the joint distribution of $(\Pi^r, (\QA X, \QB Y))$. % to the transcript $\Pi^r = (\Pi_1, \ldots, \Pi_r)$ (jointly with the pair $(X\QA, Y \QB)$).

  % Since $\Pi_i$ is a deterministic function of $\QA X \Pi^{i-1}$ if $i \in \MO^r$, and it is a deterministic function of $\QB Y \Pi^{i-1}$ if $i \in \ME^r$, the following Markov conditions also hold:
  Let $\MO^r \subset [r]$ denote the odd integers from 1 to $r$, and $\ME^r \subset [r]$ denote the even integers from 1 to $r$. That $\Pi$ is a randomized (private-coin) protocol with inputs $(X\QA, Y\QB)$ means that the following Markov conditions hold:
  \begin{eqnarray}
    \label{eq:mark_ra}
    \Pi_i \dd \QA X \Pi^{i-1} \dd \QB Y & \quad & \forall i \in \MO^r \\
    \label{eq:mark_rb}
    \QA X \dd \QB Y \Pi^{i-1} \dd \Pi_i & \quad & \forall i \in \ME^r.
  \end{eqnarray}

  % Recall that the messages of $\Pi$ are denoted by $\Pi_1, \ldots, \Pi_r$, so the fact that $\Pi$ is a randomized protocol is equivalent to the following Markov conditions \todo{this may not actually be completely obvious, e.g., if $\Pi_i, \Pi_{i+1}$ use the same part of $\QA$...}:
  It follows immediately from (\ref{eq:mark_ra}) and (\ref{eq:mark_rb}) and the fact that $\QA$, $\QB$, and $(X,Y)$ are all independent that the following Markov conditions also hold: % \todo{verify once more...}
  \begin{eqnarray}
    \label{eq:mark_a}
    \Pi_i \dd X \Pi^{i-1} \dd  Y & \quad & \forall i \in \MO^r \\
    \label{eq:mark_b}
     X \dd  Y \Pi^{i-1} \dd \Pi_i & \quad & \forall i \in \ME^r.
\end{eqnarray}
  It follows from (\ref{eq:mark_ra}) and (\ref{eq:mark_rb}) and the chain rule that
  \begin{eqnarray}
    \ICext_\nu(\Pi) &=& \sum_{i \in \MO^r} I(\Pi_i; \QA X \QB Y | \Pi^{i-1}) + \sum_{i \in \ME^r} I(\Pi_i ; \QA X \QB Y | \Pi^{i-1})\nonumber\\
    \label{eq:ext_nu}
    &=& \sum_{i \in \MO^r} I(\Pi_i ; \QA X | \Pi^{i-1}) + \sum_{i \in \ME^r} I(\Pi_i ; \QB Y | \Pi^{i-1}).
  \end{eqnarray}
  In a similar manner, it follows from (\ref{eq:mark_a}) and (\ref{eq:mark_b}) that
  \begin{equation}
    \label{eq:ext_mu}
    \ICext_\mu(\Pi) = \sum_{i \in \MO^r} I(\Pi_i; X | \Pi^{i-1}) + \sum_{i \in \ME^r} I(\Pi_i ; Y | \Pi^{i-1}).
  \end{equation}
  Thus, from (\ref{eq:ext_nu}) and (\ref{eq:ext_mu}),
  \begin{equation}
    \label{eq:ext_numu}
    \ICext_\nu(\Pi) - \ICext_\mu(\Pi) = \sum_{i \in \MO^r} I(\Pi_i ; \QA | \Pi^{i-1} X) + \sum_{i \in \ME^r} I(\Pi_i ; \QB | \Pi^{i-1}Y).
  \end{equation}
  As for internal information cost, from (\ref{eq:mark_ra}) and (\ref{eq:mark_rb}) we have
  \begin{equation}
    \label{eq:int_nu}
    \ICint_\nu(\Pi) = \sum_{i \in \MO^r} I(\Pi_i ; \QA X | \Pi^{i-1}) - I(\Pi_i ; \QB Y | \Pi^{i-1}) + \sum_{i \in \ME^r} I(\Pi_i ; \QB Y | \Pi^{i-1}) - I(\Pi_i ; \QA X | \Pi^{i-1}),
  \end{equation}
  and from (\ref{eq:mark_a}) and (\ref{eq:mark_b}), we have
  \begin{equation}
    \label{eq:int_mu}
    \ICint_\mu(\Pi) = \sum_{i \in \MO^r} I(\Pi_i ;  X | \Pi^{i-1}) - I(\Pi_i ;  Y | \Pi^{i-1}) + \sum_{i \in \ME^r} I(\Pi_i ;  Y | \Pi^{i-1}) - I(\Pi_i ;  X | \Pi^{i-1}),
  \end{equation}
  Thus, from (\ref{eq:int_nu}) and (\ref{eq:int_mu}),
  \begin{eqnarray}
    \label{eq:int_numu}
    \ICint_\nu(\Pi) - \ICint_\mu(\Pi) &=& \sum_{i \in \MO^r} I(\Pi_i; \QA | \Pi^{i-1} X) - I(\Pi_i ; \QB | \Pi^{i-1} Y) \nonumber\\
    && + \sum_{i \in \ME^r} I(\Pi_i ; \QB | \Pi^{i-1} Y) - I(\Pi_i ; \QA | \Pi^{i-1} X)\nonumber.
    \end{eqnarray}

    Next we claim that for all $i \in \MO^r$, $I(\Pi_i ; \QB | Y \Pi^{i-1}) = 0$ and for all $i \in \ME^r$, $I(\Pi_i ; \QA | X \Pi^{i-1}) = 0$. For $i \in \MO^r$, we have
    \begin{eqnarray}
      && I(\Pi_i ; \QB | Y, \Pi^{i-1}) - I(\Pi_i ; \QB | Y, \Pi^{i-1}, X, \QA)\nonumber\\
      &=& H(\QB | Y, \Pi^{i-1}) - H(\QB | Y, \Pi^i) - H(\QB | Y, \Pi^{i-1}, X, \QA) + H(\QB | Y, \Pi^i, X, \QA)\nonumber\\
      &=& I(\QB; X, \QA | Y, \Pi^{i-1}) - I(\QB ; X, \QA | Y, \Pi^i)\nonumber.
    \end{eqnarray}
    Thus
    \begin{eqnarray}
      \label{eq:3_zero}
      I(\Pi_i ; \QB | Y, \Pi^{i-1}) &=& I(\Pi_i ; \QB | Y, \Pi^{i-1}, X, \QA) + I(\QB; X, \QA | Y, \Pi^{i-1}) - I(\QB ; X, \QA | Y, \Pi^i)\\
      & = & 0,
    \end{eqnarray}
    where the first term of (\ref{eq:3_zero}) is 0 by (\ref{eq:mark_ra}), and the second and third terms are 0 since $\QB \perp (X, \QA)$ and by the monotinicity of correlation property of communication protocols (this is the standard fact that if $(X,Y) \sim \mu$ for some distribution $\mu$ on $\MX \times \MY$, then for any $t > 0$, $I_{\mu}(X; Y | \Pi^t) \leq I_\mu(X;Y)$).

    \if 0
   \begin{eqnarray}
    I(\Pi_i ; \QB | Y \Pi^{i-1}) & \leq & I(\Pi_i ; \QB | \Pi^{i-1}) + I(\Pi_i ; Y | \QB \Pi^{i-1})\nonumber\\
                                 & \leq & I(\Pi_i ; \QB | \Pi^{i-1} X \QA) + I(X\QA; \QB | \Pi^{i-1}) \nonumber\\
    \label{eq:4_zero}
                                 && + I(\Pi_i ; Y | \QB \Pi^{i-1} \QA X) + I(\QA X; Y | \QB \Pi^{i-1})\\
    & = & 0\nonumber,
  \end{eqnarray}
  where the first two inequalities follows from Lemma \ref{lem:cond_ineq_1}, the first and third terms in (\ref{eq:4_zero}) are 0 by (\ref{eq:mark_ra}), and the second and fourth terms in (\ref{eq:4_zero}) are 0 by the monotonicity of correlation property of communication protocols. 
  \fi
  It follows in a similar manner that for $i \in \ME^r$, $I(\Pi_i : \QA | Y \Pi^{i-1}) = 0$. 
  Therefore, we obtain from (\ref{eq:ext_numu}) and (\ref{eq:int_numu}) that 
  $$
\alpha = \ICext_\nu(\Pi) - \ICext_\mu(\Pi) = \ICint_\nu(\Pi) - \ICint_\mu(\Pi) = \sum_{i \in \MO^r} I(\Pi_i ; \QA | \Pi^{i-1} X) + \sum_{i \in \ME^r} I(\Pi_i ; \QB | \Pi^{i-1} Y).
  $$
\end{proof}
Now we may prove the converse direction of Theorem \ref{thm:achievability_main}, i.e., Theorem \ref{thm:achievability_converse}.
\begin{proof}[Proof of Theorem \ref{thm:achievability_converse}]
  Fix a source $(X,Y) \sim \mu$ and any $C \geq 0$. By definition of $\FCamcr_r(\cdot)$, for any $L < \FCamcr_r(C)$, we have that $(C,L) \in \MT_r(X,Y)$, i.e., there is a private-coin $r$-round protocol $\Pi$ that achieves the rate $(C,L)$. As in Lemma \ref{lem:simulate_pri_rand}, we interpret $\Pi$ as a deterministic protocol with respect to the tuple $(X\RA, Y\RB)$ (and denote the corresponding joint distribution by $\nu$).

  Then by Lemma \ref{lem:augment}, for any $C' > C$ and $L' < \FCamcr_r(C)$, there is some $N$ such that there is an $r$-round protocol $\Pi$ with inputs $(X^N\RA^N, Y^N\RB^N) \sim \nu^{\otimes N}$ such that $\ICext_{\mu^{\otimes N}}(\Pi) \geq L' N$ and $\ICint_{\mu^{\otimes N}}(\Pi) \leq C' N$. Then by Lemma \ref{lem:cr_sim}, there is an $r$-round private-coin protocol $\Pi'$ for the inputs $(X\RA, Y \RB) \sim \nu$ such that $\ICint_\nu(\Pi') \leq C'$ and $\ICext_\nu(\Pi') \geq L'$. It follows from Lemma \ref{lem:simulate_pri_rand} with $\QA = \RA, \QB = \RB$ that there is an $r$-round private-coin protocol $\Pi''$ for the inputs $(X, Y) \sim \mu$ such that $\ICint_\mu(\Pi'') \leq C' - \alpha$ and $\ICext_\mu(\Pi'') \geq L' - \alpha$, for some $\alpha \geq 0$.

  By definition of $\tilFCamcr_r(\cdot)$, it follows that $\tilFCamcr_r(C' - \alpha) \geq L' - \alpha$. By Lemma \ref{lem:fc_concave}, it follows that $\tilFCamcr_r(C') \geq L'$. %, or that for any $L'' < L$, $(C', L'') \in \MTd_r(X,Y)$.
  By taking $C' \ra C, L' \ra \FCamcr_r(C)$, it follows by continuity of $\tilFCamcr_r(C)$ (Lemma \ref{lem:fc_concave}) that $\tilFCamcr_r(C) \geq L$. Since $L < \FCamcr_r(C)$ is arbitrary, we get $\tilFCamcr_r(C) \geq \FCamcr_r(C)$, as desired.
\end{proof}

\section{Information Theoretic Lemmas}
\label{sec:it_lemmas}
In this section we collect several information theoretic lemmas which are used throughout the paper. % We present proofs for completeness.

%   The data processing inequality states that if $X, Y$ are jointly distributed random variables, and then we compute some randomized function $Z$ of $Y$ (i.e., we ``process $Y$''), then the mutual information between $X$ and $Z$ can be no greater than the mutual information between $X$ and $Y$.
  \begin{proposition}[Data processing inequality]
    If $X \dd Y \dd Z$ is a Markov chain, then $I(X;Z) \leq I(X;Y)$.
    \end{proposition}
    
  \begin{lemma}[\cite{haitner_communication_2018}, Lemma 2.9]
    \label{lem:cond_ineq_1}
    For random variables $X, Y, Z, W$, we have that
    $$
    I(X;W | Y,Z) \geq I(X; Y | W, Z) - I(X; Y | Z) \geq -I(X;W | Z).
    $$
    In particular,
    $$
    H(W) \geq I(X; Y | W, Z) - I(X; Y | Z) \geq -H(W).
    $$
    % $$
%     I(X;Y | Z) + \min\{ H(W | Z,X), H(W | Z,Y) \} \geq I(X; Y | Z, W) \geq I(X; Y|Z) - H(W | Z).
  %   $$
  \end{lemma}
  \begin{proof}
    Using the definition of mutual information, we observe
    \begin{eqnarray}
      && I(X;W | Y,Z) - I(X ; W | Z)\nonumber\\
      &=& H(X | Y,Z) - H(X | W,Y, Z) - H(X | Z) + H(X | W,Z)\nonumber\\
      &=& I(X; Y | W,Z) - I(X; Y |Z) \nonumber.
    \end{eqnarray}
    The claimed equalities hold by non-negativity of the mutual information.
  \end{proof}
\if 0
  \begin{proof}
    For the lower bound, notice that
    $$
H(X | Z) - H(X | Z, W) = I(W; X | Z) \leq H(W | Z).
$$
Then
$$
I(X; Y | Z, W) = H(X | Z, W) - H(X | Y, Z, W) \geq H(X | Z) - H(W |  Z) - H(X | Z, Y) = I(X; Y | Z) - H(W | Z).
$$
For the upper bound, we have that
$$
H(X|Y,Z) - H(X|Y,Z,W) = I(X; W | Y,Z) \leq H(W | Y,Z),
$$
so that
$$
I(X; Y | Z, W) = H(X | Z,W) - H(X | Y, Z, W) \leq H(X|Z) - (H(X | Y, Z) - H(W | Y,Z)) = I(X; Y | Z) + H(W | Y,Z).
$$
A symmetrical argument gives that $I(X;Y | Z, W) \leq I(X;Y | Z) + H(W | X,Z)$, establishing the upper bound.
\end{proof}
\fi

Pinsker's inequality gives an upper bound on total variation distance in terms of the KL divergence between two distributions.
\begin{proposition}[Pinsker's inequality]
  Let $\mu, \nu$ be two distributions supported on a set $\MX$. Then
  $$
\Delta(\mu, \nu) \leq \sqrt{\frac{\KL(\mu || \nu)}{2}}.
  $$
\end{proposition}

The following lemma implies that the entropy functional $H(\cdot)$ is continuous on the set of distributions on a finite $\MX$ set with respect to the topology induced by total variation distance. 
  \begin{lemma}[\cite{ho_interplay_2010}, Theorem 6]
    \label{lem:reverse_pinsker}
    Suppose $X_1, X_2$ are random variables whose distributions are supported on a set $\MX$, and let $\delta = \Delta(X_1, X_2)$. If $0 \leq \delta \leq \frac{|\MX| - 1}{|\MX|}$, then
    $$
|H(X_1) - H(X_2)| \leq h(\delta) + \delta \log(|\MX| - 1).
    $$
  \end{lemma}

  Corollary \ref{cor:rev_cond_pinsk} derives a conditional version of Lemma \ref{lem:reverse_pinsker}.
  \begin{corollary}
    \label{cor:rev_cond_pinsk}
    Suppose that $X_1, X_2$ are random variables whose distributions are supported on a set $\MX$, $Y_1, Y_2$ are random variables whose distributions are supported on a set $\MY$, and that each of the pairs $(X_1,Y_1)$ and $(X_2,Y_2)$ are jointly distributed according to some distributions. Let $\delta = \Delta(X_1Y_1, X_2Y_2)$. Then
    $$
| H(X_1 | Y_1) - H(X_2 | Y_2)| \leq 1 + 6 \delta \log |\MX|.
    $$
  \end{corollary}
  \begin{proof}
    For $x \in \MX, y \in \MY$, write $p_{X_1Y_1}(x,y)$ for the probability of the event $\{ X_1 = x, Y_1 = y\}$, and similarly $p_{X_2Y_2}(x,y), p_{Y_1}(y), p_{Y_2}(y), p_{X_1|Y_1}(x | y), p_{X_2 | Y_2}(x|y)$, and so on. % Write $\supp(Z)$ for the support of a random variable $Z$. 
For any $y$ not in the support of $Y_2$, and any $x \in \MX$, we will write $p_{X_2 | Y_2}(x | y) = H(X_2 | Y_2 = y) = 0$ as a notational convention (and similarly for $p_{X_1 | Y_1}(x | y), H(X_1 | Y_1 = y)$ for $y$ not in the support of $Y_1$). Choose an arbitrary element $x^* \in \MX$, and define a random variable $\tilde X_2$ with support in $\MX$ that is jointly distributed with $Y_1$ as follows. For $y$ in the support of $Y_2$, let $p_{\tilde X_2 | Y_1}(x | y) = p_{X_2 | Y_2}(x | y)$, for $x \in \MX$. For $y$ not in the support of $Y_2$, let $p_{\tilde X_2 | Y_1}(\cdot | y)$ have all its mass on $x^* \in \MX$.

    By the data processing inequality, $\Delta(Y_1, Y_2) \leq \delta$, so
    \begin{eqnarray}
      \Delta(X_1Y_1, X_2Y_2) &=& \frac 12 \sum_{x \in \MX, y \in \MY} | p_{X_1Y_1}(x,y) - p_{X_2Y_2}(x,y) | \nonumber\\
                             &=& \frac 12 \sum_{x \in \MX, y \in \MY} | p_{X_1|Y_1}(x | y) p_{Y_1}(y) - p_{X_2 | Y_2}(x | y) p_{Y_2}(y) |\nonumber\\
                             &=& \frac 12 \sum_{x \in \MX, y \in \MY} | p_{X_1|Y_1}(x | y) p_{Y_1}(y) - p_{\tilde X_2 | Y_1}(x | y) p_{Y_2}(y) |\nonumber\\
      & \geq & \frac 12 \sum_{x \in \MX, y \in \MY}  p_{Y_1}(y)| p_{X_1|Y_1}(x | y)  - p_{\tilde X_2 | Y_1}(x | y)  | - \frac 12 \sum_{x \in \MX, y \in \MY} |p_{Y_1}(y) - p_{Y_2}(y)| \cdot p_{\tilde X_2|Y_1}(x|y)\nonumber\\
                             &\geq & -\delta + \frac 12 \sum_{x \in \MX,y \in \MY} p_{Y_1}(y) \cdot | p_{X_1 | Y_1}(x | y) - p_{\tilde X_2 | Y_1}(x | y)| \nonumber,
    \end{eqnarray}
  %   where the inequality uses H\"{o}lder's inequality. 
    For $y \in \MY$, write $\delta_y = \frac 12 \sum_{x \in \MX} | p_{X_1 | Y_1}(x | y) - p_{\tilde X_2 | Y_1}(x | y)|$, so that the above gives $\E_{y \sim Y_1} [\delta_y] \leq 2\delta$. 
    
    Next, notice that by H\"{o}lder's inequality,
    \begin{eqnarray}
      |H(X_1 | Y_1) - H(X_2 | Y_2)| &=& \left| \E_{y_1 \sim Y_1}[H(X_1 | Y_1 = y_1)] - \E_{y_2 \sim Y_2}[H(X_2 | Y_2 = y_2)]\right|\nonumber\\
                                    &\leq & \left| \sum_{y_1 \in \MY} p_{Y_1}(y_1) \left( H(X_1 | Y_1 = y_1) - H(X_2 | Y_2 = y_1) \right) \right| + 2\delta \cdot \log |\MX|\nonumber\\
                                    &=& \left| \sum_{y_1 \in \MY} p_{Y_1}(y_1) (H(X_1 | Y_1 = y_1) - H(\tilde X_2 | Y_1 = y_1)) \right| + 2\delta \cdot \log |\MX|\nonumber\\
                                    &\leq & \E_{y_1 \sim Y_1} \left[ |H(X_1 | Y_1 = y_1) - H(\tilde X_2 | Y_1 = y_1)|\right] + 2\delta \log |\MX|\nonumber.
    \end{eqnarray}
    For each $y \in \supp(Y_1)$, we have from Lemma \ref{lem:reverse_pinsker} that $|H(X_1 | Y_1 = y) - H(\tilde X_2 | Y_1 = y)| \leq h(\delta_y) + \delta_y \log |\MX|$ as long as $\delta_y \leq \frac{|\MX| - 1}{|\MX|}$, which happens with probability at least $1-4\delta$ by Markov's inequality. Thus,
    \begin{eqnarray}
      |H(X_1 | Y_1) - H(X_2 | Y_2)| & \leq & \E_{y \sim Y_1} \left[ h(\delta_y) + \delta_y \log |\MX| \right] + 4\delta \log |\MX|\nonumber\\
      & \leq & 1 + 6\delta \log | \MX|\nonumber.
    \end{eqnarray}
  \end{proof}

  \section{Acknowledgements}
  We are grateful to Badih Ghazi for useful discussions and to Salil Vadhan for helpful comments on an earlier version of this work. N.G.~would like to thank Venkat Anantharam for an insightful conversation.

\appendix
\section{Non-amortized CRG}
\label{apx:alt_defns}
As opposed to Definition \ref{def:na_crg}, much of the literature on the non-amortized CRG problem \cite{bogdanov2011extracting,canonne2017communication,guruswami2016tight,ghazi2018resource} has used the following definition, which only guarantees that the agreed-upon key is ``close to uniform over a set of size $2^L$'', in the sense that it has min-entropy at least $L$:
  \begin{defn}[Non-amortized CRG (alternate definition to Definition \ref{def:na_crg})]
    \label{def:na_crg_alt}
    For $r, C \in \BN$, and $L, \ep \in \BR_{\geq 0}$, we say that {\it the tuple $(C, L, \ep)$ is $r$-quasi-achievable from the source $\nu$ (for CRG)} if there is some $N \in \BN$ and an $r$-round protocol $\Pi$ with private randomness that takes as input $(X^N,Y^N) \sim \nu^{\otimes N}$, such that at the end of $\Pi$, Alice and Bob output keys $\KA, \KB \in \MK$, given by deterministic functions $\KA = \KA(X^N, \RA, \Pi^r)$, $\KB = \KB(Y^N, \RB, \Pi^r)$, such that:
    \begin{enumerate}
    \item $\CC(\Pi) \leq C$.
      \item $\min\{ H_\infty(\KA), H_\infty(\KB) \} \geq L$.
    \item $\p_{\nu}[\KA = \KB] \geq 1-\ep$.
    \end{enumerate}
  \end{defn}

  % Notice that in Definition \ref{def:na_crg}, we require a lower bound on the min-entropy as opposed to the entropy of the keys $\KA, \KB$.
  It is instructive to consider what would result if we were to change the second item in Definition \ref{def:na_crg_alt} to the requirement that $\min\{ H(\KA), H(\KB) \} \geq L$: for any $L, \ep > 0$ and any source $\mu$, the tuple $(1, L, \ep)$ would be $1$-achievable from the source $\mu$. In other words, under this alternative definition, Alice and Bob would be able to generate arbitrarily large amounts of common randomness with only 1 bit of communication. To see this claim, consider the protocol where Alice uses private randomness to generate a random bit $B \in \{0,1\}$ that is 1 with probability $\ep$, and 0 otherwise. Alice then sends $B$ to Bob. Then the keys, which are elements of $\MK := \{0,1\}^{\lceil L/\ep \rceil}$, are given as follows: if $B = 0$, then Alice and Bob both output the string of all 0s as the key. If $B = 1$, then Alice and Bob each use private randomness to choose a random element of $\MK$, and output their respective elements as $\KA, \KB$, respectively. The probability of agreement is at least $1-\ep$ (as Alice and Bob agree whenever $B = 0$), and the entropy of each of $\KA, \KB$ is at least $\ep \cdot \lceil L/\ep \rceil \geq L$.

    Next we verify the simple fact that Definitions \ref{def:na_crg} and \ref{def:na_crg_alt} are essentially equivalent:
    \begin{proposition}
      The following two statements hold:
      \begin{itemize}
      \item Suppose that $\Pi$ is an $r$-round protocol that achieves the tuple $(C, L, \ep)$ according to Definition \ref{def:na_crg}, for some $r,C,L,\ep$. Then there is an $r$-round protocol $\Pi'$ that quasi-achieves the tuple $(C, L, 3\ep)$ in the sense of Definition \ref{def:na_crg_alt}.% Then $\Pi$ quasi-achieves the tuple $(C,
        \item Suppose that $\Pi$ is an $r$-round protocol that quasi-achieves the tuple $(C, L, \ep)$ according to Definition \ref{def:na_crg_alt}, for some $r,C,L,\ep$. Then for any $\delta > 0$, there is an $r$-round protocol $\Pi'$ that achieves the tuple $(C, \lfloor L - 2 \log 1/\delta \rfloor, \ep + \delta)$ in the sen of Definition \ref{def:na_crg}. %  $(C, \lfloor L - \log 1/\delta \rfloor, \ep + \sqrt{\delta})$
      \end{itemize}
    \end{proposition}
    \begin{proof}
      First suppose that $\Pi$ is an $r$-round protocol achieving the tuple $(C, L, \ep)$ in the sense of Definition \ref{def:na_crg}. Definition \ref{def:na_crg} gives that if $\KA, \KB$ denote the parties' keys from the protocol $\Pi$, and if $K$ denotes a uniformly distributed key on $\MK$, a set of size at least $2^L$, then $\Delta(\KA,K) \leq \ep$ and $\Delta(\KB,K) \leq \ep$. Therefore, there are randomized functions $g_\A : \MK \ra \MK$ and $g_\B : \MK \ra \MK$ such that $g_\A(\KA)$ and $g_\B(\KB)$ are distributed uniformly on $\MK$, and such that $\p[\KA \neq g_\A(\KA)] \leq \ep$ and $\p[\KB \neq g_\B(\KB)] \leq \ep$. By the union bound, it follows that $\p[g_\A(\KA) \neq g_\B(\KB)] \leq 3\ep$. Certainly $H_\infty(g_\A(\KA)) = H_\infty(g_\B(\KB)) = L$. Therefore, the protocol $\Pi'$ in which Alice and Bob run $\Pi$ but then output $g_\A(\KA), g_\B(\KB)$ as their keys, respectively, quasi-achieves the tuple $(C,L,3\ep)$ in the sense of Definition \ref{def:na_crg_alt}. % $\Delta(\KA, g_\A(\KA)) \leq \ep$ and $\Delta(\KB, g_\B(\KB)) \leq \ep$.

      Next suppose that $\Pi$ is an $r$-round protocol that quasi-achieves the tuple $(C,L,\ep)$ in the sense of Definition \ref{def:na_crg_alt}. Letting $\KA, \KB$ be Alice's and Bob's keys at the conclusion of $\Pi$, we have that $\min \{ H_\infty(\KA), H_\infty(\KB) \} \geq L$. We need the below lemma before continuing:
      \begin{lemma}
        \label{lem:compress_minent}
      Suppose $L > 0$ and $0 < \delta < 1$. Suppose a random variable $K$ is distributed on a set $\MK$ so that $H_\infty(K) \geq L$. Let $\MK'$ be a set of size $\lfloor 2^{L - \log 1/\delta} \rfloor = \lfloor \delta 2^L \rfloor$. Then there is a deterministic function $f : \MK \ra \MK'$ such that $H(f(K)) \geq H_\infty(f(K)) \geq (\log |\MK'|) - \delta$.
    \end{lemma}
    \begin{proof}
      Pick some ordering on $\MK$, and for each $k \in \MK$ according to this ordering, set $f(k)$ to be the element in $\MK'$ which has minimal probability mass assigned to it already under the distribution of $f(K)$. After this procedure, let $k'_* \in \MK'$ have maximum probability under the distribution of $f(K)$, and suppose the last $k \in \MK$ for which we set $f(k) = k'$ is denoted $k_*$. It must be the case that $\p[K \in \{ k \in \MK : k \neq k_*, f(k) = k'_*\}] \leq 1/|\MK'|$ since before setting $f(k) = k'_*$ we had that $k'_*$ had minimal probability mass under all $k' \in \MK'$. Since $\p[K = k_*] \leq 2^{-L} \leq \delta/|\MK'|$, it follows that $\p[f(K) = k_*'] \leq (1+\delta)/|\MK'|$, and so $H_\infty(f(K)) \geq (\log |\MK'|) - \log(1+\delta) \geq (\log |\MK'|) - \delta$.
    \end{proof}
    Let $\MK'$ be a set of size $\lfloor 2^{L - \log 1/\delta} \rfloor$, as in Lemma \ref{lem:compress_minent}. Notice that $|\MK'| \geq 2^{\lfloor L - \log 1/\delta \rfloor}$. By Lemma \ref{lem:compress_minent}, there is a deterministic function, $f_\A : \MK\ra \MK'$ such that $H(f_\A(\KA)) \geq | \MK'| - \delta$. By Pinsker's inequality, it follows that if $K'$ denotes the random variable that is uniformly distributed on $\MK'$, then $\Delta(K', f_\A(\KA)) \leq \sqrt{\delta/2}$. In particular, there is a coupling of $K', f_\A(\KA)$ such that $\p[K' \neq f_\A(\KA)] \leq \sqrt{\delta/2}$. Now, the protocol $\Pi'$ proceeds as follows: Alice and Bob first simulate $\Pi$, and then output $f_\A(\KA)$ and $f_\A(\KB)$ as their keys, respectively. Since $\p[\KA \neq \KB] \leq \ep$, we have $\p[f_\A(\KA) \neq f_\A(\KB)] \leq \ep$ and $\p[K' \neq f_\A(\KA)] \leq \sqrt{\delta/2}$, it follows by the union bound that $\p[f_\A(\KA) = f_\A(\KB) = K'] \geq 1 - \ep - \sqrt{\delta/2}$. It follows  that $\Pi'$ achieves the tuple $(C, \lfloor L - \log 1/\delta \rfloor, \sqrt{\delta})$ in the sense of Definition \ref{def:na_crg}; the statement of the proposition then follows by replacing $\delta$ with $\delta^2$.
    \end{proof}

      \bibliographystyle{alpha}
\bibliography{references,thesis_cc_it}

\end{document}